\documentclass[11pt,english]{article}
\usepackage{lmodern}
\usepackage[T1]{fontenc}
\usepackage[latin9]{inputenc}
\usepackage{geometry}
\usepackage{color}
\usepackage{babel}
\usepackage{mathtools}
\usepackage{amsmath}
\usepackage{amsthm}
\usepackage{amssymb}
\usepackage{graphicx}
\usepackage{setspace}
\usepackage[authoryear]{natbib}
\usepackage{subcaption}
\usepackage[unicode=true,
 bookmarks=true,bookmarksnumbered=false,bookmarksopen=false,
 breaklinks=false,pdfborder={0 0 0},pdfborderstyle={},backref=false,colorlinks=true]
 {hyperref}
\hypersetup{pdftitle={Aggregative Efficiency of Bayesian Learning in Networks},
 pdfpagelayout=OneColumn, pdfnewwindow=true, pdfstartview=XYZ, plainpages=false, urlcolor=[rgb]{0.0430 ,0, 0.5}, linkcolor=[rgb]{0.0430 ,0, 0.5}, citecolor=[rgb]{0.0430 ,0, 0.5}, hypertexnames=false}

\makeatletter

\providecommand{\tabularnewline}{\\}

\theoremstyle{plain}
\newtheorem{prop}{\protect\propositionname}
\theoremstyle{definition}
\newtheorem{defn}{\protect\definitionname}
\theoremstyle{plain}
\newtheorem{cor}{\protect\corollaryname}
\theoremstyle{plain}
\newtheorem{thm}{\protect\theoremname}
\theoremstyle{plain}
\newtheorem{lem}{\protect\lemmaname}
\theoremstyle{definition}
\newtheorem{example}{\protect\examplename}
\newtheorem{lyxalgorithm}{\protect\algorithmname}

\usepackage{chngcntr}
\setcitestyle{round}
\usepackage{mathtools}
\usepackage{breakcites}
\usepackage[all]{hypcap}
\usepackage{dcolumn}

\makeatother

\providecommand{\corollaryname}{Corollary}
\providecommand{\definitionname}{Definition}
\providecommand{\lemmaname}{Lemma}
\providecommand{\propositionname}{Proposition}
\providecommand{\examplename}{Example}
\providecommand{\theoremname}{Theorem}
\providecommand{\algorithmname}{Algorithm}

\begin{document}
\title{Aggregative Efficiency of Bayesian Learning in Networks\thanks{We thank Nageeb Ali, Simon Board, Aislinn Bohren, Matt Elliott, Mira
Frick, Drew Fudenberg, Ben Golub, Sanjeev Goyal, Faruk Gul, Rachel Kranton, Luis Henrique Ribeiro Linhares, George Mailath, Suraj
Malladi, Moritz Meyer-ter-Vehn, Pooya Molavi, Juan Ortner, Mallesh Pai, Andy Postlewaite,
Matthew Rabin, Evan Sadler, Philipp Strack, Omer Tamuz, Fernando Vega-Redondo,
Leeat Yariv, and various conference and seminar participants for helpful
comments. Kevin He thanks the California Institute of Technology for
hospitality when some of the work on this paper was completed. We gratefully acknowledge financial support from NSF Grants SES-2214950 and SES-2215256.}}
\author{Krishna Dasaratha\thanks{Boston University. Email: \texttt{\protect\href{mailto:krishnadasaratha@gmail.com}{krishnadasaratha@gmail.com}}}
\and Kevin He\thanks{University of Pennsylvania. Email: \texttt{\protect\href{mailto:hesichao@gmail.com}{hesichao@gmail.com}}}}
\date{October 26, 2026}
\maketitle
\begin{abstract}
{\normalsize{}\thispagestyle{empty}
\setcounter{page}{0}}{\normalsize\par}

When individuals in a social network learn about an unknown state
from private signals and neighbors' actions, the network structure
often causes information loss. We consider rational agents and Gaussian
signals in the canonical sequential social-learning problem and ask
how the network changes the efficiency of signal aggregation. Rational
actions in our model are log-linear functions of observations and
admit a signal-counting interpretation of accuracy. Networks where agents observe multiple neighbors
but not their common predecessors \emph{confound} information, and
even a small amount of confounding can lead to much lower accuracy. In  a class of  networks where agents move in generations and observe the previous generations, we quantify the information loss with an \emph{aggregative
efficiency} index.  Aggregative efficiency is a simple function of network
parameters: increasing in observations and decreasing in confounding.
Later generations contribute little additional information, even when generations are arbitrarily large and agents observe arbitrarily far into the past. 

\bigskip{}

\noindent \textbf{Keywords}: social networks, sequential social learning,
Bayesian learning, confounding
\end{abstract}
\begin{flushleft}
{\small{}\newpage}{\small\par}
\par\end{flushleft}

\interfootnotelinepenalty=10000
\renewcommand*\&{and}

\section{Introduction}

In many settings, people learn from both their own private information and others' past choices. Such learning can be challenging when people do not observe everyone's  actions, but only those of their neighbors in a network. For instance, suppose an agent observes multiple neighbors whose choices have all been influenced by one person's past action. If the agent observed this  person's action, she  would rationally ``anti-imitate'' it to subtract out its duplicate effect and infer  her other neighbors' private information \citep{eyster2014extensive}. When the agent does not observe the shared source, however,  the observation network generates an obstruction to learning which we term \emph{information confounding}: it is impossible to fully incorporate the private information of these
neighbors without over-weighting the private information of their common influence.

% Starting with \citet{banerjee1992simple} and \citet*{bikhchandani1992theory},
% the economic theory literature contains a large body of work on Bayesian
% models of sequential social learning, where privately informed individuals
% move in turn and draw rational inferences from their observations.
% Much of this literature has focused on settings where individuals
% see all predecessors or peers (i.e., the complete observation network).
% Less is known about how learning compares across different networks,
% and the main existing results in this area (beginning with \cite*{acemoglu2011bayesian}) identify  environments  where rational agents will \emph{eventually} learn completely in all networks satisfying a mild condition. The primary open questions concern how various social network structures
% affect the efficiency of signal aggregation (i.e., the rate of learning).\footnote{A survey by \citet{golub_sadler_2016} points out that
% ``a significant gap in our knowledge concerns short-run dynamics
% and rates of learning in these models. {[}...{]} The complexity of
% Bayesian updating in a network makes this difficult, but even limited
% results would offer a valuable contribution to the literature.''}

This paper shows that networks generating information confounding can severely obstruct Bayesian social learning. In standard models of rational social learning, comparing learning across networks is challenging. We work with the canonical sequential
social-learning model, which features a binary state, but make two assumptions to make our analysis
tractable. First,  agents have Gaussian
private signals about this binary state. Second, agents have
sufficiently informative actions so that their behavior fully reveals
their beliefs.\footnote{The simplest example is that agents choose actions equal to their
posterior beliefs given their information. Our analysis also applies
more generally to any decision problem where actions fully communicate beliefs.} This \emph{rich-signals, rich-actions} world removes some other
obstructions to efficient learning and isolates the role of the social network. Our main results apply this framework to classes of networks where agents arrive in generations and observe some or all members of recent generations. In these networks, we show that information confounding can make learning arbitrarily slow.

To formalize these findings, we first establish several general properties
of social learning in our model. The unique rational strategy profile of the social-learning
game has a log-linear form. We characterize the  strategy
profile that solves agents' signal-extraction problems and give a
procedure to efficiently compute every agent's accuracy in any network. In our model, it turns out the 
action of each rational agent is distributed as if she saw some (possibly non-integer)
number of independent private signals. This signal-counting interpretation gives a simple measure of accuracy in the binary-state setting  studied in \cite{banerjee1992simple}, \cite*{bikhchandani1992theory}, and much of the subsequent sequential social-learning literature. We can quantify each agent's learning outcome in any network in units of private signals. These properties allow the numerical computation of learning outcomes in large networks, as we demonstrate on a dataset of collegiate Facebook friendship networks among students at various universities.

Certain network structures aggregate information quite poorly, and we first illustrate this with several examples in finite networks. The leading example considers a network where an agent has many neighbors who in turn share some common predecessors that the agent does not observe. We show that observing any number of neighbors who share just one common predecessor is less informative than observing four independent signals. So even a small amount of confounding can lead to arbitrarily large information loss. The basic finding that some networks lead to substantial information loss also holds in our dataset of Facebook networks. There is considerable variation in information loss across schools, with losses ranging from 25\% to 74\% of available signals, and much of this variation can be explained by a standard network statistic. To study learning outcomes beyond these (analytic and numerical) examples, we define a measure of the rate of social learning on a network which we call \textit{aggregative efficiency}. This measure relies on the signal-counting interpretation of accuracy to determine what fraction of  private signals are incorporated into agents' beliefs and what fraction are effectively lost. For example, aggregative efficiency is $\frac13$ if later agents learn about as well as if they saw $\frac13$ of their predecessors' private signals (and no other information).

Nearly complete information loss is not particular to our finite examples, and our main results show that aggregative efficiency can be arbitrarily small in several classes of \emph{generations
networks}. Agents are arranged into generations of size $K$, and each
agent in generation $t$ observes some subset of her predecessors from generation $t-1$ and earlier. This network structure could correspond to actual generations
in families or countries, or successive cohorts in organizations like
firms or universities. A broad insight is that these networks often cannot
sustain much learning: even if generation sizes are large, later
generations contribute little additional information.

This information loss is easiest to see in the special case of \emph{symmetric} observation
structures between generations: all agents observe the same number of neighbors from the immediate previous generation 
and all pairs of distinct agents in the same generation share the
same number of common neighbors. In these networks, we obtain an exact expression for aggregative efficiency. No matter the size of the generations,
social learning accumulates no more than two signals per
generation asymptotically. Therefore, aggregative efficiency is arbitrarily
close to zero when generations are large. A large number of endogenously correlated observations, such as the actions of all predecessors from the previous generation, can be less informative than a small number of independent signals. This conclusion holds even for networks where one's neighbors have large and almost non-overlapping observation sets, such as when they see many distinct predecessors and each pair of neighbors only shares one predecessor in common. 

Our expression for aggregative efficiency in symmetric generations networks also lets us compare learning across networks. Aggregative efficiency is higher when agents have more observations but lower when pairs of agents have more common neighbors, thus quantifying the trade-offs as we change the
network. For instance, increasing the density of the observation network has two
countervailing effects on learning: it speeds up the per-generation
learning rate by adding more social observations, but also slows it
down by lowering the informational content of each observation through
extra confounding.

A series of results show the basic forces obstructing information aggregation extend beyond symmetric generations networks. First, we show it is not essential that only the previous generation is observed. If agents observe all predecessors in the past $L \geq 1$ generations, the number of signals aggregated per generation asymptotically is bounded above by $4$. Notably, this bound holds for any generation size $K$ and any number of previous generations observed $L$. Second, we show there is substantial information loss beyond our strong notion of symmetry. We analyze a class of random networks with a regular structure (all agents have the same number of observations) and show that in a specific early generation, information confounds prevent agents from aggregating almost all newly available signals. Numerical results suggest these forces are broader still: in a variety of random networks, the number of signals aggregated per generation always drops below $2$ after early generations. (We note that these exercises assume common knowledge of the realized network, and it is an open question whether similar results obtain when agents only know the distribution over networks.) Finally, confounding limits the rate of learning but can also cause almost total information loss in early periods. We show this for networks where agents observe all predecessors from the previous generation (and no other agents).

%We can say more in the special case of \emph{maximal generations networks} where each agent in generation $t$ observes the actions of all predecessors in generation $t-1.$ Aggregative efficiency is worse with larger generation sizes, as illustrated in Figure \ref{fig:max_gen}. We also show that even early generations learn slowly in maximal generations networks. Social learning accumulates no more than three signals per generation starting with the third generation. If everyone in the first generation observes a single additional common ancestor, then the same bound also holds for all generations.

%\begin{figure}
%\noindent \begin{centering}
%\begin{minipage}[c]{0.24\columnwidth}%
%\noindent \begin{center}
%\includegraphics[scale=0.28]{3_layers_K3}
%\par\end{center}%
%\end{minipage}%
%\begin{minipage}[c]{0.38\columnwidth}%
%\noindent \begin{center}
%\includegraphics[scale=0.4]{signal_per_gen}
%\par\end{center}%
%\end{minipage}%
%\begin{minipage}[c]{0.38\columnwidth}%
%\noindent \begin{center}
%\includegraphics[scale=0.4]{ae}
%\par\end{center}%
%\end{minipage}
%\par\end{centering}
%\caption{\textbf{\label{fig:max_gen}Left}: A maximal generations network with a
%generation size of three. An arrow from $i$ to $j$ means $j$ observes
%$i$'s action. \textbf{Middle}: Number of signals aggregated per generation
%asymptotically in maximal generations networks, as a function of generation
%size. \textbf{Right}: Aggregative efficiency in maximal generations
%networks, as a function of generation size.}
%\end{figure}

These results raise the question of whether any networks with cohort structures can sustain efficient learning. We discuss two organizational changes that can improve outcomes. First, we demonstrate that opening up new channels of communication
in an organization, such as starting a mentorship program where seniors
share their private signals with newcomers, can have large benefits
for organizational learning. We show this type of mentorship can essentially fully correct large learning inefficiencies. Second, we show that \emph{information
silos} \textemdash{} partitioning some employees into insular groups
that do not communicate with each other \textemdash{} improve executives'
information aggregation at the expense of workers' learning. Asymmetric networks can improve learning in some parts of organizations by degrading learning in other parts, which may be beneficial if certain agents make particularly consequential decisions.

Our main results characterize the aggregative efficiency of learning in generations networks, and our final result shows these characterizations have implications for welfare. If signals are very precise or the planner is very impatient, welfare will depend primarily on the learning outcomes of the first few agents. We relate aggregative efficiency to welfare outside of these cases: networks with higher aggregative efficiency lead to higher welfare when signals are sufficiently imprecise and the social welfare function is sufficiently patient. We also give an example showing the arbitrarily large information loss in generations networks can have arbitrarily large welfare consequences.

\section{\label{sec:Model}Model}

There are two equally likely states of the world, $\omega\in\{0,1\}$.
An infinite sequence of agents indexed by $i\in\mathbf{\mathbb{N}_{+}}$
move in order, each acting once. (In some examples, we work with a finite subset $\{1,\hdots,n\}$ of this infinite sequence.) On her turn, agent $i$ observes
a\emph{ private signal }$s_{i}\in\mathbb{R}$ and the actions of her
\emph{neighbors}, $N(i)\subseteq\{1,...i-1\}.$ Agent $i$ then chooses
an \emph{action} $a_{i}\in(0,1)$ to maximize the expectation of $u_{i}(a_{i},\omega)\coloneqq-(a_{i}-\omega)^{2}$
given all of her information. So her action is equal to her belief about the probability that $\omega=1$.\footnote{The quadratic-loss form of the utility functions is not crucial, and
our results on learning remain unchanged if actions are ``rich''
enough to fully reflect beliefs (see \citet{ali2018role} for details).}

We consider a Gaussian information structure where private signals
$(s_{i})$ are conditionally i.i.d. given the state. We have $s_{i}\sim\mathcal{N}(1,\sigma^{2})$
when $\omega=1$ and $s_{i}\sim\mathcal{N}(-1,\sigma^{2})$ when $\omega=0,$
where $\mathcal{N}(a,b^{2})$ is the Gaussian distribution with mean
$a$ and variance $b^{2},$ and $0<1/\sigma^{2}<\infty$ is the private-signal precision.

Neighborhoods of different agents define a deterministic network $M$,
where there is a directed link $j \rightarrow i$ if and only if $j\in N(i)$.
We also view $M$ as the adjacency matrix, with $M_{i,j}=1$ if $j\in N(i)$
and $M_{i,j}=0$ otherwise. Since $N(i)\subseteq\{1,...,i-1\},$ $M$
is lower triangular. The network $M$ is common knowledge. The goal
of this paper is to explore to what extent network structures can hinder efficient information aggregation.

With the network $M$ fixed, let $d_{i}:=|N(i)|$ denote the number
of $i$'s neighbors. A \emph{strategy} for agent $i$ is a function
$A_{i}:(0,1)^{N(i)}\times\mathbb{R}\to(0,1)$, where $A_{i}(a_{j(1)},...,a_{j(d_{i})},s_{i})$
specifies $i$'s play after observing actions $a_{j(1)},...,a_{j(d_{i})}$
from neighbors\footnote{We use $j(k)$ to indicate the $k$-th neighbor of $i$ and suppress
the dependence of $j$ on $i$ when no confusion arises.} $N(i)=\{j(1),...,j(d_{i})\}$ and when her own private signal is $s_{i}$.\footnote{It is without loss  to focus on pure strategies,
since every belief about the state induces a unique optimal action
for each agent.} There is a unique 
 strategy profile $(A_{i}^{*})_{i\in\mathbb{N}_{+}}$ consistent with common knowledge of rationality at the interim stage: for all
$i$ and for all observations of $i,$ $A_{i}^{*}$ maximizes 
Bayesian expected utility given the posterior belief about $\omega.$\footnote{We will see that in the rational strategy profile, $s_{i}\mapsto A_{i}^{*}(a_{j(1)},...,a_{j(d_{i})},s_{i})$
is a surjective function onto $(0,1)$ for all $i$ and $a_{j(1)},...,a_{j(d_{i})}$.
So all observations are on-path and the
posterior beliefs are well-defined.} Uniqueness of this profile follows from the sequential nature of the social-learning game and the existence
of a unique optimal action at any belief. Agent 1 has no social observations, so there is a unique rational strategy
$A_{1}^{*}(s_{1})$. Given unique rational best responses for agents $1,\hdots,i-1$, there is a unique rational best response $A_{i}^{*}$, as we have fixed the behavior of the preceding agents. Proceeding by induction, there is
a unique  strategy profile $(A_{i}^{*})_{i\in\mathbb{N}_{+}}$ consistent with common knowledge of rationality at the interim stage, which we abbreviate as ``rational.''

\section{\label{sec:Linearity-of-Equilibrium}Basic Results about Beliefs
and Learning}

In this section, we show that rational actions are log-linear and satisfy a signal-counting interpretation. We then use these properties to demonstrate information confounding in several examples. The final subsection gives a condition for long-run learning and defines an asymptotic measure of how efficiently information is aggregated.

\subsection{Optimal Actions Are Log-Linear}

As is common in analyzing social-learning problems, we will find it
convenient to work with the following log-transformations of variables:
$\lambda_{i}\coloneqq\ln\left(\frac{\mathbb{P}[\omega=1|s_{i}]}{\mathbb{P}[\omega=0|s_{i}]}\right)$,
$\ell_{i}\coloneqq\ln\left(\frac{a_{i}}{1-a_{i}}\right)$. We call
$\lambda_{i}$ the \emph{log-signal} of $i$ and $\ell_{i}$ the \emph{log-action
}of $i.$ These changes are bijective, so it is without loss to use
the log versions. Write ${S}_{i}^{*}(\ell_{j(1)},...,\ell_{j(d_{i})},\lambda_{i})$
for $i$'s rational log-strategy: the (unique) rational map
from $i$'s neighbors' log-actions  and $i$'s  log-signal
to $i$'s log-action. In this section, we show that every ${S}_{i}^{*}$
is a linear function of its arguments, with coefficients that only
depend on the network $M$ and not on the precision of private signals.

The following result shows the optimal aggregation is linear in log-signals
and log-actions (\emph{log-linear}) and gives an explicit expression
for the coefficients. All proofs are in the Appendix.
\begin{prop}
\label{prop:linear}For each agent $i$ with $N(i)=\{j(1),...,j(d_{i})\},$
there exist constants $(\beta_{i,j(k)})_{k=1}^{d_{i}}$ so that $i$'s
rational log-strategy is given by
\[
{S}_{i}^{*}(\ell_{j(1)},...,\ell_{j(d_{i})},\lambda_{i})=\lambda_{i}+\sum_{k=1}^{d_{i}}\beta_{i,j(k)}\ell_{j(k)}.
\]
The vector of coefficients $\vec{\beta}_{i}$ is given by 
\[
\vec{\beta}_{i}=2\left(\mathbb{E}[(\ell_{j(1)},...,\ell_{j(d_{i})})\mid\omega=1]\times\textsc{\emph{Cov}}[\ell_{j(1)},...,\ell_{j(d_{i})}\mid\omega=1]^{-1}\right),
\]
 where $\textsc{\emph{Cov}}[\ell_{j(1)},...,\ell_{j(d_{i})}\mid\omega=1]^{-1}$
is the inverse of the conditional covariance matrix. These coefficients do not depend on the signal variance $\sigma^{2}.$
\end{prop}

For general private-signal distributions, models of Bayesian updating
in networks have tractability issues, as \citet{golub_sadler_2016}
point out. The key lemma to proving Proposition \ref{prop:linear}
shows that given our Gaussian information structure, agent $i$'s
observations have a jointly Gaussian distribution conditional on $\omega$.
This permits us to study optimal inference in closed form. The interpretation
of the inverse covariance matrix that appears in the coefficients
$\vec{\beta}_{i}$ is that $i$ rationally discounts the actions of
two neighbors $j(1)$ and $j(2)$ if their actions are conditionally
correlated.

The proposition (along with the measure of accuracy in the following subsection) can be used to efficiently compute action distributions numerically. We describe an algorithm in detail in Appendix~\ref{a:computation}, and note this algorithm has polynomial runtime in any network but is most efficient in sparse networks. We implement our algorithm in a dataset of collegiate Facebook friendship networks, where some of the networks have more than 20,000 nodes.

\subsection{Measure of Accuracy}

We would like to evaluate networks in terms of their social-learning
accuracy so that we can compare the rates of Bayesian learning in different
networks. Towards a measure of accuracy, imagine that agent $i$'s
only information about $\omega$ consists of $n\in\mathbb{N}_{+}$
conditionally independent private signals. Then, the Bayesian $i$
would play the log-action equal to the sum of the $n$ log-signals,
and it turns out (by Lemma \ref{lem:trick} in the Appendix) her behavior
would follow the conditional distributions $\ell_{i}\sim\mathcal{N}\left(\pm n\cdot\frac{2}{\sigma^{2}},n\cdot\frac{4}{\sigma^{2}}\right)$,
with the positive and negative means in states $\omega=1$ and $\omega=0$
respectively. We quantify learning accuracy using distributions of
this form that allow for non-integer $n$, thus denominating accuracy
in the units of private signals.
\begin{defn}
\label{def:Social-learning-aggregates}Social learning \emph{aggregates
$r\in\mathbb{R}_{+}$ signals by agent $i$} if the rational log-action
$\ell_{i}$ has the conditional distributions $\mathcal{N}\left(\pm r\cdot\frac{2}{\sigma^{2}},r\cdot\frac{4}{\sigma^{2}}\right)$
in the two states. If this holds for some $r\in\mathbb{R}_{+}$, then
we say $i$'s behavior has a \emph{signal-counting interpretation}.
\end{defn}
Having a signal-counting interpretation imposes a meaningful restriction on $i$'s action profile, even if $i$'s log-action is a linear function of her log-signal and neighbors' log-actions. Indeed, if an action profile results in $i$ putting certain 
weights $(w_{i,j})_{j\le i}$ on log-signals $(\lambda_{j})_{j\le i},$
then $\ell_{i}$ has a signal-counting interpretation if and only
if $\sum_{j=1}^{i}w_{i,j}=\sum_{j=1}^{i}w_{i,j}^{2}.$ But as the next result shows, the rational log-actions always admit
a signal-counting interpretation in any network.
\begin{prop}
\label{prop:signal_counting}There exist $(r_{i})_{i\ge1}$ so that
social learning aggregates $r_{i}$ signals by agent $i.$ These $(r_{i})_{i\ge1}$
depend on the network $M,$ but not on private-signal precision.
\end{prop}
The signal-counting interpretation gives a way to compare agents' accuracy and welfare across different networks or positions in a given network in a binary-state setting. Rather than comparing the full distributions of beliefs, we can compare the summary statistics $r_i$. A consequence is that agents' beliefs, which \emph{a priori} are multi-dimensional objects, are in fact ranked in the standard Blackwell ordering: a higher value of $r_i$ implies a weakly higher expected utility for any decision problem.

Such comparisons are straightforward in a different framework with a Gaussian state and Gaussian signals (e.g., \cite{morris2002social} and, in the context of social learning, \cite*{dasaratha2020learning}). In these models, Bayesian agents' beliefs are ranked by their precisions and the analogous number of signals aggregated is  simply proportional to precision. We use Proposition~\ref{prop:signal_counting} to study the binary-state model used in most of the sequential social-learning literature (beginning with \cite{banerjee1992simple} and \cite*{bikhchandani1992theory}). But we expect information confounding to have similar effects, including severely obstructing learning in some networks, with a Gaussian state.\footnote{Indeed, one can show the example networks in Section~\ref{sec:Examples} give the same results in either model.}

\subsection{\label{sec:Examples}Examples of Information Confounding in Networks}

We say the network causes \emph{information confounding} if there are multiple paths between $i$ and $j$ in the network but the later agent $j$ does not observe $i$ directly. We begin with several examples of social learning in finite networks that cause information confounding. These examples illustrate how information confounding obstructs social learning, highlight the extent of possible information
loss due to confounding, and provide intuition for our main results
on generations networks.

\begin{example}[The Shield and the Diamond]
\label{ex:shield_diamond}

Consider two network structures with four agents, as shown in Figure
\ref{fig:diamond_shield}. In a \emph{shield network}, agent $4$ observes agents $1,2,$ and $3$ while agents $2$ and $3$ observe agent $1$. In a \emph{diamond network}, agent $4$ observes
agents $2$ and $3$ while agents $2$ and $3$ observe agent $1$.\footnote{Our terminology follows \citet{eyster2014extensive}, who focus on
rational learning in networks without diamonds.}

\begin{figure}
\begin{centering}
\includegraphics[scale=0.25]{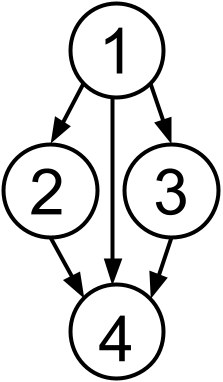}\qquad{}\includegraphics[scale=0.25]{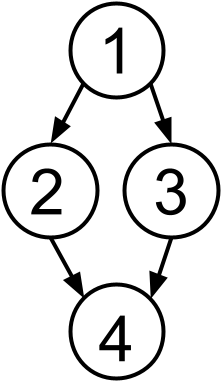}
\par\end{centering}
\caption{\label{fig:diamond_shield}\textbf{Left}: The shield network with
four agents. \textbf{Right}: The diamond network with four agents.}
\end{figure}

In a shield network, agent $4$ observes
all predecessors and can compute the private signals of all agents. To see this, note that $\ell_1 = \lambda_1$, $\ell_2=\lambda_1+\lambda_2$, and $\ell_3 = \lambda_1+\lambda_3$. So $$\ell_4=\lambda_4+\ell_2+\ell_3-\ell_1 = \sum_{j=1}^4 \lambda_j$$
is the optimal action given her private signal and those of her three predecessors, and $r_4=4$. As in \cite{eyster2014extensive}, the optimal action involves anti-imitating, or placing a negative weight on, agent $1$'s action.

In a diamond network, however, agent $4$ observes the actions of agents $2$
and $3$ that combine their private signals with agent $1$'s signal,
which agent $4$ does not observe. Agent $4$ faces an unavoidable
tradeoff between overweighting agent $1$'s signal and underweighting
agents $2$ and $3$'s signals. As in the shield network, we have $\ell_1 = \lambda_1$, $\ell_2=\lambda_1+\lambda_2$, and $\ell_3 = \lambda_1+\lambda_3$. Using Proposition~\ref{prop:linear}, we
can calculate that now $$\ell_4 = \lambda_4 + \frac{2}{3}\ell_2 + \frac{2}{3}\ell_3 = \frac{4}{3} \lambda_1 + \frac{2}{3}\lambda_2+ \frac{2}{3}\lambda_3 + \lambda_4,$$and therefore $r_{4}=\frac{11}{3}<4$. Even though agent $4$
is Bayesian and optimally adjusts for the confounding signal that
she does not observe, some information is lost. This information loss
is not too severe here, but the next example shows it can be much
worse with more agents.

\end{example}

\begin{example}\label{ex:Many-Neighbors-with}

To see that confounding can lead to more severe information loss, we next generalize the diamond network to allow more agents. Consider a network with agents in three generations, shown in Figure
\ref{fig:A-three-generation-network}.\footnote{In Section \ref{sec:The-Generations-Network} we will study generations
of equal size.} In the first generation, agents $1,2,\hdots,K_{1}$ have no neighbors.
In the second generation, agents $K_{1}+1,K_{1}+2,\hdots,K_{1}+K_{2}$
observe all agents in the first generation. Finally, the third generation
consists of a single agent who observes all agents in the second generation
but does not observe the first generation. The purpose of this example
is to study the beliefs of an agent with many neighbors who all observe
a common confound.\footnote{We could equivalently relax the assumption of i.i.d.
signals and replace the first generation with a single agent with
a (potentially) more precise signal.}

\begin{figure}
\begin{centering}
\includegraphics[scale=0.6]{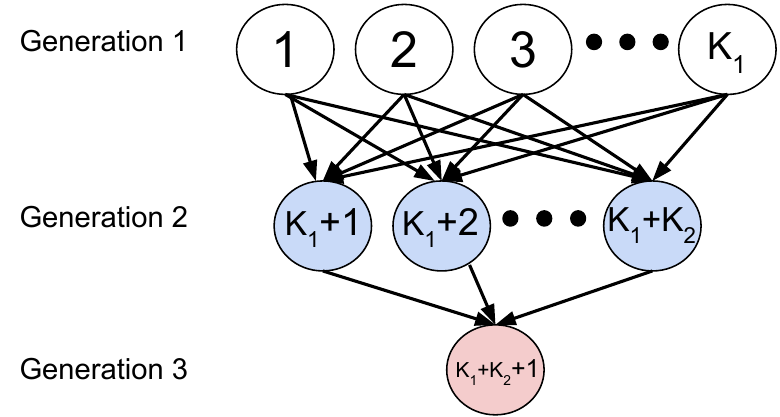}
\par\end{centering}
\caption{\label{fig:A-three-generation-network}A three-generation network
with $K_{1}$ agents in generation 1, $K_{2}$ agents in generation
2, and one agent in generation 3.}
\end{figure}

The agent in generation 3 rationally calculates the log-likelihood
of state $\omega=1$ by taking a weighted sum of the log-actions of generation $2$ agents and  her own signal, where the weights depend on $K_{1}$ and $K_{2}$. As in Example \ref{ex:shield_diamond}, the final agent faces an unavoidable tradeoff between overweighting generation $1$'s private signals and underweighting generation $2$'s private signals. 
Using Proposition~\ref{prop:linear}, we
can compute that the optimal action places weight $\frac{1+K_{1}}{1+K_{1}K_{2}}$ on each  generation 2 action
 (see Appendix \ref{subsec:details_ex} for details of the calculations in this example). When generation $2$ is large, this
weight is close to $1/K_{2}$: it is optimal for the final agent to severely underweight
the private signals of generation $2$.

We can also show that the
actions of the agents in generation $2$ are distributed as if they
see $1+K_{1}$ conditionally independent private signals, while the
action of the final agent is distributed as if she sees $1+\frac{K_{2}+K_{1}K_{2}}{1+K_{1}K_{2}}\cdot(1+K_{1})$
such signals. The difference between the accuracy of generation $2$
and $3$'s actions is just

\[
1+\frac{K_{2}+K_{1}K_{2}}{1+K_{1}K_{2}}\cdot(1+K_{1})-(1+K_{1})=1+\frac{(K_{2}-1)(K_{1}+1)}{K_{1}K_{2}+1}<3
\]
private signals, for any values of $K_{1}$ and $K_{2}$. So there
is always very little learning between generations $2$ and $3$,
even when the size of generation $2$ is large and many private signals
arrive in that generation. The idea is that confounding significantly
limits how much information the final agent can extract from arbitrarily
many neighbors' actions.

We emphasize that the first generation can generate substantial confounding
even when it is small. For example, if there is a single agent in
the first generation ($K_{1}=1$), then the action of the agent in generation $3$
will be less accurate than that of someone who saw just 5 independent
private signals. But if generation $1$ were empty, then the action
of the generation 3 agent  would be equivalent to $K_{2}+1$ private
signals. So even a small confound can prevent almost all information
aggregation. Also, a simple calculation shows that the difference
between the accuracy of generation 2 and generation 3 strictly decreases
in $K_{1}$ and strictly increases in $K_{2},$ provided $K_{2}\ge2.$
That is, the incremental amount of learning in the final generation
decreases with confounding (a larger $K_{1})$ but increases with
the number of observations (a larger $K_{2}).$ We will later see
that the same comparative statics hold for a class of infinite networks
where agents move in generations.

\end{example}

\subsection{Long-Run Learning and Quantifying Information Loss}

We now return to studying infinite networks. We begin with a benchmark result providing necessary and sufficient conditions for long-run
learning. These conditions turn out to be similar to those in the
existing literature, which shows our model is comparable to the standard
models on this dimension. A key contribution of our model is ranking networks where agents learn in the long run based on the rate of this learning,  and this section concludes by explaining how we can use $(r_i)_{i \in \mathbb{N}_+}$ for this purpose.

\begin{defn}
    Society \emph{learns completely in the long run} if $(a_{i})_{i \in \mathbb{N}_+}$
converges to $\omega$ in probability.
\end{defn}
For a given network $M,$ write
$\mathcal{PL}(i)\in\mathbb{N}$ to refer to the length of the longest
path in $M$ terminating at $i$ (this length is 0 if $N(i)=\varnothing$).
\begin{prop}
\label{prop:long_run_conditions} The following are equivalent: (1)
$\lim\limits _{{\it i\to\infty}}\mathcal{PL}(i)=\infty$; (2) $\lim\limits _{{\it i\to\infty}}\left[\max_{j\in N(i)}j\right]=\infty$;
(3) $\lim_{i\to\infty}r_{i}=\infty$; (4) society learns completely
in the long run.
\end{prop}
Condition (2) is the analog of \citet*{acemoglu2011bayesian}'s \emph{expanding
observations} property for a deterministic network. It says if we
consider the most recent neighbor observed by each agent, then this
sequence of most recent neighbors tends to infinity. It is straightforward to see the expanding observations condition is necessary for long-run learning, and \citet{acemoglu2011bayesian}
show it is also sufficient in a random-networks model with unboundedly informative signals
and binary actions. With continuous actions, Proposition \ref{prop:long_run_conditions} states the same result. The intuition is that each agent learns at least as well as if she optimally combined her most accurate social observation with her private signal.

The key takeaway message from Proposition \ref{prop:long_run_conditions}
is that whether society learns in the long run is not a useful criterion
for comparing different networks in this setting, as the conditions
(1) and (2) that guarantee long-run learning are quite mild. It is
of course possible that agents learn completely but do so very slowly, and we can compare how well agents learn on different networks by comparing the $(r_{i})_{i \in \mathbb{N}_+}$ sequences  associated with these networks. It is easy to see that a network achieves zero information loss ($r_i = i$ for every agent $i$) if and only if $i-1 \in N(i)$ for every $i$. An agent $i$ who does not observe $i-1$ cannot access $i-1$'s private information, so $r_i \le i-1$. Conversely, an inductive argument shows that any network with $i-1 \in N(i)$ for all $i \le n$ ensures that if $n \in N(n+1)$, then  agent $n+1$ can perfectly aggregate all signals up to $n+1$ by combining agent $n$'s action with the signal $s_{n+1}$. 

On networks supporting long-run learning with positive information loss, it is useful to have a measure of the efficiency of learning. Our main measure normalizes the number of signals aggregated by each agent $i$ by the total number of signals that have arrived:
\begin{defn}
If $\lim_{i\to\infty}(r_{i}/i)$ exists, it is called the\emph{ aggregative
efficiency} of the network.
\end{defn}
 Aggregative efficiency measures the fraction of signals in the entire society that individuals manage to aggregate under social learning. Networks with higher levels of aggregative efficiency induce faster social learning in the long run.

In Section \ref{sec:The-Generations-Network},  we study generations networks where agents move in cohorts and observe some members of the previous cohorts. The limit defining aggregative efficiency exists for all the networks we focus on in this next section, and we will use aggregative efficiency to compare learning across these networks.\footnote{In infinite networks where $\lim_{i\to\infty}(r_{i}/i)$ does not exist because there are multiple subgroups of agents and $\lim_{i\to\infty}(r_{i}/i)$ is different in each subgroup (such as the information silos network in Section \ref{subsec:Application-2:-Information}), we can compute this limit separately for each subgroup. More generally, $\liminf_{i\to\infty}(r_{i}/i)$ and $\limsup_{i\to\infty}(r_{i}/i)$ exist in every infinite network and provide long-run bounds on the rate of signal aggregation.} In general, comparisons of aggregative efficiency also translate into welfare comparisons, as Section \ref{sec:Welfare} will show.

Appendix \ref{subsec:facebook} calculates a finite analog of aggregative efficiency for Facebook friendship networks from $100$ college campuses. We find the fraction of signals aggregated ranges from 26\% to 75\% across campuses and that much of this variation can be explained based on the amount of clustering in the network. These simulations suggest that aggregative efficiency can capture meaningful differences across real-world networks and that the corresponding losses can be substantial.

While we focus on aggregative efficiency, we note that other normalizations of $r_i$ may also be useful. As one example, in Appendix \ref{subsec:facebook}, we also report an alternate measure dividing each $r_i$ by the size of $i$'s indirect neighborhood (containing all predecessors that have a path of any length to $i$). In less connected networks, this measure isolates information loss due to aggregation failures.

\section{\label{sec:The-Generations-Network}Generations Networks}

In this section, we study information confounding in generations networks, where agents are arranged into sequential generations and observe some members from the previous generations. We find these network structures can lead to almost complete information loss. We begin by deriving the aggregative efficiencies of a class of networks with symmetric observations between generations, and find these networks aggregate at most 2 signals of information per generation asymptotically. We then show the basic force can also extend to settings where agents observe many previous generations, to networks with randomly drawn observations, and to early generations.

  Before beginning our analysis, we briefly comment on the relevance of networks with cohort structure.  In various empirical settings and theoretical models in the social sciences, people primarily learn from the recent past in ways that resemble the generations networks we study in this section. For example, the early years of medical residency typically involve the junior doctors rotating with different units in the hospital and learning from various attending physicians (i.e., the agents from the previous generation). In creating rotation schedules, hospitals try to aim for balanced assignments across residents and attending physicians.\footnote{Several producers of specialized software for scheduling medical residents discuss this as a software feature. Lightning Bolt Scheduling's blog says that ``Creating balanced schedules is crucial for promoting
resident wellness and avoiding perceptions of favoritism'' (see \href{https://www.lightning-bolt.com/blog/residency-schedule-challenges/}{https://www.lightning-bolt.com/blog/residency-schedule-challenges/}).
 Intrigma Scheduler advertises  ``Workload Balance and Fairness'' as
one of the key features of their software on their homepage (see \href{https://www.intrigma.com/specialties/residency}{https://www.intrigma.com/specialties/residency}).} Similar rotation programs where newcomers learn from incumbents are common in laboratory-based doctoral programs, government agencies, and large companies. On a broader societal scale, generation structures appear in leading models of cultural transmission. In anthropology, theoretical models such as conformist transmission and prestige bias (\cite{boyd1988culture}; \cite{richerson2008not}) can be viewed as behavioral models of social learning on generations networks. Unlike our analysis, these models assume mechanical learning rules such as following a majority of observed predecessors.

\subsection{Symmetric Generations Networks}

Consider generations networks where agents only observe predecessors from the previous generation.\footnote{This class of networks follows a strand of social-learning literature
where agents move in generations, for instance \citet{wolitzky2018learning},
\citet{banerjee2004word}, \citet{burguet2000social}, and \citet*{dasaratha2020learning}.} 
Agents are sequentially arranged into generations of size $K$, with
agents within each generation placed into \emph{positions} 1 through
$K.$ Agents in the first generation (i.e., $i=1,...,K$) have no
neighbors. A collection of \emph{observation sets} $\Psi_{k}\subseteq\{1,...,K\}$
for $k=1,...,K$ defines the network $M$ among the agents.
The agent in position $k$ in generation $t\ge2$ observes agents
in positions $\Psi_{k}$ from generation $t-1$ (and no agents from
any other generation). That is, for $i=(t-1)K+k$ where $t\ge2$ and
$1\le k\le K,$ network $M$ has $N(i)=\{(t-2)K+\psi:\psi\in\Psi_{k}\}.$\footnote{\citet*{stolarczyk2017loss} study a related model where only the
first generation observes private signals. Their main results characterize
when no information gets lost between generations, i.e., social learning
is completely efficient.} Figure \ref{fig:d2c1} shows an example with $K=3$.
\begin{figure}[h]
\begin{centering}
\includegraphics[scale=0.2]{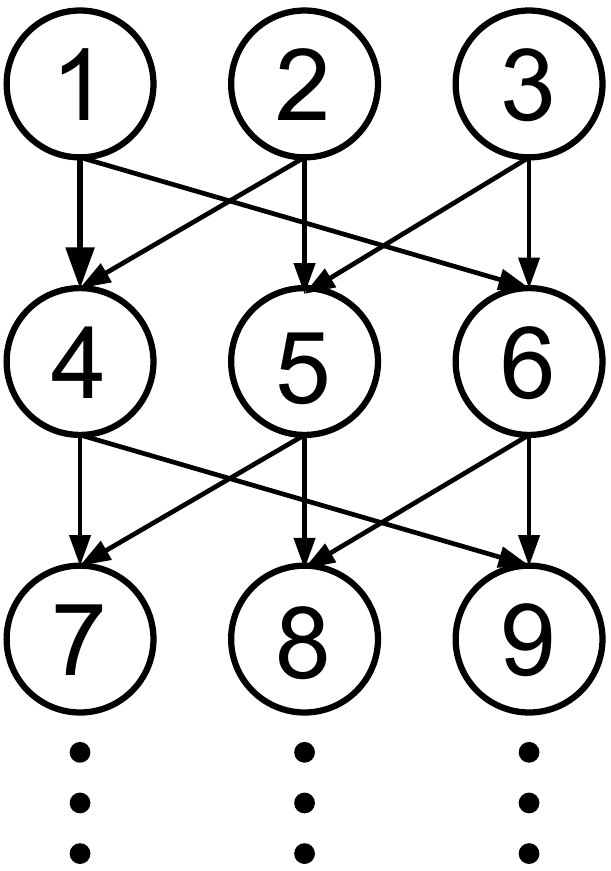}
\par\end{centering}
\caption{\label{fig:d2c1}A generations network with $K=3$ agents per generation
and the observation sets $\Psi_{1}=\{1,2\},$ $\Psi_{2}=\{2,3\},$
and $\Psi_{3}=\{1,3\}$.}
\end{figure}

We focus on observation sets $(\Psi_{k})_{k}$ satisfying a symmetry condition:
\begin{defn}
The observation sets are \emph{symmetric} if all agents observe $d\ge1$
neighbors and all pairs of agents in the same generation share $c$
common neighbors, i.e. $|\Psi_{k}|=d$ for every $1\le k\le K$ and
$|\Psi_{k_{1}}\cap\Psi_{k_{2}}|=c$ for distinct positions $1\leq k_{1}<k_{2}\leq K$.
\end{defn}
A generations network defined by symmetric observation sets is called
a \emph{symmetric network}. To give a class of examples of symmetric
networks, fix any non-empty subset $E\subseteq\{1,...,K\},$ and let
$(\Psi_{k})_{k}$ be such that for all $1\le k\le K,$ $\Psi_{k}=E.$
Here we have $d=c=|E|.$ To interpret, the set $E$ represents the prominent
positions in the society, and agents only observe predecessors in
these prominent positions from the past generation. We call the special
case of $E=\{1,...,K\}$ the \emph{maximal generations network}, where
agents in generation $t$ for $t\ge2$ have all agents in generation
$t-1$ as their neighbors. For another class of examples, suppose $K\ge2$ and each agent observes
a different subset of $K-1$ predecessors from the previous generation.
Specifically, $\Psi_{k}=\{1,...,K\}\backslash\{k-1\}$ for $2\le k\le K,$
and $\Psi_{1}=\{1,...,K-1\}.$ This network is symmetric with $d=K-1$
and $c=K-2.$ (The network in Figure \ref{fig:d2c1} has this structure,
with $d=2$ and $c=1$.) 

There are also a large variety of other symmetric
networks that are not covered by these two classes of examples: one
enumeration shows there are at least 103 pairs of feasible $(d,c)$
parameters in the range of $3\le d\le41$ and $1\le c\le d-2$ that
correspond to at least one symmetric network, typically with multiple
non-isomorphic networks for each feasible parameter pair \citep{MATHON1985275}.

We provide an exact expression for the aggregative efficiency in symmetric
generations networks to quantify the information loss due to confounding.

\begin{thm}
\label{thm:efficiency_dc}Given any symmetric observation sets $(\Psi_{k})_{k}$
where every agent observes $d$ neighbors and every pair of agents
in the same generation share $c$ common neighbors, aggregative efficiency
is\footnote{In the case $d=1$ and $c=0$, we adopt the convention $0/0=0$.}
\[
\lim_{i\to\infty}(r_{i}/i)=\left(1+\frac{d^{2}-d}{d^{2}-d+c}\right)\frac{1}{K}.
\]
The number of signals aggregated per generation asymptotically $(\lim_{i\to\infty}r_{i+K}-r_{i})$ is less than $2$.
For $c\ge1,$ this number is strictly
increasing in $d$ and strictly decreasing in $c$.
\end{thm}
Theorem \ref{thm:efficiency_dc} calculates the aggregative efficiency in any symmetric generations network in terms of the parameters $d$ and $c$. The expression on the right-hand side normalizes
by the size of the generation $K,$ so the term in the parentheses
provides a uniform learning-rate upper bound of two signals per generation
across all symmetric networks (as $\frac{d^{2}-d}{d^{2}-d+c}\le1$).

The interpretation of the comparative statics result in $d$ and $c$
is that more observations speed up the rate of learning per generation
but more common neighbors slow it down by worsening confounding, all else equal. This result lets
us compare learning dynamics across different symmetric networks characterized
by different $(d,c)$ parameter pairs. Changing from one network to
another, we can change both $d$ and $c$ (along with the generation
size $K$). Theorem \ref{thm:efficiency_dc} decomposes the repercussions
of such changes on the \emph{per-generation} learning rate  into their effects on the two primitive
network parameters that have monotonic influences on said rate.

The main content of the theorem is the uniform bound on the learning rate, which implies learning is very inefficient in large symmetric generations networks. The proof of Proposition \ref{prop:long_run_conditions} provides
a lower bound of one signal aggregated per generation,
since agents could always optimally combine their private signal with one observed action.
Theorem \ref{thm:efficiency_dc} shows this lower bound is not too
far from the actual learning rate, which is fewer than two signals per generation.

For maximal generations networks (i.e., agents observe all predecessors from the previous generation), the basic intuition for this bound is similar to
Example \ref{ex:Many-Neighbors-with}, which tells us that when any number
of agents observe one or more common signals in addition to their
private signals,  a successor who observes all of these agents cannot
improve on their accuracy by more than three signals worth of information.
The successor must balance overweighting the common confound and underweighting
her neighbors' private signals, and the optimal weights severely underweight recent signals.

Extending this intuition beyond maximal generations networks is more
subtle, because different agents in a generation  may observe different
predecessors whose actions may be less correlated. This can alleviate
the information confounding in early generations, but we show the
benefits  are limited: even if agents in the same generation have almost disjoint observation sets, actions become highly correlated in later generations. To prove this, we  use a mixing argument to show that
the actions of two agents in the same generation are influenced in
very similar ways by the signal realizations of their common ancestors
from many generations ago. An implication is that recent signals are severely underweighted, as in the maximal generations case: the \emph{total} weight an agent places on private signals from the previous generation converges to one, while in the absence of network-based confounds the agent would place a weight of one on \emph{each} signal.

Perhaps surprisingly, an implication of Theorem \ref{thm:efficiency_dc} is that  aggregative efficiency in symmetric networks  only
depends  on the generation size. Compare the symmetric
network from Figure \ref{fig:d2c1} with $d=2,c=1,K=3$ with the maximal
generations network with $K=3$. Theorem \ref{thm:efficiency_dc}
implies they have the same aggregative efficiency. The extra social
observations in the second network exactly cancel out the reduced
informational content of each observation, due to the more severe
information confounds. Our next result shows that more generally,
any symmetric network with parameters $(d,c,K)$ where $d\ge2$, $c<d$
has the same aggregative efficiency as the maximal generations network
with the same generation size $K$.
\begin{cor}
\label{cor:cancellation} In any symmetric network with $d\ge2$ and $c<d$ and in the maximal generations network,
aggregative efficiency is $\lim_{i\to\infty}(r_{i}/i)=(2-(1/K))\cdot\frac{1}{K}.$ In particular, aggregative efficiency is strictly decreasing in $K$, while the number of signals eventually aggregated per generation $2-(1/K)$ is strictly increasing in $K$. 
\end{cor}
This corollary follows from the fact that the symmetry condition imposes
some combinatorial constraints on the feasible $(d,c,K)$ parameter
triplets. It turns out
these constraints allow us to simplify the expression in Theorem \ref{thm:efficiency_dc}
when we know the generation size. While Corollary \ref{cor:cancellation}
gives a simple expression of aggregative efficiency that just depends
on $K$, Theorem \ref{thm:efficiency_dc} gives an expression for the number of additional signals  per generation that depends only on $d$ and $c$. This lets us compare networks
that differ in $d$ and $c$ more transparently (allowing $K$ to adjust to maintain feasibility).
%Therefore, the idea of worse efficiency
%with larger generations depicted in Figure \ref{fig:Number-of-signals}
%for maximal generations networks also holds in the broader class of
%symmetric networks.

In particular, these results imply that the aggregative efficiency of the maximal generations network with $K$ agents per generation is $\lim_{i\to\infty}(r_{i}/i)=\frac{(2K-1)}{K^{2}}$, which decreases with $K$.
Indeed, if $K=1,$ then every agent perfectly incorporates all past
private signals and the speed of social learning is the highest possible.
Not only does this result about the aggregative efficiency imply that
asymptotically fewer signals are aggregated by the same agent in networks
with larger $K$, but the same comparative statics  also hold numerically
for all agents $i\ge17$ when comparing among $K\in\{2,3,4,5\}$,
as shown in Figure \ref{fig:Number-of-signals}.  

\begin{figure}
\vspace{-5bp}

\begin{centering}
\includegraphics[scale=0.75]{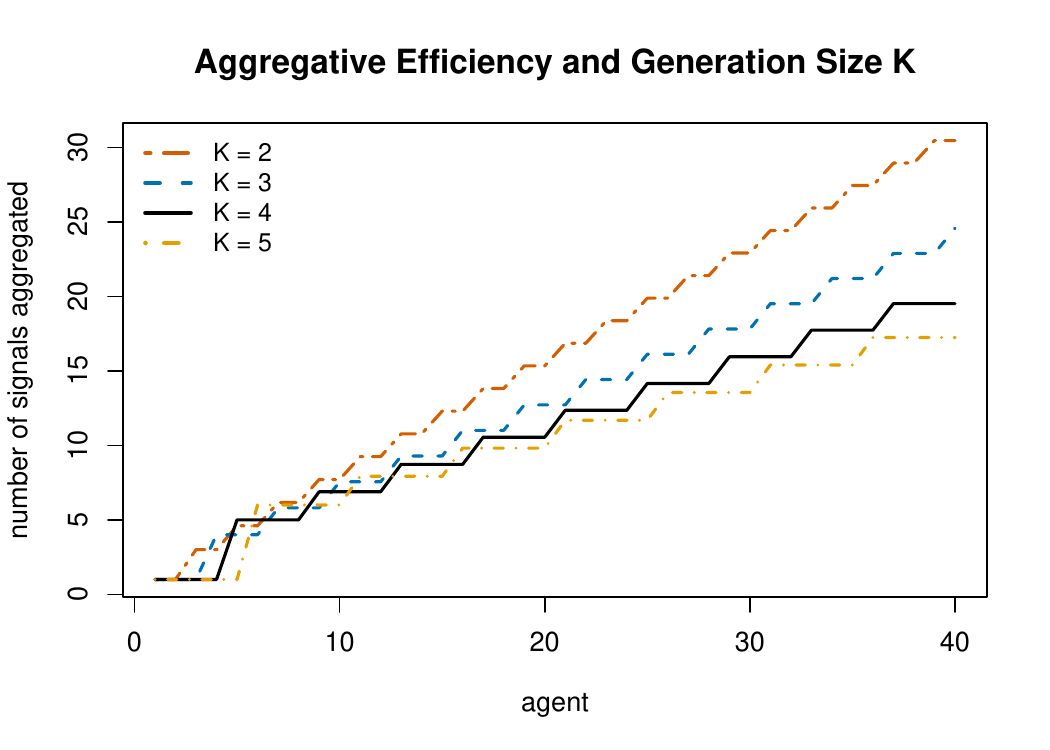}
\par\end{centering}
\vspace{-10bp}

\caption{\label{fig:Number-of-signals}Number of signals aggregated by social
learning in maximal generations networks with different generation
sizes, $K\in\{2,3,4,5\}.$}
\end{figure}

We next discuss the role of the symmetric observation  sets in  the bound of two additional signals per generation. We have implicitly imposed three restrictions on the network, beyond the basic generations structure: the observation sets are the same across generations, all generations have the same size, and observations are symmetric within generations. We assume the same observation structure across generations primarily for simplicity of exposition, and could bound aggregative efficiency with the same techniques while allowing different symmetric observation structures in different generations. The key step in extending the proof is the Markov chain mixing result, which must be replaced by a mixing result for non-homogeneous Markov chains (for examples of such results, see \cite{blackwell1945finite} and \cite{tahbaz2006consensus}). We could also allow different sized generations and obtain bounds on aggregative efficiency. For example, the logic of Example \ref{ex:Many-Neighbors-with} would extend to maximal generations networks with generations of varying sizes.

The assumption of symmetry within generations is more substantive, and our bound of two signals of additional information per generation does not always apply to non-symmetric generations networks. For example, suppose $\Psi_k = \{k\}$ for some position $k$. The accuracy of agents in this position will only increase by one additional signal per generation, but these agents can aggregate independent information that can be valuable to agents in other positions (see Section~\ref{subsec:Application-2:-Information} for details). An unresolved question is whether asymmetric observation sets could let \emph{all} agents aggregate more than two signals per generation.

\subsection{Multi-Generation Observation Networks}

In symmetric generations networks, agents only observe members of the most recent previous generation. This implies they do not see any of their neighbors' social information, creating confounding if multiple neighbors share common observations. 

Now consider a generations network where agents can see all predecessors from the previous \emph{two} generations (agents in generation 1 still have no observations and agents in generation 2 see all members of generation 1). This cures information confounding for generation 3, who can now fully account for  generation 2's social observations and perfectly aggregate all signals from the first two generations. Unfortunately, generation 3's additional observations create new confounds down the line for generation 4. Generation 4 observes generation 3 and generation 2, which would be enough to disentangle generation 3 agents' private signals if they only observed generation 2 (as in a maximal generations network). But, because generation 3 additionally observes generation 1 in this new network, these extra observations that helped de-confound generation 3 end up confounding generation 4. 

More generally, consider multi-generation observation networks where agents observe all predecessors from the past $L$ generations, where $1 \leq L < \infty$.  Agents in generations $t \leq L$  observe predecessors from all previous generations. The next result characterizes aggregative efficiency in these networks in terms of the largest eigenvalue of the $L \times L$ tridiagonal matrix \[
\left(\begin{array}{ccccccc}
K & K-1 & 0 & 0 & 0 & \cdots & 0\\
K-1 & K-1 & K-1 & 0 & 0 & \ddots & \vdots\\
0 & K-1 & K-1 & K-1 & 0 & \ddots & 0\\
\vdots & \ddots & \ddots & \ddots & \ddots & \ddots & K-1\\
0 & \cdots & 0 & 0 & 0 & K-1 & K-1
\end{array}\right),
\]
which we call $\rho(K,L).$

\begin{thm} \label{thm:L_gen}
Suppose generations have size $K \ge 2$ and agents in generation $t$ 
observe all agents from the previous $L$ generations. 
Then, $\lim_{i\to\infty}r_{i}/i=\frac{1}{K}g(K,L),$ where $g(K,L)=\frac{\rho(K,L)+K-1}{K}$. This implies: 
\begin{itemize}
\item For any $L,$ $g(K,L)\le4-4/K<4.$
\item $g(K,L)$ is strictly increasing in $K$ for any $L\ge1$ and strictly
increasing in $L$ for any $K\ge2$. 
\end{itemize}
\end{thm}

In the case $L=1$, we saw in the previous section that each generation adds  fewer  than $2$ signals of information asymptotically. Theorem \ref{thm:L_gen} gives a similar finite upper bound for each $L>1$, and moreover gives a small finite upper bound that is independent of $L$: the number of signals $g(K,L)$ eventually aggregated per generation is bounded by $4$ for \emph{any} $L$. The intuition is that, as the example above with $L=2$ suggests, the benefit of observing predecessors' social information is offset by the cost of creating additional confounds for successors. (If $K=1$, then there is no confounding and signal aggregation is perfect for any $L$.) We also find that increasing the generation size $K$ and the number of past generations observed $L$ both increase the number of signals aggregated per generation, but the benefits are capped by the uniform upper bound of $4$.

Unlike in symmetric generations networks, agents do not treat all of their observations symmetrically in multi-generation observation networks with $L \ge 2$. For example, in generation 3, agents put negative weights on generation 1's log-actions (as in \cite{eyster2014extensive}) to subtract off the common social information observed by generation 2. The proof of Theorem \ref{thm:L_gen} requires us to keep track of the conditional covariance matrix for the previous $L$ generations' log-actions, which in turn determines the next generation's rational weights on previous generations' log-actions and hence how this next generation correlates with the previous ones. It turns out this is equivalent to keeping track of the conditional covariance matrix for the \emph{differences} between successive generations' log-actions, and the limit of this covariance iteration has an eigenvalue characterization. The key step uses a series of recursions to express these differences as a continued fraction and then relates this same continued fraction to the eigenvalue $\rho(K,L)$.

\subsection{Random Regular Networks}

Theorem~\ref{thm:efficiency_dc} found that confounding severely obstructs learning in generations networks with symmetric observation sets between generations. To explore the robustness of the basic forces, we next consider generations networks for which the network between generations is randomly drawn. While a full characterization of aggregative efficiency is considerably more difficult in this setting, we show a class of regular networks with suitable degrees also causes substantial confounding. This confounding can again lead to the loss of almost all of a large generation's information.

To generate a $d$-regular directed graph, define observation sets $\Psi_k$ by sampling $d$ slots $j \in \{1,\hdots,K\}$ uniformly at random (without replacement), independently for different $1 \le k \le K$. The realized observation sets $\Psi_k$ are common knowledge for the agents and are used for all generations. We suppose $d=\lfloor K^{\alpha} \rfloor$, where $0 < \alpha < 1/4$  and  $\alpha \ne 1/n$ for any integer $n$. So in the random regular networks, all agents observe $d$ neighbors from the previous generation, but in general different pairs of neighbors will be correlated in different ways. The observed neighbors form an arbitrarily small fraction of the generation as $K\to \infty$. 

Let $t^* = \lceil \frac{1}{\alpha} \rceil.$  We analyze learning outcomes around generation $t^* + 1$ as the generation size $K \rightarrow \infty$.

\begin{thm}\label{thm:d_regular}
    Suppose $d=\lfloor K^{\alpha} \rfloor$ for some $\alpha < \frac14$. Then there exists $\epsilon>0$ such that $$r_{i}-r_j \leq K^{1-\epsilon}$$
    for each $i$ in generation $t^*+2$ and $j$ in generation $t^*+1$ a.a.s.
\end{thm}

At distances less than $t^*$, the observation network largely resembles a tree: most nodes reachable from a node $i$ can only be reached by  a unique path. But at distance $t^*$, there are many paths between any two nodes. So generation $t^*+1$ is the first to have (indirect) access to the signals of a full generation. Generation $t^*+2$, in turn, gains (indirect) access to the signals of two full generations. But Theorem \ref{thm:d_regular} tells us almost all of the additional generation's signals are lost due to confounding: the number of additional signals aggregated is at most a $K^{-\epsilon}$ fraction of the generation size, which approaches 0 as $K \to \infty$.

A difference from Theorem~\ref{thm:efficiency_dc} is that most pairs of  agents in the same generation have no common neighbors. But we can still show a similar mixing result: all agents in generation $t^*+1$ put very similar weights on agents in generation $1$. As in Theorem~\ref{thm:efficiency_dc}, this is sufficient to produce confounding that prevents almost all aggregation of the subsequent generation's signals. The proof is considerably more technically involved because agents no longer weight their observations symmetrically. A careful graph-theoretic analysis establishes a mixing result despite these deviations.

The theorem is stated probabilistically, but there is an equivalent deterministic interpretation. Let $\mathcal{G}(d,K)$ be the set of $d$-regular directed graphs with $K$ nodes. Given an element $G$ of $\mathcal{G}(d,K)$, we can define a generations network with generation size $K$ by setting observation set $\Psi_k$ equal to the neighborhood of $k$ in the graph $G$ for each node $k$. Then the bound in Theorem~\ref{thm:d_regular} holds for all but a vanishing fraction of graphs in $\mathcal{G}(d,K)$ as $K$ grows large. 

The result shows generations networks that do not have symmetric observation sets can still create  severe confounding. The class of random networks studied here assumes a particular structure that requires agents to have many observations, but is also less restrictive in several respects. First, the notion of symmetry is weaker: we only require that agents have the same number of observations. (The proof also allows some deviation in agents' degrees as long as they are close to equal.) This could arise from various organizational structures specifying, e.g., a fixed number of mentors or a fixed number of rotations for newcomers. Second, symmetric generations networks require that all pairs of agents in the same generation have  $c$ common neighbors. We show that this is not necessary for the network to generate confounding, as long as pairs of agents are indirectly connected to a common ancestor at some distance (which can be large if $\alpha$ is small).

Simulations suggest severe confounding also appears in generations networks with much smaller degrees and with substantial variation in degrees across agents. We first simulated random regular networks with $K=100$ and $d \in \{2, 4, 8, 16\}$. We conducted 1000 repetitions of the simulation, redrawing the observation sets $(\Psi_k)_{k=1}^{K}$ in each repetition and calculating each agent's $r_i$ up to generation 100. Figure \ref{fig:random_d} shows the maximum increase in $r_i$ in one generation across all 100 positions. We find that across all parameter values, the maximum one-period gain in signals quickly converges to 2 by around generation 30. Indeed, in none of the 1000 simulations with different random observation sets did we find any agent position that aggregated more than 2 additional signals in any generation $t\ge 43$ for any parameter value. This convergence is more rapid for higher $d$, analogous to a smaller $t^*$ in Theorem \ref{thm:d_regular} when $\alpha$ (and hence $d$) is larger. Figure \ref{fig:random_d1d2} shows the same numerical evidence for generations networks with $K=100$ where half of the positions have observation sets with size $d_1$ and the other half have observation sets with size $d_2 \ne d_1$, again with all observation sets randomly drawn across  1000 repetitions of the simulation. Here, we found no position that aggregated more than 2 additional signals in any generation $t\ge 27$ for any  parameter value and any of the random observation sets.

\begin{figure}[htbp]
  \centering
  \begin{subfigure}{0.47\textwidth}
    \centering
    \includegraphics[width=\linewidth]{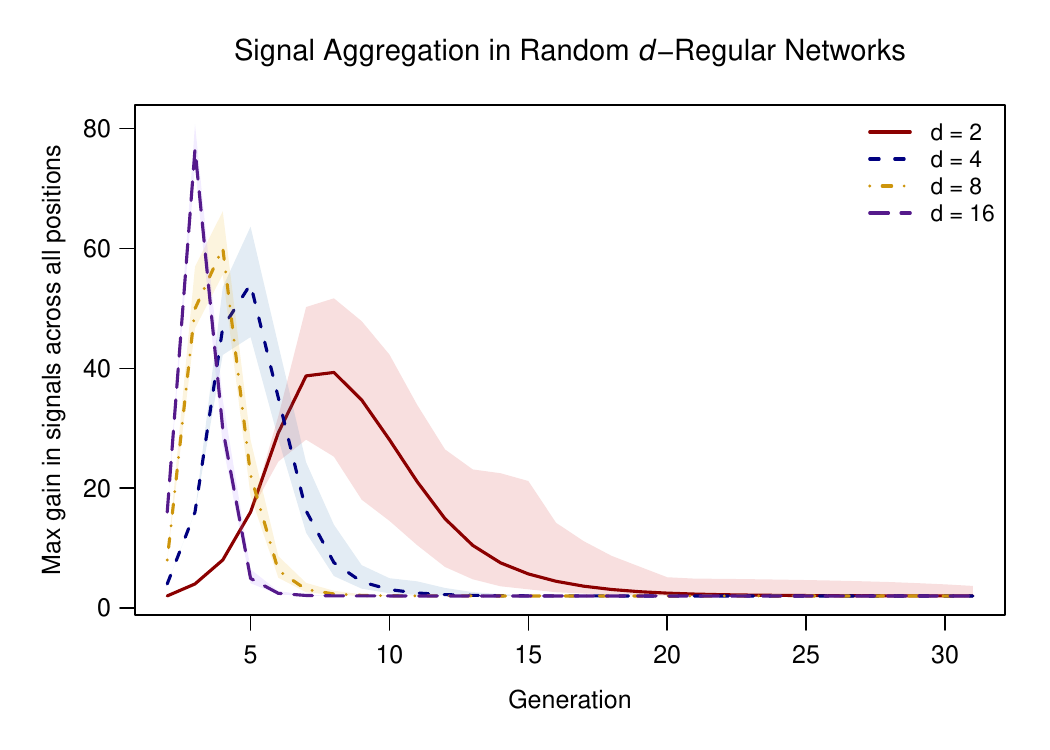}
    \caption{Random $d$-regular networks.}
    \label{fig:random_d}
  \end{subfigure}
  \hfill
  \begin{subfigure}{0.47\textwidth}
    \centering
    \includegraphics[width=\linewidth]{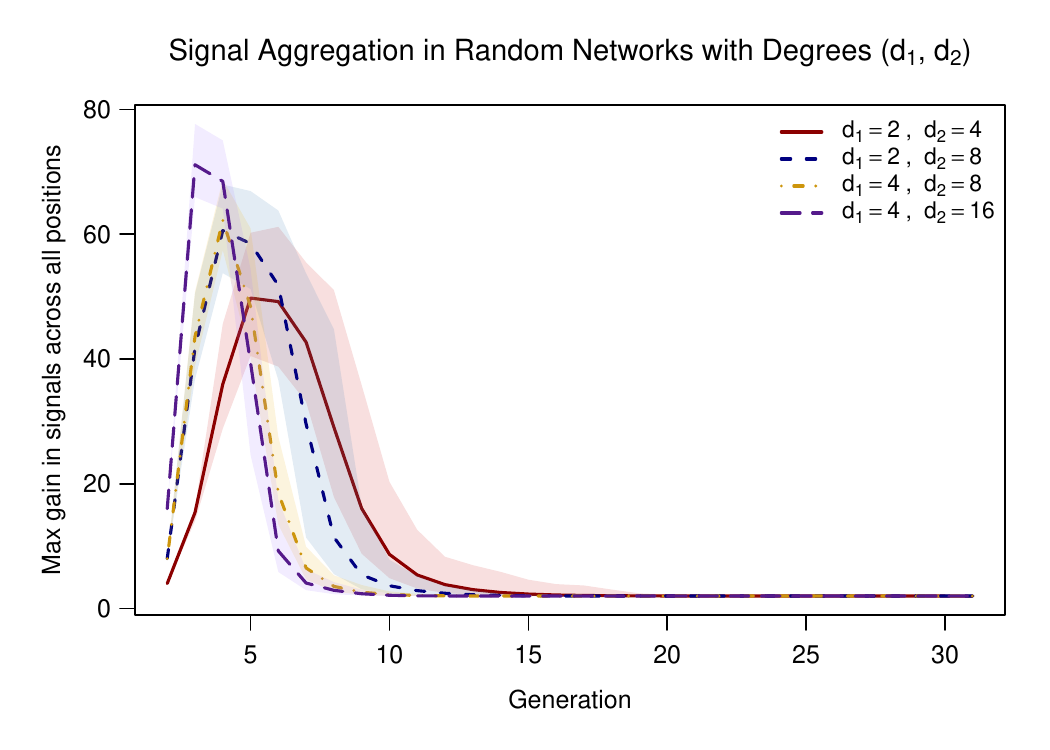}
    \caption{Random networks with degrees $(d_1,d_2)$.}
    \label{fig:random_d1d2}
  \end{subfigure}

  \caption{Signal aggregation by generation in networks with $K=100$ agents per generation. In both plots, the solid lines show the largest one-generation increase in $r_i$ across agents in all 100 positions, averaged across 1000 repetitions of the simulation. The shaded areas show the range of this increase across the 1000 repetitions. In each repetition, the observation sets for the 100 positions are randomly and independently drawn, but kept fixed for all generations.   
  Subfigure (a): All positions have $d$ neighbors each. Subfigure (b): Positions 1 through 50 have $d_1$ neighbors each and positions 51 through 100 have $d_2$ neighbors each.}
  \label{fig:example}
\end{figure}

If networks are sufficiently sparse and do not have generations structure, however, confounding can disappear. In tree networks, for example, each agent perfectly aggregates their private signal and the signals of all predecessors in the tree. The same result holds for almost all agents in large random graphs which are locally tree-like, and this structure arises if the average degree is less than $1$. So very sparsely connected networks can have little or no confounding---but do have other constraints on learning because most or all agents cannot access most of society's signals.

\subsection{Finite-Population Learning}

Our framework not only allows us to study the asymptotic rate of learning,
but also lets us derive finite-population bounds that apply from early generations.
Theorem \ref{thm:efficiency_dc} tells us social learning aggregates
fewer than two signals per generation asymptotically. There is also
a short-run version of this result in the maximal generations network:
starting with generation 3, fewer than three signals are aggregated
per generation for any $K$. 
\begin{prop}
\label{prop:starting_3}In any maximal generations network, for any
agents $i,i^{'}$ in generation $t$ and $t-1$ with $t\ge3,$ $r_{i}-r_{i^{'}}\le3$.
\end{prop}
We find an even starker bound on $r_{i}$ if we consider a modified
version of the maximal generations network: there is a zeroth generation
with only one agent, and all subsequent generations contain $K$ agents
each. Agents in generation $t\ge1$ observe all predecessors from
generation $t-1.$ 
\begin{prop}
\label{prop:zeroth_gen}In this modified maximal generations network,
$r_{i}\le3t-1$ for every $i$ in every generation $t\ge1$. 
\end{prop}
Similar to Example \ref{ex:Many-Neighbors-with},
the single agent before the first generation causes significant information
confounding. With this additional agent, there is a uniform bound on \emph{every} generation's
accuracy across all generation sizes $K$. 

Finally, the proof of Theorem \ref{thm:L_gen} implies a lower bound on information aggregation that applies to  every generation in multi-generation observation networks. 

\begin{prop}
\label{prop:lower_bound}
    In any multi-generation observation  network, for any
agents $i,i^{'}$ in generation $t$ and $t-1$ with $t\ge2,$ $r_{i}-r_{i^{'}}\ge 2- 1/K$.
\end{prop}

In particular, this implies that in maximal generations networks, every generation aggregates more than the asymptotic number of $2-1/K$ signals per generation.

\section{Organizational Applications}

Section~\ref{sec:The-Generations-Network} suggests that sustaining substantial social learning can be challenging in settings with cohort structures. We now describe two organizational changes that can generate better outcomes for at least some agents. First, if new workers are paired with mentors who can share their private information or relevant organizational context, learning becomes much more efficient. Second, asymmetric network structures dividing some agents into isolated \emph{information silos} can help other agents learn better.

\subsection{Value of Mentorship\label{subsec:AppMentorship}}

We provide an economic application of our results in terms of the
value of mentors who share their private signals with mentees in the
next generation.

Many organizations with cohort structures, such as universities and
firms, have mentorship programs that pair newcomers with members of
a previous cohort. Our results suggest that one benefit of such programs
is that mentors provide information that helps newcomers interpret
others' actions, thus increasing the speed of learning within the
organization.

Formally, we model a mentor as someone who shares her private signal
with a mentee in the subsequent generation. Equivalently, the mentor
could share a sufficient statistic describing her best estimate of
the state based on her social observations. If we begin with the maximal
generations network and add mentorship relationships in this way,
learning is nearly efficient.
\begin{cor}
\label{cor:mentor}Suppose each agent observes the actions of all
members of the previous generation and the private signal of one member
of the previous generation. Then social learning aggregates more than
$i-K$ signals by every agent $i,$ so aggregative efficiency is 1.
\end{cor}
If an agent observes the actions of the previous generation along
with one predecessor's private signal, she can calculate the common confounding
information and fully compensate for this confound. In networks with
large $K$, showing each agent just one extra signal (of someone from
the previous generation) increases aggregative efficiency from nearly
0 to 1.

In the context of the application, incumbents in the organization
act based on private information and shared organizational knowledge.
A newcomer ignorant of the organizational knowledge cannot fully separate
these two forces that shape others' behavior. But by describing her
perspective, a mentor can help a newcomer interpret everyone else's
behavior. This removes the informational confound facing the newcomer
and lets her extract the private information underlying these predecessors'
actions. A related force is described in management literature:
\begin{quote}
\textquotedblleft Mentors can be powerful socializing agents as an
individual adjusts to a new job or organization. As prot\'{e}g\'{e}s learn
about their roles within the organization, mentors can help them correctly
interpret their experiences within the organization\textquoteright s
expectations and culture.\textquotedblright{} \textendash{} \citet{ragins2007handbook}
\end{quote}
Our result formalizes this intuition in a social-learning environment.
Our stylized model of mentorship abstracts away from many of its other
benefits (e.g., the expertise of the mentor in terms of being able
to generate more precise signals than the mentee), and shows how the
``interpretive'' value of mentorship improves learning within the
organization. Potentially consistent with this interpretation, \cite*{sandvik2020workplace} find that encouraging meetings where employees seek advice from peers led to large performance gains in a sales firm. The idea that observing predecessors' private information can be quite beneficial also appears in \cite*{bloom2013does}, who find that implementing record-keeping practices to document and analyze the context behind past behavior (such as the reasons behind machine failures and product defects) improves firm performance.

\subsection{Information Silos\label{subsec:Application-2:-Information}}

Within some organizations, information is fragmented among various
subgroups (departments, product divisions, trading desks) that fail
to communicate with each other, creating \emph{information silos}.\footnote{We thank Suraj Malladi for suggesting this application.}
These silos have a number of causes: compensation structures that
discourage collaboration between teams, different subunits storing
information in mutually incompatible databases, or technical language
barriers that stop ideas from flowing between specialties.\footnote{\citet{sethi2016communication} show that a silo-like information
segregation may become endogenously entrenched in an organization
as each agent learns the subjective perspectives of the people she
talks to most often. This encourages the agent to keep consulting
the same people's opinions in the future, as she can better account
for their subjective biases and extract more precise information from
their opinions.} \citet{tett2015silo} documents the prevalence of information silos
in government bureaucracies, technology firms, and banks, noting that
many of these silos persist for many decades. She joins a number of
other authors and management consultants in arguing that information
silos are a necessary evil for running a complex workforce, but they
hurt the organization by obstructing internal information exchange.\footnote{Arguments in favor of eliminating information silos are common in
the popular press: see for example, \citet{gleeson2013silo} and \citet*{edmondson2019cross}.}

\begin{figure}[h]
\begin{centering}
\includegraphics[scale=0.3]{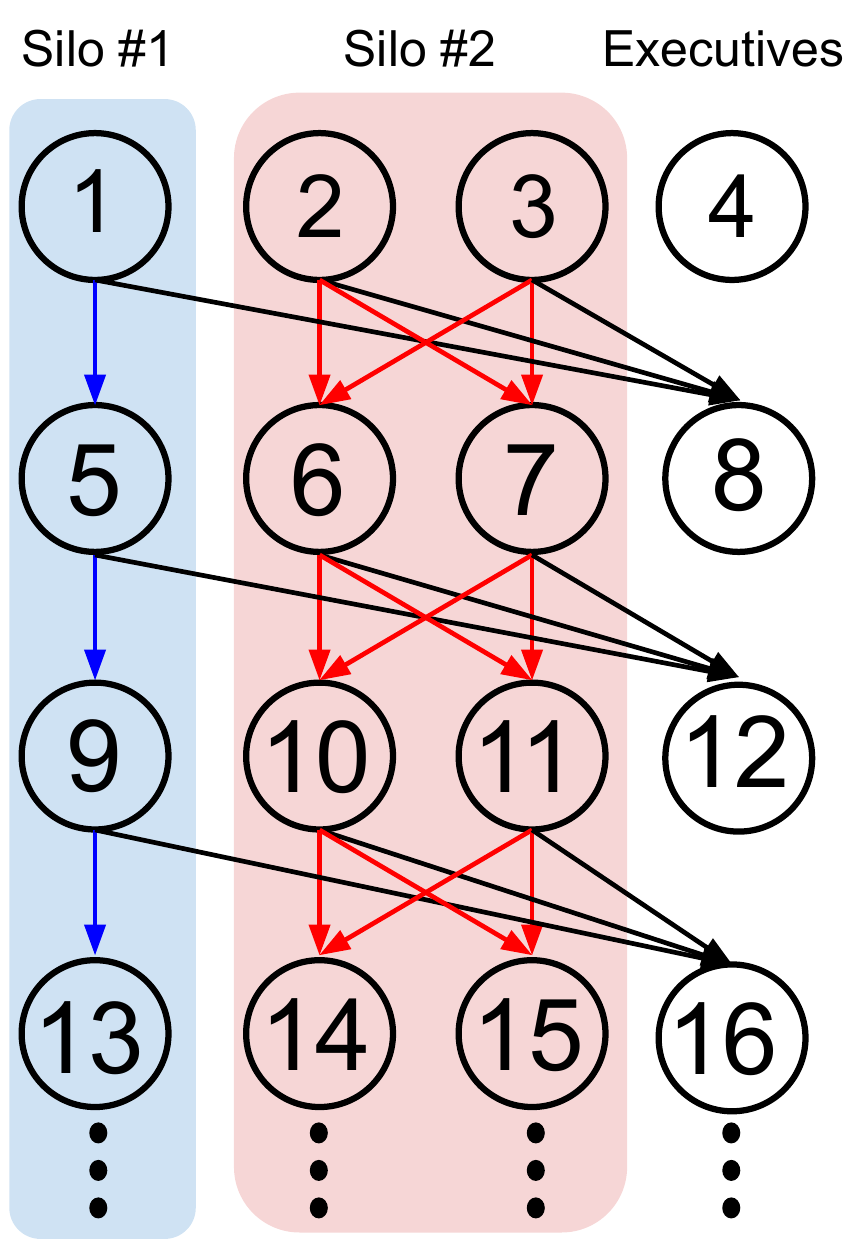}
\par\end{centering}
\caption{\label{fig:infosilos}A generations network with executives in 
 position 4 and two information silos, $S_{1}=\{1\}$ and $S_{2}=\{2,3\}.$}
\end{figure}

We use a generations network to show that information silos may benefit
the organization compared with fully transparent data sharing,\footnote{Similar network structures can also improve social learning in experimentation
settings. In a model where a sequence of short-lived, behavioral agents
take turns interacting with a multi-armed bandit, \citet*{immorlica2020incentivizing}
show that an observation structure featuring many information silos
ensures at least one silo produces a large amount of information about
the payoffs of each bandit arm, thus improving the welfare of later
agents who observe all the information generated in every silo.} when the organization's success primarily depends on the actions
of a few executives who can observe and process the behavior in all
the silos. Consider a generations network with $K+1\ge2$ agents per
generation. Suppose positions $\{1,...,K\}$ are the workers partitioned into
$N\ge1$ silos $S_{1},...,S_{N}$ so that each position $k$ only
observes predecessors in the same silo, $\Psi_{k}=S_{n}$ for $k\in S_{n}$,
while the agents in position $K+1$ are executives who can observe all of the silos, $\Psi_{K+1}=\{1,...,K\}.$ Figure \ref{fig:infosilos}
shows an example with two information silos that contain one and two
workers respectively in each generation. As a new cohort of workers join the organization,
each newcomer learns by observing their seniors from the same silo,
and information does not flow across different silos.

\begin{cor}
\label{cor:silo}
Agents in  position $K+1$ (who can observe all silos) eventually aggregate $$\lim_{t\to\infty}\frac{r_{(t-1)(K+1)+(K+1)}}{t}=\sum_{n=1}^{N}2 - \frac{1}{|S_{n}|}$$
signals per generation. If $k\in S_{n}$, then agents in position $k$
eventually aggregate $$\lim_{t\to\infty}\frac{r_{(t-1)(K+1)+k}}{t}=2 - \frac{1}{|S_{n}|}$$
signals per generation.
\end{cor}
Corollary
\ref{cor:silo} shows that executives can aggregate up to $K$ signals
per generation, depending on the sizes of the silos.  So, social learning can aggregate
more than three signals per generation for the executives when there
are silos. By contrast, the executives' information improves by no
more than three signals per generation starting from generation 3
in the maximal generations network (Proposition \ref{prop:starting_3}),
which represents an organization with full data transparency.

If the organization's payoff is closely identified with its executive's information aggregation, then information silos
can help the organization. Such an organizational structure
provides less confounded information to the key decision-makers by
sacrificing the rate of learning within silos. Indeed, behavior in
different silos are conditionally independent of each other. If the
organization is instead one where every member's information aggregation significantly
contributes to its payoff, then information silos are detrimental
to the organization. Newcomers could learn better by observing all
incumbents in the organization, instead of only those in the same
silo.

To formalize these tradeoffs, we ask how to optimally partition agents into silos. Consider an organization-design problem where the organization chooses how to partition the worker positions $\{1,...,K\}$ into silos.
For each position $k\in\{1,...,K+1\}$, write $g_{k}:=\lim_{t\to\infty}\frac{r_{(t-1)(K+1)+k}}{t}$
as the number of signals that agents in position $k$ eventually aggregate
per generation. Suppose the organization maximizes $\sum_{k=1}^{K}g_{k}+\alpha g_{K+1}$
for some $\alpha$. To interpret, positions $\{1,...,K\}$ refer
to the workers in the organization and position $K+1$ refers to the
executive. The organization's success depends on some weighted sum of
 how much information the workers aggregate and how much information
the executive aggregates per generation, with the weight on each worker's
information aggregation normalized to one.

The following result characterizes the optimal structure when $\alpha \geq 1$ (so executives' learning is at least as valuable as other agents' learning).

\begin{prop} \label{prop:optimal_silos}
Consider the organization's problem of maximizing $\sum_{k=1}^{K}g_{k}+\alpha g_{K+1}$
for some $\alpha \geq 1$. The optimal structure satisfies the following properties:
\begin{itemize}
\item The optimal number of silos is increasing in $\alpha$.
\item When $1\le\alpha<2$ and $K$ is even, the only optimal organization
has $N=K/2$ and every worker in a silo of size 2. When $1\le\alpha<2$
and $K$ is odd, every optimal organization has one silo of size
1 or 3 and all other silos have a size of 2. 
\item When $\alpha>2$, the only optimal organization is $N=K$ and every
worker is in a silo of size 1. 
\end{itemize}
\end{prop}
When $\alpha \geq 1$ and $K \geq 3$,  it is optimal to partition agents into multiple small silos. Moreover, the optimal  number of silos is higher when the objective function places more weight on the executive's learning.

%Formalizing this idea, in Appendix \ref{subsec:optimal_silos} we consider an organization-design problem where the organization chooses how to partition the worker positions $\{1,...,K\}$ into silos. Let $g_k$ represent the number of signals that agents in position $k$ eventually aggregate per generation, and suppose the organization maximizes $\sum_{k=1}^{K}g_{k}+\alpha g_{K+1}$
%for some $\alpha>0$ that represents the relative importance of the executive's information aggregation for the organization's success.

A large literature in organizational economics describes advantages to decentralized network structures, including for information aggregation (e.g., \cite{radner1993organization} and subsequent work on information processing and \cite{migrow2021designing} on cheap talk communication). Perhaps closest to our result, \cite{wolinsky2002eliciting} shows information silos can help a designer learn from information sources with misaligned incentives. The mechanism is that insulating agents in silos prevents them from collusively managing their pooled information when they communicate with the designer. We find that information silos can have benefits when information transmission instead happens via observational learning. Within the social learning literature, a number of papers point out that isolating individuals can benefit others by providing more independent information. For example, \cite{sgroi2002optimizing} discusses the benefits of ``guinea pigs'' who provide clean information to successors at the cost of their own learning and solves for an optimal number of such guinea pigs numerically. In our  organizational setting, related design problems have straightforward analytic solutions.

\section{Aggregative Efficiency and Welfare \label{sec:Welfare}}

In this section, we relate aggregative efficiency comparisons to welfare
comparisons. When signals are precise enough for agents to learn well,
welfare will depend largely on learning outcomes of finite networks
(such as the examples in Section \ref{sec:Examples}) rather than
asymptotic quantities. Higher aggregative efficiency implies higher
welfare, however, when signals are sufficiently imprecise.

Let $v_{i}^M:=\mathbb{E}[u_{i}(a_{i}^{*},\omega)]$ denote the expected
welfare of agent $i$ in network $M$, and recall that $-0.25<v_{i}^M<0$
for every $i$ in any network and with any private signal precision
$0<1/\sigma^{2}<\infty.$ It turns out that whenever the aggregative
efficiency of a network $M$ is strictly positive, $v_{i}^M\to0$ and this
convergence happens at an exponential rate. This implies the undiscounted
sum of utilities of all agents, $\sum_{i}v_{i}^M$, is convergent.

We show that if two networks are ranked by aggregative efficiency,
then the undiscounted sums of all agents' utilities on these networks
follow the same ranking whenever private signals are sufficiently
imprecise. The same result also applies to the discounted sums of
utilities, provided the discounting does not weigh the welfare of
the earliest agents too heavily.
\begin{prop}
\label{prop:summing_welfare}Suppose networks $M$ and $M'$ have aggregative
efficiencies $AE_{M}>AE_{M'}>0.$ Then there exists some $\underline{\sigma}^{2}>0$
such that for any signal variance $\sigma^{2}\ge\underline{\sigma}^{2}$,
we have $\sum_{i}\delta^{i-1}v_{i}^{M}>\sum_{i}\delta^{i-1}v_{i}^{M'}$
for $\delta$ sufficiently close to or equal to $1$.
\end{prop}
This result provides a foundation for the aggregative efficiency measure
in terms of a conventional social welfare function: the (un)discounted
sum of utilities. The result applies to arbitrary networks, and does not require the generations structure from Section \ref{sec:The-Generations-Network}.

\begin{figure}
\begin{centering}
\includegraphics[scale=0.5]{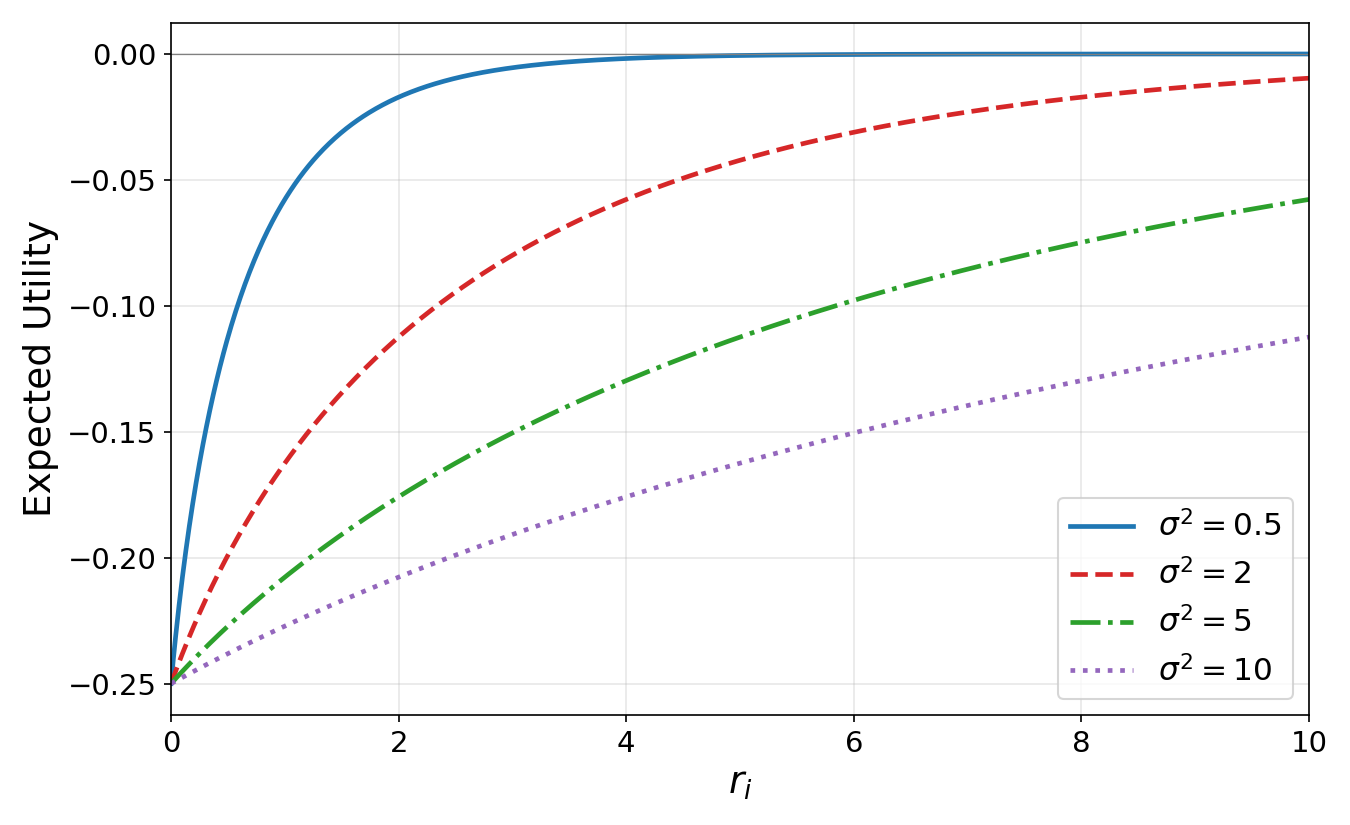}
\par\end{centering}
\caption{Expected utility as a function of the number of signals $r_i$ for different signal variances $\sigma^2$.}\label{fig:Utility-signal-counts}

\end{figure}

To give some intuition for how large $\sigma^2$ needs to be for signal counts to be relevant for welfare, Figure~\ref{fig:Utility-signal-counts} plots expected utility as a function of the number of signals $r_i$ for different signal variances $\sigma^2$. When signals are more precise the marginal value of increasing $r_i$ decreases quickly, but for less precise signals the plot is much closer to linear.

The arbitrarily large information loss we have highlighted in Section \ref{sec:The-Generations-Network} can have large welfare consequences. To illustrate this, we give an example comparing complete networks to maximal generations networks with large generations but a fixed amount of information per generation.

\begin{example}\label{ex:welfare_loss}
Let $M$ be the maximal generations network with generation size $K$. We will let $K$ grow large and set signal variance $\sigma^2 = K\sigma_0^2$ for a constant $\sigma_0^2>0$. That is, we increase the generation size but fix the total informativeness of a single generation's private signals. Let $M'$ be the complete network (or any other network with $r_i = i$ for all $i$). Then 
$$\lim_{K \rightarrow \infty} \frac{\sum_{i}v_{i}^{M}}{\sum_{i}v_{i}^{M'}} = \infty.$$
(We provide the details in the Appendix.) Since utilities are negative for all agents, the limit implies that the total disutility in the maximal generations network is unboundedly larger than when agents can extract all previous private signals.
\end{example}

The results in this section do not rely heavily on agents' utility functions: Proposition \ref{prop:summing_welfare} and Example \ref{ex:welfare_loss} continue to hold for any decision problem satisfying mild conditions as long as the social-learning process is unchanged. For example, these results go through if agents could directly observe their neighbors' beliefs (as in the main model) but  choose a welfare-relevant binary action. More generally, suppose each agent
$i$ directly observes the posterior belief of each of her neighbors
in $N(i)$ along with the private $s_{i}$, so agents' beliefs remain the same as in the main model. Agents also choose an action from a compact set $A\subseteq\mathbb{R}$, and the utility $u(a,\omega)$ from this action is the welfare-relevant object. Suppose the utilities $u(a,0)$
and $u(a,1)$ are continuous in $a$ and no action is weakly dominated, and normalize the expected utility under full information of the state to be zero. Then Proposition \ref{prop:summing_welfare} and Example \ref{ex:welfare_loss} continue to hold as stated, and the appendix gives proofs in this more general setting.

%The negative case studies that \citet{tett2015silo} and others use
%to advocate breaking down silos mostly involve workers in silos who
%take actions that severely harm the company, or executives who are
%unable to process the data from multiple silos. For instance, \citet{tett2015silo}
%discusses two product divisions of Sony simultaneously producing two
%very similar music players that ended up competing with each other
%on the market, a situation where the organization's welfare depends
%on the actions taken within the silos rather than the action of a
%single executive who oversees all silos.

%An important qualification is that if employees communicate private
%ignals along with actions (perhaps as in the application in Section
%\ref{subsec:AppMentorship}), then information silos will be harmful
%for social learning. When full information sharing is possible, information
%silos will lead to less informed workers without meaningfully improving
%executives' actions.

\section{Related Literature}\label{s:literature}

We study rational social learning in a sequential model (as first introduced by \cite{banerjee1992simple}
and \cite*{bikhchandani1992theory}) where agents  only observe some predecessors. %\footnote{In a different class of non-sequential social-learning models where
%a finite set of agents repeatedly observe their neighbors in a fixed
%network and simultaneously choose actions every period, \citet*{gale2003bayesian}
%and \citet{goyal2012connections} have compared learning dynamics
%in specific networks to highlight a possible trade-off between the
%accuracy of the long-run consensus and the speed of convergence to
%said consensus. In our setting with rich actions, this trade-off is
%absent as the long-run consensus is correct in all reasonable networks.}
Our contribution is to quantify how network structure affects
the rate and short-run accuracy of learning through information confounding. This leads us
to the new  conclusion that a small amount of  confounding can generate arbitrarily
inefficient social learning, even when agents perfectly observe their neighbors' beliefs. The remainder of our literature review focuses on fully rational learning, but much more is known about how network structure affects learning in settings where agents are misspecified about neighbors' social information  (including DeGroot learning \citep{demarzo2003persuasion} and the experimentation model of \cite{bala1998learning}, along with the large subsequent literatures).

With unboundedly informative signals, complete long-run learning holds on the complete network \citep{smith2000pathological} and more generally on all networks satisfying mild connectivity conditions \citep{acemoglu2011bayesian,lobel2015information}. \citet{rosenberg2017efficiency} study networks that exclude confounding and reach a similar conclusion that ``the nature of the feedback on previous choices matters little'' under their learning criterion. Our results instead show that networks with different levels of information confounding can exhibit large differences in short-run accuracy and learning rates.\footnote{A precedent to our comparison of learning rates across networks    is \citet*{lobel2009rate},
who examine two particular networks, both involving each agent seeing
exactly one neighbor. Our results allow us to compare networks that vary along richer dimensions, including the number of neighbors that agents have.}

The most closely related network-based obstructions to learning appear in \citet{eyster2014extensive}, who mention the possibility of such
confounds but restrict their analysis to networks where rational agents
can fully correct for correlations in observations via anti-imitation.
They note that relaxing this restriction to allow confounds would
lead to ``distributional complications''; our framework and results
resolve these complications and study the implications of the confounds. Related obstructions arise in \cite*{dasaratha2020learning}, who study generation-like networks in a changing-state environment  but have no formal results about how learning differs across networks.  They focus instead on how private signal precisions and state evolution affect learning, arguing that the network structure matters much less than the state and information structures in their setting. By contrast, we show that in a standard fixed-state environment the network structure can trace out a wide range of learning efficiencies,  including nearly total information loss.

%learning failures in network structures similar to our generations networks but has no formal results about how learning differs across networks. They focus instead on how private signal precisions and the evolution of the state, which changes over time, affect learning. Indeed, they argue that network structure matters much less than the state and information structures in their setting. By contrast, we show that in a standard fixed-state environment learning can be quite efficient on some networks and highly confounded on others. 
%Variations in the network structure can trace out a wide range of learning efficiencies, including nearly total information loss,  which highlights the power of the confounding.

The previous two paragraphs discussed settings with unboundedly informative signals, but sequential models with boundedly informative signals can provide an alternative
setting for asking how network structure affects learning. Complete long-run learning fails when signals are boundedly informative \citep{smith2000pathological}, and 
%a largely open question is how the probability of failures depend on the network structure. 
several papers show that incomplete networks can improve learning outcomes relative to the complete
network \citep*{sgroi2002optimizing,acemoglu2011bayesian,arieli2019multidimensional}.
These papers  compare  specific classes of incomplete networks with
the complete network, but  do not allow analytic comparisons of different incomplete
networks. We instead focus on  a framework
with unboundedly informative Gaussian signals, where the log-linearity of actions yields a tractable measure of learning efficiency.

In a paper combining diffusion and social learning literatures, \citet{board2018learning} study a product adoption model where agents arrive at random times, observe network neighbors' adoption choices, 
and can pay  for a fully revealing private signal. Networks matter through different channels in their setting than the confounding mechanism that we focus on: indeed, networks that cause information confounding do not appear (or have vanishing probability) in \citet{board2018learning}. The main force is instead that agents infer from the absence of product adoptions, and this inference can depend on network structure.

Finally, a different strand of the literature 
examines different obstructions to efficient social learning in settings without network-based information loss. \citet*{harel2018rational} study
a social-learning environment with coarse communication and find,
as in our generations network, that agents learn at the same rate
as they would if they perfectly observed an arbitrarily small fraction
of private signals. The mechanism behind their result (``rational
groupthink'') is not related to an observation network preventing some agents from seeing others' social information, but rather relies on agents' finite action spaces obscuring all information about their private signals for many periods.\footnote{\cite*{huang2021learning} extend the results of \cite{harel2018rational} to give a uniform bound on the rate of learning across strongly connected networks. The obstruction to learning continues to be coarse actions and not network structure, however. Indeed, \cite{huang2021learning}'s result on general networks (which may introduce additional confounding) allows faster learning than the bound on the complete network from \cite{harel2018rational}.} Another group of papers study socially inefficient information acquisition when signals are endogenous \citep{burguet2000social,mueller2016social,ali2018herding,lomys2019collective,liang2019complementary}.
We assume rich action spaces and exogenous signals to abstract from
these obstructions and focus on the role of the network.

\section{Conclusion}

This paper presents a tractable model of sequential social learning
that lets us compare social-learning dynamics across different observation
networks. In our environment, rational actions  are a log-linear function of observations and admit a signal-counting interpretation. Thus, we can measure the efficiency of learning in terms of the fraction of available signals incorporated
into beliefs asymptotically (``aggregative efficiency'') and make precise comparisons about the
rate of learning and welfare across different networks. 

The network causes information confounding   when an agent does not see an early predecessor whose action influenced several of the agent's neighbors. We show that confounding can be a powerful obstacle to social learning:   even little confounding can cause almost total information loss. For a class of symmetric networks where agents move in generations,
we derive a simple expression for aggregative efficiency. For any network
in this class, social learning aggregates no more than two signals
per generation in the long run, even for arbitrarily large generations. We also compute comparative statics of learning with respect to network parameters, finding that additional observations speed
up learning but extra confounding slows it down.

We have focused on how the network structure affects social learning
and abstracted away from many other sources of learning-rate inefficiency.
These other sources may realistically co-exist with the informational-confounding
issues discussed here and complicate the analysis. For instance, even
though the complete network allows agents to exactly infer every predecessor's
private signal, it could lead to worse informational free-riding incentives
in settings where agents must pay for the precision of their private
signals (compared to networks where agents have fewer observations).
Studying the trade-offs and/or interactions between network-based
information confounding and other obstructions to fast learning could
lead to fruitful future work.

\bibliographystyle{ecta}
\bibliography{rational_sequential}

\newpage
\appendix
\begin{center}
\textbf{\Large{}Appendix}{\Large\par}
\par\end{center}

\section{Proofs of Results in the Main Text}\label{app:proofs}

\renewcommand{\theprop}{A.\arabic{prop}}  
\renewcommand{\thelem}{A.\arabic{lem}} 
\setcounter{prop}{0}
\setcounter{lem}{0}

\subsection{Details on Example \ref{ex:Many-Neighbors-with}}\label{subsec:details_ex}

We show that for the network in Figure \ref{fig:A-three-generation-network},
agent $K_{1}+K_{2}+1$'s rational log-strategy puts weight $\frac{1+K_{1}}{1+K_{1}K_{2}}$
on each neighbor's log-action, and hence $r_{K_{1}+K_{2}+1}=1+\frac{K_{2}+K_{1}K_{2}}{1+K_{1}K_{2}}\cdot(1+K_{1}).$

For $1\le j\le K_{2},$ we have $\ell_{K_{1}+j}=\sum_{i=1}^{K_{1}}\lambda_{i}+\lambda_{K_{1}+j}$.
So, $\mathbb{E}[\ell_{K_{1}+j}\mid\omega=1]=(K_{1}+1)\cdot\frac{2}{\sigma^{2}}$,
$\text{Var}[\ell_{K_{1}+j}\mid\omega=1]=\frac{4}{\sigma^{2}}(K_{1}+1)$,
while $\text{Cov}[\ell_{K_{1}+j},\ell_{K_{1}+j'}\mid\omega=1]=\frac{4}{\sigma^{2}}K_{1}$
for $1\le j<j'\le K_{2}$. By Proposition \ref{prop:linear}, the
vector of weights that the final agent's rational log-strategy
puts on neighbors' log actions is given by 
\[
2\cdot\frac{2}{\sigma^{2}}\cdot\left[\begin{array}{ccc}
K_{1}+1 & K_{1}+1 & \cdots\end{array}\right]\cdot\frac{\sigma^{2}}{4} \cdot \begin{array}{c}
\underset{K_{2}\text{ by }K_{2}}{\underbrace{\left[\begin{array}{cccc}
K_{1}+1 & K_{1} & \cdots & K_{1}\\
K_{1} & K_{1}+1 & \cdots & K_{1}\\
\vdots & \vdots & \ddots & \vdots\\
K_{1} & K_{1} & \cdots & K_{1}+1
\end{array}\right]}}\end{array}^{-1}.
\]

The matrix inverse is equal to 
\[
\frac{1}{K_{1}K_{2}+1}\underset{K_{2}\text{ by }K_{2}}{\underbrace{\left[\begin{array}{cccc}
(K_{2}-1)K_{1}+1 & -K_{1} & \cdots & -K_{1}\\
-K_{1} & (K_{2}-1)K_{1}+1 & \cdots & -K_{1}\\
\vdots & \vdots & \ddots & \vdots\\
-K_{1} & -K_{1} & \cdots & (K_{2}-1)K_{1}+1
\end{array}\right]}}.
\]
Therefore, weight on each neighbor is $\frac{1+K_{1}}{1+K_{1}K_{2}}.$
Also, since $\mathbb{E}[\ell_{K_{1}+j}\mid\omega=1]=(K_{1}+1)\cdot\frac{2}{\sigma^{2}}$
for each neighbor $K_{1}+j$ and there are $K_{2}$ neighbors, we
get $\mathbb{E}[\ell_{K_{1}+K_{2}+1}\mid\omega=1]=[1+\frac{K_{2}+K_{1}K_{2}}{1+K_{1}K_{2}}\cdot(1+K_{1})]\cdot\frac{2}{\sigma^{2}}$.
By the signal counting interpretation, $r_{K_{1}+K_{2}+1}=1+\frac{K_{2}+K_{1}K_{2}}{1+K_{1}K_{2}}\cdot(1+K_{1}).$

\subsection{Proof of Proposition \ref{prop:linear}}

We first prove a lemma about the conditional distributions of the
log-signals.
\begin{lem}
\label{lem:trick} For each $i,$ the log-signal $\lambda_{i}$ has
a Gaussian distribution conditional on $\omega$, with $\mathbb{E}[\lambda_{i}\mid\omega=0]=-2/\sigma^{2},$
$\mathbb{E}[\lambda_{i}\mid\omega=1]=2/\sigma^{2},$ and $\textsc{\emph{Var}}[\lambda_{i}\mid\omega=0]=\textsc{\emph{Var}}[\lambda_{i}\mid\omega=1]=4/\sigma^{2}.$
\begin{proof}
We show that $\lambda_{i}=\frac{2}{\sigma^{2}}s_{i}.$ This is because
\begin{align*}
\lambda_{i} & =\ln\left(\frac{\mathbb{P}[\omega=1|s_{i}]}{\mathbb{P}[\omega=0|s_{i}]}\right)=\ln\left(\frac{\mathbb{P}[s_{i}|\omega=1]}{\mathbb{P}[s_{i}|\omega=0]}\right)=\ln\left(\frac{\exp\left(\frac{-(s_{i}-1)^{2}}{2\sigma^{2}}\right)}{\exp\left(\frac{-(s_{i}+1)^{2}}{2\sigma^{2}}\right)}\right)\\
 & =\frac{-(s_{i}^{2}-2s_{i}+1)+(s_{i}^{2}+2s_{i}+1)}{2\sigma^{2}}=\frac{2}{\sigma^{2}}s_{i}.
\end{align*}
The result then follows from scaling the conditional distributions
of $s_{i},$ $(s_{i}\mid\omega=1)\sim\mathcal{N}(1,\sigma^{2})$ and
$(s_{i}\mid\omega=0)\sim\mathcal{N}(-1,\sigma^{2})$.
\end{proof}
\end{lem}
Now we prove Proposition \ref{prop:linear}.
\begin{proof}
Agent 1 does not observe any predecessors, so clearly ${S}_{1}^{*}(\lambda_{1})=\lambda_{1}.$
Suppose by way of induction that the rational strategies of all
agents $j\le I-1$ are linear. Then each $\ell_{j}$ for $j\le I-1$
is a linear combination of $(\lambda_{h})_{h=1}^{I-1},$ which by Lemma
\ref{lem:trick} are conditionally Gaussian with conditional means
$\pm2/\sigma^{2}$ in states $\omega=1$ and $\omega=0$ and conditional
variance $4/\sigma^{2}$ in each state. This implies $(\ell_{j(1)},...,\ell_{j(d_I)})$
have a conditional joint Gaussian distribution with $(\ell_{j(1)},...,\ell_{j(d_{I})})\sim\mathcal{N}(\vec{\mu},\Sigma)$
conditional on $\omega=1$, and  $(\ell_{j(1)},...,\ell_{j(d_I)})\sim\mathcal{N}(-\vec{\mu},\Sigma)$
conditional on $\omega=0$, where $\vec{\mu}=\mathbb{E}[(\ell_{j(1)},...,\ell_{j(d_{I})})^{\prime}\mid\omega=1]$
and $\Sigma=\textsc{Cov}[\ell_{j(1)},...,\ell_{j(d_{I})}\mid\omega=1]$.

From the  multivariate Gaussian density, (writing $(\ell_{j(1)},...,\ell_{j(d_{I})})^{\prime}=\vec{a})$,
\begin{align*}
\ln\left(\frac{\mathbb{P}[\ell_{j(1)},...,\ell_{j(d_{I})}\mid\omega=1]}{\mathbb{P}[\ell_{j(1)},...,\ell_{j(d_{I})}\mid\omega=0]}\right) & =\ln\left(\frac{\exp(-\frac{1}{2}(\vec{a}-\vec{\mu})^{\prime}\Sigma^{-1}(\vec{a}-\vec{\mu}))}{\exp(-\frac{1}{2}(\vec{a}+\vec{\mu})^{\prime}\Sigma^{-1}(\vec{a}+\vec{\mu}))}\right)\\
 & =\vec{a}^{\prime}\Sigma^{-1}\vec{\mu}+\vec{\mu}^{\prime}\Sigma^{-1}\vec{a}
\end{align*}
which is $2\left(\vec{\mu}^{\prime}\Sigma^{-1}\right)\cdot(\ell_{j(1)},...,\ell_{j(d_{I})})^{\prime}$
because $\Sigma$ is symmetric. This then shows agent $I$'s rational
strategy must also be linear, completing the inductive step. This
argument also gives the explicit form of $\vec{\beta}_{I,\cdot}$.

For the final statement, we prove another lemma. The argument so far
implies that we may find weights $(w_{i,j})_{j\le i}$ so that the
realizations of rational log-actions are related to the realizations
of log-signals by $\ell_{i}=\sum_{j=1}^{i}w_{i,j}\lambda_{j}$. Let
$W$ be the matrix containing all such weights.
\begin{lem}
\label{lem:beta_and_W} Let $\hat{W}$ be the submatrix of $W$ with
rows $N(i)$ and columns $\{1,...,i-1\}.$ Then $\vec{\beta}_{i}=\vec{\boldsymbol{1}}_{(i-1)}^{\prime}\times\hat{W}^{\prime}(\hat{W}\hat{W}^{\prime})^{-1}$
and the $i$-th row of $W$ is $W_{i}=\left((\vec{\beta}_{i,\cdot}^{\prime}\times\hat{W}),1,0,0,...\right)$.
\end{lem}
\begin{proof}
Suppose $N(i)=\{j(1),...,j(d_{i})\}$ with $j(1)<...<j(d_{i}).$ By
Lemma \ref{lem:trick} and construction of $\hat{W}$, we have $\mathbb{E}[\ell_{j(k)}\mid\omega=1]=\frac{2}{\sigma^{2}}\sum_{h=1}^{i-1}\hat{W}_{k,h}$.
So, $\mathbb{E}[(\ell_{j(1)},...,\ell_{j(d_{i})})\mid\omega=1]=\frac{2}{\sigma^{2}}(\hat{W}\cdot\boldsymbol{1}_{(i-1)})^{\prime}=\frac{2}{\sigma^{2}}\boldsymbol{1}_{(i-1)}^{\prime}\hat{W}{}^{\prime}.$
Also, again by Lemma \ref{lem:trick} and construction of $\hat{W}$,
we can calculate that for $1\le k_{1}\le k_{2}\le d_{i},$ $\textsc{Cov}[\ell_{j(k_{1})},\ell_{j(k_{2})}\mid\omega=1]=\frac{4}{\sigma^{2}}\sum_{h=1}^{i-1}(\hat{W}_{k_{1},h}\hat{W}_{k_{2},h}),$
meaning $\textsc{Cov}[\ell_{j(1)},...,\ell_{j(d_{i})}\mid\omega=1]=\frac{4}{\sigma^{2}}\hat{W}\hat{W}^{\prime}.$
It then follows from what we have shown above that $\vec{\beta}_{i,\cdot}=2\cdot\frac{2}{\sigma^{2}}\boldsymbol{1}_{(i-1)}^{\prime}\hat{W}{}^{\prime}\times\left[\frac{4}{\sigma^{2}}\hat{W}\hat{W}^{\prime}\right]^{-1}=\vec{\boldsymbol{1}}_{(i-1)}^{\prime}\times\hat{W}^{\prime}(\hat{W}\hat{W}^{\prime})^{-1}.$

Since $i$ puts weight 1 on $\lambda_{i}$ and weights $\vec{\beta}_{i,\cdot}$
on $(\ell_{j(1)},...,\ell_{j(d_{i})})^{\prime}=\hat{W}\times(\lambda_{1},...,\lambda_{i-1})^{\prime},$
this shows the first $i-1$ elements in the row $W_{i}$ must be $\vec{\beta}_{i,\cdot}^{\prime}\cdot\hat{W}$
while the $i$-th element is 1.
\end{proof}
To prove the final statement of Proposition \ref{prop:linear}, first
observe that $W_{1}=(1,0,0,...)$ does not depend on $\sigma^{2}.$
The same applies to $\vec{\beta}_{1,\cdot}$. By way of induction,
suppose rows $W_{i}$ and vectors $\vec{\beta}_{i,\cdot}$ do not
depend on $\sigma^{2}$ for any $i\le I$. If $\hat{W}$ is the submatrix
of $W$ with rows $N(I+1)$, then since $N(I+1)\subseteq\{1,...,I\},$
by the inductive hypothesis $\hat{W}$ must be independent of $\sigma^{2}$.
Thus the same independence also applies to $\vec{\beta}_{I+1,\cdot}$
since this vector only depends on $\hat{W}$ by the result just derived.
In turn, since $W_{I+1}$ is only a function of $\vec{\beta}_{I+1,\cdot}^{\prime}$
and $\hat{W}$, and these terms are independent of $\sigma^{2}$ as
argued before, same goes for $W_{I+1},$ completing the inductive
step.
\end{proof}

\subsection{Proof of Proposition \ref{prop:signal_counting}}
\begin{proof}
It suffices to show that $\mathbb{E}[\ell_{i}\mid\omega=1]=\frac{1}{2}\textsc{Var}\left[\ell_{i}\mid\omega=1\right]$.
By Proposition \ref{prop:linear}, $\ell_{i}=\lambda_{i}+\sum_{k=1}^{d_{i}}\beta_{i,j(k)}\ell_{j(k)}$.
From Lemma \ref{lem:trick}, we have $\mathbb{E}[\lambda_{i}\mid\omega=1]=\frac{1}{2}\textsc{Var}\left[\lambda_{i}\mid\omega=1\right]$.
Furthermore, $\lambda_{i}$ is independent from $\sum_{k=1}^{d_{i}}\beta_{i,j(k)}\ell_{j(k)},$
as the latter term only depends on $\lambda_{1},...,\lambda_{i-1}.$
So we need only show $\mathbb{E}[\sum_{k=1}^{d_{i}}\beta_{i,j(k)}\ell_{j(k)}\mid\omega=1]=\frac{1}{2}\textsc{Var}\left[\sum_{k=1}^{d_{i}}\beta_{i,j(k)}\ell_{j(k)}\mid\omega=1\right]$

Let $\vec{\mu}=\mathbb{E}[(\ell_{j(1)},...,\ell_{j(d_{i})})^{\prime}\mid\omega=1]$
and $\Sigma=\textsc{Cov}[\ell_{j(1)},...,\ell_{j(d_{i})}\mid\omega=1]$.
Using the expression for $\vec{\beta}_{i,\cdot}$ from Proposition
\ref{prop:linear}, $\mathbb{E}\left[\sum_{k=1}^{d_{i}}\beta_{i,j(k)}\ell_{j(k)}\mid\omega=1\right]=2\left(\vec{\mu}^{\prime}\Sigma^{-1}\right)\cdot\vec{\mu}.$
Also, 
\begin{align*}
\textsc{Var}\left[\sum_{k=1}^{d_{i}}\beta_{i,j(k)}\ell_{j(k)}\mid\omega=1\right] & =\left(2\vec{\mu}^{\prime}\Sigma^{-1}\right)\Sigma\left(2\vec{\mu}^{\prime}\Sigma^{-1}\right)^{\prime}=4\vec{\mu}^{\prime}\Sigma^{-1}\vec{\mu}
\end{align*}
 using the fact that $\Sigma$ is a symmetric matrix. This is twice
$\mathbb{E}\left[\sum_{k=1}^{d_{i}}\beta_{i,j(k)}\ell_{j(k)}\mid\omega=1\right]$
as desired.
\end{proof}

\subsection{Proof of Proposition \ref{prop:long_run_conditions}}

We first state and prove an auxiliary lemma.
\begin{lem}
\label{lem:accuracy} For any $0<\epsilon<0.5,$ $\mathbb{P}[a_{i}>1-\epsilon\mid\omega=1]=1-\Phi\left(\frac{\ln\left(\frac{1-\epsilon}{\epsilon}\right)-r_{i}\frac{2}{\sigma^{2}}}{\sqrt{r_{i}}\frac{2}{\sigma}}\right),$where
$\Phi$ is the standard Gaussian distribution function. This expression
is increasing in $r_{i}$ and approaches 1. Also, $\mathbb{P}[a_{i}<\epsilon\mid\omega=0]=\Phi\left(\frac{\ln\left(\frac{1-\epsilon}{\epsilon}\right)+r_{i}\frac{2}{\sigma^{2}}}{\sqrt{r_{i}}\frac{2}{\sigma}}\right).$
This expression is increasing in $r_{i}$ and approaches 1.
\end{lem}
\begin{proof}
Note that $a_{i}>1-\epsilon$ if and only if $\ell_{i}>\ln\left(\frac{1-\epsilon}{\epsilon}\right)>0.$
Given that $(\ell_{i}\mid\omega=1)\sim\mathcal{N}\left(r_{i}\cdot\frac{2}{\sigma^{2}},r_{i}\cdot\frac{4}{\sigma^{2}}\right)$
by Proposition \ref{prop:signal_counting}, the expression for $\mathbb{P}[a_{i}>1-\epsilon\mid\omega=1]$
follows. To see that it is increasing in $r_{i}$, observe that $\frac{d}{dr_{i}}\frac{\ln\left(\frac{1-\epsilon}{\epsilon}\right)-r_{i}\frac{2}{\sigma^{2}}}{\sqrt{r_{i}}\frac{2}{\sigma}}$
has the same sign as 
\[
\frac{-2}{\sigma^{2}}(\sqrt{r_{i}}\frac{2}{\sigma})-(\ln\left(\frac{1-\epsilon}{\epsilon}\right)-r_{i}\frac{2}{\sigma^{2}})(\frac{1}{2}r_{i}^{-0.5}\frac{2}{\sigma})=-\frac{2}{\sigma^{3}}\sqrt{r_{i}}-\ln\left(\frac{1-\epsilon}{\epsilon}\right)r_{i}^{-0.5}\frac{1}{\sigma}<0.
\]
 Also, it is clear that $\lim_{r_{i}\to\infty}\frac{\ln\left(\frac{1-\epsilon}{\epsilon}\right)-r_{i}\frac{2}{\sigma^{2}}}{\sqrt{r_{i}}\frac{2}{\sigma}}=-\infty$,
hence $\lim_{r_{i}\to\infty}\mathbb{P}[a_{i}>1-\epsilon\mid\omega=1]=1.$
The results for $\mathbb{P}[a_{i}<\epsilon\mid\omega=0]$ follow from
analogous arguments.
\end{proof}
We now turn to the proof of Proposition \ref{prop:long_run_conditions}.
\begin{proof}
By Proposition \ref{prop:signal_counting}, there exist $(r_{i})_{i\ge1}$
so that social learning aggregates $r_{i}$ signals by agent $i.$
We first show that (3) and (4) in Proposition \ref{prop:long_run_conditions}
are equivalent. Let $\text{\ensuremath{\epsilon^{'}>0} }$ be given
and suppose $\lim_{i\to\infty}r_{i}=\infty.$ Putting $\epsilon=\min(\epsilon^{'},0.4),$
we get that $\mathbb{P}[|a_{i}-\omega|<\epsilon\mid\omega=1]\to1$
and $\mathbb{P}[|a_{i}-\omega|<\epsilon\mid\omega=0]\to1$ since the
two expressions in Lemma \ref{lem:accuracy} increase in $r_{i}$
and approach 1, hence also $\mathbb{P}[|a_{i}-\omega|<\epsilon^{'}]\to1.$
So society learns completely in the long run. Conversely, if we do
not have $\lim_{i\to\infty}r_{i}=\infty$, then for some $K<\infty$
we have $r_{i}<K$ for infinitely many $i$. By Lemma \ref{lem:accuracy}
we will get that $\mathbb{P}[|a_{i}-\omega|<0.1\mid\omega=1]$ are
bounded by $1-\Phi\left(\frac{\ln\left(9\right)-K\frac{2}{\sigma^{2}}}{\sqrt{K}\frac{2}{\sigma}}\right)$
for these $i,$ hence society does not learn completely in the long
run.

Next, we show that Conditions (1) and (2) in the proposition are both
equivalent to Condition (3), $\lim_{i\to\infty}r_{i}=\infty.$

\textbf{Condition (1)}: $\lim_{{\it i\to\infty}}\mathcal{PL}(i)=\infty$.

\emph{Necessity}: Suppose $\lim_{i\to\infty}r_{i}=\infty.$ For $h\in\mathbb{N},$
let $I(h):=\{i:\mathcal{PL}(i)=h\}.$ We show by induction that $I(h)$
is finite for all $h\in\mathbb{N}$. For every $i\in I(0),$ $r_{i}=1,$
so $\lim_{i\to\infty}r_{i}=\infty$ implies $|I(0)|<\infty.$ Now
suppose $|I(h)|<\infty$ for all $h\le L.$ If $i\in I(L+1),$ then
every $j$ that can be reached along $M$ from $i$ must belong to
$I(h)$ for some $h\le L.$ The subnetwork containing $i$ is therefore
a subset of $\cup_{h=0}^{L}I(h)$, a finite set by the inductive hypothesis.
Thus $r_{i}\le1+\sum_{h=0}^{L}|I(h)|$ for all $i\in I(L+1).$ So
$\lim_{i\to\infty}r_{i}=\infty$ implies $I(L+1)$ is finite, completing
the inductive step and proving $I(h)$ is finite for all $h$. Hence
$\lim_{{\it i\to\infty}}\mathcal{PL}(i)=\infty.$

\emph{Sufficiency}: First note if $j\in N(i),$ then $r_{i}\ge r_{j}+1.$
This is because  $\ell_{j}\sim\mathcal{N}\left(\pm r_{j}\cdot\frac{2}{\sigma^{2}},r_{j}\cdot\frac{4}{\sigma^{2}}\right)$
conditional on the two states, and furthermore $\ell_{j}$ is conditionally
independent of $s_{i}.$ So, $\ell_{j}+\lambda_{i}$ is a possible
play for $i,$ which would have the conditional distributions $\mathcal{N}\left(\pm(r_{j}+1)\cdot\frac{2}{\sigma^{2}},(r_{j}+1)\cdot\frac{4}{\sigma^{2}}\right)$
in the two states. If $r_{i}<r_{j}+1,$ then $i$ would have a profitable
deviation by choosing $\ell_{i}=\ell_{j}+\lambda_{i}$ instead, since
it follows from Lemma \ref{lem:accuracy} that a log-action that aggregates
more signals leads to higher expected payoffs.

\textbf{Condition (2)}: $\lim_{i\to\infty}\left[\max_{j\in N(i)}j\right]=\infty.$

\emph{Necessity}: If Condition (2) is violated, there exists some
$\bar{j}<\infty$ so that there exist infinitely many $i$'s with
$N(i)\subseteq\{1,...,\bar{j}\}.$ The subnetwork containing any such
$i$ is a subset of $\{1,...,\bar{j}\},$ so $r_{i}\le\bar{j}+1$.
We cannot have $\lim_{i\to\infty}r_{i}=\infty.$

\emph{Sufficiency}: Construct an increasing sequence $C_{1}\le C_{2}\le...$
as follows. Condition (2) implies there exists $C_{1}$ so that $\max_{j\in N(i)}j\ge1$
for all $i\ge C_{1}.$ So, $\mathcal{PL}(i)\ge1$ for all $i\ge C_{1}.$
Suppose $C_{1}\le...\le C_{n}$ are constructed with the property
that $\mathcal{PL}(i)\ge k$ for all $i\ge C_{k}$, $k=1,...,n.$
Condition (2) implies there exists $C_{n+1}$ so that $\max_{j\in N(i)}j\ge C_{n}$
for all $i\ge C_{n+1}.$ But since all $j\ge C_{n}$ have $\mathcal{PL}(j)\ge n$
by the inductive hypothesis, all $i\ge C_{n+1}$ must have $\mathcal{PL}(i)\ge n+1,$
completing the inductive step. This shows $\lim_{{\it i\to\infty}}\mathcal{PL}(i)=\infty$.
By the sufficiency of Condition (1) for $\lim_{i\to\infty}r_{i}=\infty$,
we see that Condition (2) implies the same.
\end{proof}

\subsection{Proof of Theorem \ref{thm:efficiency_dc}}
\begin{proof}
If $d=1,$ then exactly one signal is aggregated per generation so
$r_{i}/K\rightarrow1$ as required. Also, if $c=0$, then we must
have $d=1.$ From now on we assume $d\ge2$ and $c\ge1.$
\begin{lem}
\label{lem:anon_beta}For $d\ge2,$ each generation $t$ and each
$i\ne i'$ in generation $t$, $\textsc{\emph{Var}}\left[\ell_{i}\mid\omega=1\right]$
and $\textsc{\emph{Cov}}\left[\ell_{i},\ell_{i'}\mid\omega=1\right]$
depend only on $t$ and not on the identities of $i$ or $i'$, which
we call $\textsc{\emph{Var}}_{t}$ and $\textsc{\emph{Cov}}_{t}$,
respectively. Similarly, for $i$ in generation $t$ and each $j\in N(i)$,
the weight $\beta_{i,j}$ depends only on $t$, which we call $\beta_{t}$.
\begin{proof}
The results hold by inductively applying the symmetry condition. Clearly
they are true for $t=2.$ Suppose they are true for all $t\le T$.
For an agent $i$ in generation $t=T+1$, the inductive hypothesis
implies $\textsc{Var}[\ell_{j}\mid\omega=1]$ is the same for all
$j\in N(i),$ and all pairs $j,j^{'}\in N(i)$ with $j\ne j^{'}$
have the same conditional covariance. Also, using Proposition \ref{prop:signal_counting},
$\mathbb{E}[\ell_{j}\mid\omega=1]$ is the same for all $j\in N(i)$.
Thus by Proposition \ref{prop:linear}, $i$ places the same weight,
say $\beta_{t}$, on all neighbors. Using the fact that $\ell_{i}=\lambda_{i}+\sum_{j\in N(i)}\beta_{t}\ell_{j}$,
we have the recursive expressions $\textsc{Var}[\ell_{i}\mid\omega=1]=\frac{4}{\sigma^{2}}+\beta_{t}^{2}(d\textsc{Var}_{t-1}+(d^{2}-d)\textsc{Cov}_{t-1})$
for all $i$ in generation $t,$ and $\textsc{Cov}[\ell_{i},\ell_{i^{'}}\mid\omega=1]=\beta_{t}^{2}(c\textsc{Var}_{t-1}+(d^{2}-c)\textsc{Cov}_{t-1})$
for all agents $i\ne i^{'}$ in generation $t.$ This shows the claims
for $t=T+1$, and completes the proof by induction.
\end{proof}
\end{lem}
Taking the difference of the two expressions for $\textsc{Var}_{t}$
and $\textsc{Cov}_{t}$ gives:

\begin{equation}
\textsc{Var}_{t}-\textsc{Cov}_{t}=\frac{4}{\sigma^{2}}+\beta_{t}^{2}(d-c)(\textsc{Var}_{t-1}-\textsc{Cov}_{t-1}).\label{eq:varcovdiff}
\end{equation}

We now require two auxiliary lemmas.
\begin{lem}
\label{lem:limit_distr}Consider the Markov chain on $\{1,...,K\}$
with state transition matrix $p,$ with $p_{i,j}=\mathbb{P}[i\to j]=1/d$
if $j\in\Psi_{i},$ 0 otherwise. Suppose $(\Psi_{k})_{k}$ is symmetric
with $c\ge1$. Then $p_{i}^{\infty}:=\lim_{t\to\infty}(p^{t})_{i}\in[0,1]^{K}$
exists, and it is the same for all $1\le i\le K.$
\end{lem}
\begin{proof}
For existence of $p_{i}^{\infty},$ consider the decomposition of
the Markov chain into its communication classes, $C_{1},...,C_{L}\subseteq\{1,...,K\}$.
Without loss suppose the first $L^{'}$ communication classes are
closed and the rest are not.

We show that each closed communication class is aperiodic when $(\Psi_{k})_{k}$
is symmetric and $c,d\ge1$. Let $i\in C_{m}$ for $1\le m\le L^{'}.$
Let $\Psi_{i}=\{j_{1},...,j_{d}\}$. If $i\in\Psi_{i},$ then $i$'s
periodicity is 1. Otherwise, $\Psi_{i}\subseteq C_{m}$ since $C_{m}$
is closed, so for every $1\le h\le d$ there exists a cycle of some
length $Q_{h}$ starting at $i$, where the $h$-th such cycle is
$i\to j_{h}\to...\to i.$ Since $c\ge1,$ $i$ and $j_{1}$ share
a common neighbor, which must be $j_{h^{*}}$ for some $1\le h^{*}\le d.$
We can therefore construct a cycle of length $Q_{h^{*}}+1$ starting
at $i,$ $i\to j_{1}\to j_{h^{*}}\to...\to i$. Since cycle lengths
$Q_{h^{*}}$ and $Q_{h^{*}}+1$ are coprime, $i$'s periodicity is
1.

By standard results (see e.g., \citet{billingsley2013convergence})
there exist $\nu_{m}^{*},1\le m\le L^{'},$ so that $\lim_{t\to\infty}(p^{t})_{i}=\nu_{m}^{*}$
whenever $i\in C_{m}$. If $i\notin\cup_{1\le m\le L^{'}}C_{m}$,
then starting the process at $i,$ almost surely the process enters
one of the closed communication classes eventually. This shows $\lim_{t\to\infty}(p^{t})_{i}$
exists and is equal to $\sum_{m=1}^{L^{'}}q_{m}\nu_{m}^{*}$, where
$q_{m}$ is the probability that the process started at $i$ enters
$C_{m}$ before any other closed communication class.

To prove that $p_{i}^{\infty}$ is the same for all $i,$ we inductively
show that for all $i\ne j,$ $\parallel p_{i}^{\infty}-p_{j}^{\infty}\parallel_{\max}\le\left(\frac{d-c}{d}\right)^{t}$
for all $t\ge1.$ Since $c\ge1,$ this would show that in fact $p_{i}^{\infty}=p_{j}^{\infty}$
for all $i,j.$

For the base case of $t=1,$ enumerate $\Psi_{i}=\{n_{1},...,n_{c},n_{c+1},...,n_{d}\},$
$\Psi_{j}=\{n_{1},...,n_{c},n_{c+1}^{'},...,n_{d}^{'}\}$ where all
$n_{1},...,n_{d},n_{c+1}^{'},...,n_{d}^{'}\in\{1,...,K\}$ are distinct.
Then 
\[
p_{i}^{\infty}=\frac{1}{d}\left(\sum_{k=1}^{c}p_{n_{k}}^{\infty}\right)+\frac{1}{d}\left(\sum_{k=c+1}^{d}p_{n_{k}}^{\infty}\right),
\]
\[
p_{j}^{\infty}=\frac{1}{d}\left(\sum_{k=1}^{c}p_{n_{k}}^{\infty}\right)+\frac{1}{d}\left(\sum_{k=c+1}^{d}p_{n_{k}^{'}}^{\infty}\right),\text{ so}
\]
\begin{align*}
\parallel p_{i}^{\infty}-p_{j}^{\infty}\parallel_{\max} & \le\frac{1}{d}\sum_{k=c+1}^{d}\parallel p_{n_{k}}^{\infty}-p_{n_{k}^{'}}^{\infty}\parallel_{\max}\le\frac{d-c}{d}\cdot1
\end{align*}
 where the 1 comes from $\parallel x-y\parallel_{\max}\le1$ for any
two distributions $x,y$.

The inductive step just replaces the bound $\parallel x-y\parallel_{\max}\le1$
with 
\[
\parallel p_{n_{k}}^{\infty}-p_{n_{k}^{'}}^{\infty}\parallel_{\max}\le\left(\frac{d-c}{d}\right)^{t-1}
\]
 from the inductive hypothesis.
\end{proof}
\begin{lem}
\label{lem:beta_1/d}$\beta_{t}\rightarrow1/d$.
\end{lem}
\begin{proof}
For $i$ in generation $t+1,$ $\ell_{i}=\lambda_{i}+\beta_{t+1}\sum_{j\in N(i)}\ell_{j}$,
so as in the proof of Lemma \ref{lem:anon_beta}, $\textsc{Var}[\ell_{i}\mid\omega=1]=\frac{4}{\sigma^{2}}+\beta_{t+1}^{2}(d\textsc{Var}_{t}+(d^{2}-d)\textsc{Cov}_{t})$.
Using the definition of  the signal-counting interpretation and Proposition
\ref{prop:signal_counting}, $\mathbb{E}[\ell_{j}\mid\omega=1]=\frac{1}{2}\textsc{Var}_{t}$
for each $j\in N(i),$ and so $\mathbb{E}[\ell_{i}\mid\omega=1]=\frac{2}{\sigma^{2}}+d\beta_{t+1}(\frac{1}{2}\textsc{Var}_{t})$.
By the same argument we also have $\textsc{Var}[\ell_{i}\mid\omega=1]=2\cdot\mathbb{E}[\ell_{i}\mid\omega=1],$
and this lets us solve out

\[
\beta_{t+1}=\frac{\textsc{Var}_{t}}{\textsc{Var}_{t}+(d-1)\textsc{Cov}_{t}}\geq\frac{1}{d}.
\]

It is therefore sufficient to show that $\textsc{Var}_{t}/\textsc{Cov}_{t}\rightarrow1$.
The weight $w_{i,i^{'}}$ that an agent $i$ in generation $t$ places
on the private signal of an agent $i^{'}$ in generation $t-\tau$
is equal to the product of $\prod_{j=1}^{\tau}\beta_{t+1-j}$ and
the number of paths from $i$ to $i^{'}$ in the network $M.$

We can compute the number of paths as follows. Consider a Markov chain
with states $\left\{ 1,\ldots,K\right\} $ and state transition probabilities
$\mathbb{P}[k_{1}\to k_{2}]=1/d$ if $k_{2}\in\Psi_{k_{1}},$ $\mathbb{P}[k_{1}\to k_{2}]=0$.
The number of paths from $i$ in generation $t$ to $j$ in generation
$t-\tau$ is equal to $d^{\tau}$ times the probability that the state
is $j$ after $\tau$ periods.

By Lemma \ref{lem:limit_distr}, there exists a stationary distribution
$\pi^{*}\in\mathbb{R}_{+}^{K}$ with $\sum_{k=1}^{K}\pi_{k}^{*}=1$
of the Markov chain. Given $\epsilon>0$, we can choose $\tau_{0}$
such that the number of paths from $i$ in generation $t$ to $j=(t-\tau-1)K+k$
in generation $t-\tau$ is in {[}$d^{\tau}(\pi_{k}^{*}-\epsilon),d^{\tau}(\pi_{k}^{*}+\epsilon)]$
for all $t$ and all $\tau\geq\tau_{0}$. (This number is zero for $k$ such that $\pi_k^* =0$.)

Fixing distinct agents $i$ and $i'$ in generation $t$:
\[
\textsc{Var}_{t}=\frac{4}{\sigma^{2}}+\frac{4}{\sigma^{2}}\sum_{\tau=2}^{t}\sum_{k=1}^{K}w_{i,(t-\tau)K+k}^{2}\text{ and }\textsc{Cov}_{t}=\frac{4}{\sigma^{2}}\sum_{\tau=2}^{t}\sum_{k=1}^{K}w_{i,(t-\tau)K+k}w{}_{i^{'},(t-\tau)K+k}.
\]

We want to show that
\[
\textsc{Var}_{t}/\textsc{Cov}_{t}=\frac{1+\sum_{\tau=2}^{t}\sum_{k=1}^{K}w_{i,(t-\tau)K+k}^{2}}{\sum_{\tau=2}^{t}\sum_{k=1}^{K}w_{i,(t-\tau)K+k}w{}_{i^{'},(t-\tau)K+k}}\rightarrow1.
\]
Take $\epsilon>0$ smaller than $\pi_{k}^{*}$ for all $k$ such that $\pi^*_k > 0$. For $\tau\geq\tau_{0}$ and any $k$ such that  $\pi^*_k > 0$,
we have 
\[
w_{i,(t-\tau)K+k}w{}_{i^{'},(t-\tau)K+k}\geq(d^{\tau}\prod_{j=1}^{\tau}\beta_{t+1-j})^{2}(\pi_{k}^{*}-\epsilon)^{2}\text{ and }w_{i,(t-\tau)K+k}^{2}\leq(d^{\tau}\prod_{j=1}^{\tau}\beta_{t+1-j})^{2}(\pi_{k}^{*}+\epsilon)^{2}
\]
The covariance grows at least linearly in $t$ since each $\beta\ge1/d$,
while the contribution from periods $t-\tau_{0}+1,\ldots,t$ is bounded
and therefore lower order. Thus,
\[
\limsup_{t\rightarrow\infty}\textsc{Var}_{t}/\textsc{Cov}_{t}\leq\limsup_{t\rightarrow\infty}\frac{\sum_{k=1}^{K}\sum_{\tau=\tau_{0}}^{t-1}(d^{\tau}\prod_{j=1}^{\tau}\beta_{t+1-j})^{2}(\pi_{k}^{*}+\epsilon)^{2}}{\sum_{k=1}^{K}\sum_{\tau=\tau_{0}}^{t-1}(d^{\tau}\prod_{j=1}^{\tau}\beta_{t+1-j})^{2}(\pi_{k}^{*}-\epsilon)^{2}}\le\max_{1\le k\le K}\frac{(\pi_{k}^{*}+\epsilon)^{2}}{(\pi_{k}^{*}-\epsilon)^{2}}.
\]
Since $\epsilon$ is arbitrary, this completes the proof of the lemma.
\end{proof}
We return to the proof of Theorem \ref{thm:efficiency_dc}. Fix small
$\epsilon>0$. By Lemma \ref{lem:beta_1/d}, we can choose $T$ such
that $\beta_{t}\leq\frac{1+\epsilon}{d}$ for all $t\geq T$. Therefore,
$\beta_{t}^{2}(d-c)\leq\frac{(1+\epsilon)^{2}}{d^{2}}(d-c)$ for $t\ge T$.
Consider the contraction map $\varphi(x)=\frac{4}{\sigma^{2}}+\frac{(1+\epsilon)^{2}}{d^{2}}(d-c)x$.
Iterating Equation (\ref{eq:varcovdiff}) starting with $t=T$, we
find that $\textsc{Var}_{t}-\textsc{Cov}_{t}\le\varphi^{(t-T)}(\textsc{Var}_{T}-\textsc{Cov}_{T}),$
so this shows

\[
\limsup_{t\to\infty}\left(\textsc{Var}_{t}-\textsc{Cov}_{t}\right)\le\frac{4}{\sigma^{2}}\cdot\frac{d^{2}}{d^{2}-(1+\epsilon)^{2}d+(1+\epsilon)^{2}c}
\]
 where the RHS is the fixed point of $\varphi.$ Since this holds
for all small $\epsilon>0,$ we get $\limsup_{t\to\infty}(\textsc{Var}_{t}-\textsc{Cov}_{t})\le\frac{4}{\sigma^{2}}\frac{d^{2}}{d^{2}-d+c}$.

At the same time, $\beta_{t}\ge\frac{1}{d}$ for all $t.$ Consider
the contraction map $\varphi(x)=\frac{4}{\sigma^{2}}+\frac{1}{d^{2}}(d-c)x$.
Iterating Equation (\ref{eq:varcovdiff}) starting with $t=1$, we
find that $\textsc{Var}_{t}-\textsc{Cov}_{t}\ge\varphi^{(t-1)}(\textsc{Var}_{1}-\textsc{Cov}_{1}),$
so this shows 
\[
\liminf_{t\to\infty}\left(\textsc{Var}_{t}-\textsc{Cov}_{t}\right)\ge\frac{4}{\sigma^{2}}\cdot\frac{d^{2}}{d^{2}-d+c}
\]
 where the RHS is the fixed point of $\varphi.$ Combining with the
result before, we get $\lim_{t\to\infty}(\textsc{Var}_{t}-\textsc{Cov}_{t})=\frac{4}{\sigma^{2}}\cdot\frac{d^{2}}{d^{2}-d+c}.$

As in the proof of Lemma \ref{lem:beta_1/d}, for $i$ in generation
$t+1$, $\mathbb{E}[\ell_{i}\mid\omega=1]=\frac{2}{\sigma^{2}}+d\beta_{t+1}(\frac{1}{2}\textsc{Var}_{t})$.
Using the definition of signal-counting interpretation and Proposition
\ref{prop:signal_counting}, we have $\textsc{Var}_{t+1}=2\cdot\mathbb{E}[\ell_{i}\mid\omega=1]=2(\beta_{t+1}d(\textsc{Var}_{t}/2)+2/\sigma^{2})$,
so

\begin{eqnarray*}
\textsc{Var}_{t+1}-\textsc{Var}_{t} & = & (\beta_{t+1}d-1)\textsc{Var}_{t}+\frac{4}{\sigma^{2}}\\
 & = & \left(\frac{d\textsc{Var}_{t}}{\textsc{Var}_{t}+(d-1)\textsc{Cov}_{t}}-1\right)\textsc{Var}_{t}+\frac{4}{\sigma^{2}}\\
 & = & \left(\frac{d\textsc{Var}_{t}}{d\textsc{Var}_{t}-(d-1)(\textsc{Var}_{t}-\textsc{Cov}_{t})}-1\right)\textsc{Var}_{t}+\frac{4}{\sigma^{2}}
\end{eqnarray*}

Using $\lim_{t\to\infty}(\textsc{Var}_{t}-\textsc{Cov}_{t})=\frac{4}{\sigma^{2}}\cdot\frac{d^{2}}{d^{2}-d+c}$,
we conclude 
\begin{align*}
\lim_{t\to\infty}\left(\textsc{Var}_{t+1}-\textsc{Var}_{t}\right) & =\lim_{t\to\infty}\left(\frac{\textsc{Var}_{t}}{\textsc{Var}_{t}-\frac{4}{\sigma^{2}}\frac{d^{2}-d}{d^{2}-d+c}}-1\right)\textsc{Var}_{t}+\frac{4}{\sigma^{2}}.\\
 & =\lim_{t\to\infty}\left(\frac{4}{\sigma^{2}}\frac{d^{2}-d}{d^{2}-d+c}\cdot\frac{\textsc{Var}_{t}}{\textsc{Var}_{t}-\frac{4}{\sigma^{2}}\frac{d^{2}-d}{d^{2}-d+c}}\right)+\frac{4}{\sigma^{2}}
\end{align*}
 Since $\textsc{Var}_{t}\rightarrow\infty$, the asymptotic increase
in conditional variance across successive generations is $\lim_{t\to\infty}\left(\textsc{Var}_{t+1}-\textsc{Var}_{t}\right)=\frac{4}{\sigma^{2}}\left(\frac{d^{2}-d}{d^{2}-d+c}+1\right)$.
Since agent $i$ is in generation $\left\lfloor i/K\right\rfloor $,
we therefore have $r_{i}=\left(1+\frac{d^{2}-d}{d^{2}-d+c}\right)\frac{i}{K}+o(i).$
So $\lim_{i\to\infty}(r_{i}/i)=\left(1+\frac{d^{2}-d}{d^{2}-d+c}\right)\frac{1}{K}.$
\end{proof}

% \subsection{Proof of Corollary \ref{cor:r_for_gens}}
% \begin{proof}
% The expression for $r_{i}$ comes from specializing Theorem \ref{thm:efficiency_dc}
% to the case of $d=c=K$. Observe $\frac{(2K-1)}{K^{2}}\cdot K=(2K-1)/K<2$
% for any $K\ge1.$
% \end{proof}

\subsection{Proof of Corollary \ref{cor:cancellation}}
\begin{proof}
When $d\ge2$ and $c<d,$ the collection of symmetric observation
sets with these parameters correspond to the collection of symmetric
balanced incomplete block designs by Theorem 2.2 from Chapter 8 of
\citet{ryser1963combinatorial}. If there exists at least one symmetric
network with parameters $(d,c,K)$ under the previous inequalities,
then $K=\frac{d^{2}-d+c}{c}$ by Equation (3.17) from Chapter 8 of
\citet{ryser1963combinatorial}.

Applying this result to the expression for aggregative efficiency
from our Theorem \ref{thm:efficiency_dc}, $\lim_{i\to\infty}(r_{i}/i)=\left(1+\frac{d^{2}-d}{d^{2}-d+c}\right)\frac{1}{K}=\left(2-\frac{c}{d^{2}-d+c}\right)\frac{1}{K}=(2-\frac{1}{K})\cdot\frac{1}{K}$.
\end{proof}

\subsection{Proof of Theorem \ref{thm:L_gen}}
Without loss of generality, suppose $\sigma^{2}=4$ so that $r_{i}=\text{Var}(\ell_{i}\mid\omega=1)$
and $\mathbb{E}[\ell_{i}\mid\omega=1]=\frac{1}{2}r_{i}.$\footnote{Some calculations in this proof were produced with the assistance of GPT-5.2 Pro.} For convenience,
we also write $(t,k)$ to refer to generation $t$ agent in position
$k\in\{1,...,K\}.$ Let $S_{t}:=\sum_{j=1}^{K}\ell_{t,j}$ be the
sum of log-actions in generation $t$. 

By Proposition \ref{prop:linear} and using the fact that all observed log-actions
in the same generation are exchangeable, for every $t\ge L+1$ we
must have coefficients $\beta_{t,1},...,\beta_{t,L}$ such that for
every $k$, 
\[
\ell_{t,k}=\lambda_{t,k}+\sum_{m=1}^{L}\beta_{t,m}S_{t-m}.
\]
 We define the aggregate weights to be $b_{t,m}:=K\beta_{t,m}$. Let
$r_{t}:=\text{Var}(\ell_{t,k}\mid\omega=1)$ (which is the same for
all choices of $k$) and let $u_{t}:=\text{Cov}(\ell_{t,k},\ell_{t,k'}\mid\omega=1)$
for any $k\ne k'$. We note that since $\lambda_{t,k}$ is conditionally
independent of all log-actions from previous generations, $r_{t}=1+u_{t}.$
We also write $A_{t}:=\text{Cov}(\ell_{t,1},S_{t}\mid\omega=1)=r_{t}+(K-1)u_{t}=Kr_{t}-(K-1).$ 

The next lemma shows that the covariance between any agent's log-action
and any social neighbor's log-action is equal to the conditional variance
of that neighbor's log-action. 
\begin{lem}
\label{lem:social_cov_var_match} For any $t$ and any predecessor
$(s,j)$ observed by generation $t$ agents, we have $\text{Cov}(\ell_{t,k},\ell_{s,j}\mid\omega=1)=\text{Var}(\ell_{s,j}\mid\omega=1)=r_{s}$.
Therefore, we also have $\text{Cov}(\ell_{t,k},S_{s}\mid\omega=1)=Kr_{s}$. 
\end{lem}
\begin{proof}
Let $X$ be the vector containing the log-actions observed by generation
$t.$ We know from Proposition \ref{prop:linear} that $\text{Cov}(X\mid\omega=1)\beta=2\mathbb{E}[X\mid\omega=1]$,
where $\beta$ is the vector of rational weights. Consider row $h$
of this vector equation. We have $\text{Cov}(\beta^{T}X,X_{h}\mid\omega=1)=(\text{Cov}(X\mid\omega=1)\beta)_{h}$
by linearity of covariance, which must equal row $h$ of the right-hand
side which is $\text{Var}(X_{h}\mid\omega=1).$ At the same time,
$\text{Cov}(\ell_{t,k},X_{h}\mid\omega=1)=\text{Cov}(\beta^{T}X,X_{h}\mid\omega=1)$
since $\ell_{t,k}-\beta^{T}X=\lambda_{t,k}$, which is conditionally
independent of $X_{h}$. This shows $\text{Cov}(\ell_{t,k},\ell_{s,j}\mid\omega=1)=\text{Var}(\ell_{s,j}\mid\omega=1)$
for any observed $(s,j)$, and we know that $\text{Var}(\ell_{s,j}\mid\omega=1)=r_{s}$
by our normalization. Finally, $\text{Cov}(\ell_{t,k},S_{s}\mid\omega=1)=Kr_{s}$
obtains by summing $\text{Cov}(\ell_{t,k},\ell_{s,j}\mid\omega=1)=r_{s}$
over $1\le j\le K.$ 
\end{proof}
Let $t\ge L+1.$ Consider the $L$-vector of sums of log-actions in
the $L$ observed generations, $S_{(t-1):(t-L)}:=(S_{t-1},S_{t-2},...,S_{t-L})$.
Let $\Sigma_{t}:=\text{Cov}(S_{(t-1):(t-L)}\mid\omega=1)$ be its
conditional covariance matrix. By Lemma \ref{lem:social_cov_var_match},
for each $m=1,...,L,$ $\text{Cov}(\ell_{t,1},S_{t-m}\mid\omega=1)=Kr_{t-m}.$
This is also equal to $\text{Cov}(\sum_{m=1}^{L}\beta_{t,m}S_{t-m},S_{t-m}),$
since $\lambda_{t,1}$ is conditionally independent of predecessors'
log-actions. This establishes that $\Sigma_{t}\cdot(\beta_{t,1},...,\beta_{t,L})^{T}=(Kr_{t-1},...,Kr_{t-L})^{T}.$
Equivalently, in terms of $b_{t,m}=K\beta_{t,m}$, $\Sigma_{t}\cdot(b_{t,1},...,b_{t,L})^{T}=K(Kr_{t-1},...,Kr_{t-L})^{T}$. 

Viewing this as a linear system in the $L$ variables $b_{t,1},...,b_{t,L}$,
we solve for $b_{t,1}$ using the Schur complement. Partition $\Sigma_{t}$
as 
\[
\Sigma_{t}=\left(\begin{array}{cc}
\text{\text{Var}}(S_{t-1}\mid\omega=1) & \text{Cov}(S_{t-1},S_{(t-2):(t-L)}\mid\omega=1)\\
\text{Cov}(S_{(t-2):(t-L)},S_{t-1}\mid\omega=1) & \text{Cov}(S_{(t-2):(t-L)}\mid\omega=1)
\end{array}\right).
\]
Applying the Schur complement and the formula for the conditional
variance among a family of jointly normal random variables, we can
express $b_{t,1}$ in terms of the conditional variance of $S_{t-1}$
given $S_{t-2},...,S_{t-L}$. 
\begin{lem}
\label{lem:b_and_var}Let $H_{t-1}=\frac{1}{K}\text{Var}(S_{t-1}\mid S_{t-2},...,S_{t-L},\omega=1)$.
Then, $b_{t,1}=1+\frac{K-1}{H_{t-1}}.$ 
\end{lem}
\begin{proof}
In solving for $b_{t,1}$ in the linear system $\Sigma_{t}\cdot(b_{t,1},...,b_{t,L})^{T}=K(Kr_{t-1},...,Kr_{t-L})^{T},$
we write $\Sigma_{t}=\left(\begin{array}{cc}
\Sigma_{11} & \Sigma_{12}\\
\Sigma_{21} & \Sigma_{22}
\end{array}\right)$ where $\Sigma_{1,1}$ is $1\times1$ and $\Sigma_{2,2}$ is $(L-1)\times(L-1)$.
The standard formula is 
\[
(\Sigma_{11}-\Sigma_{12}\Sigma_{22}^{-1}\Sigma_{21})b_{t,1}=K^{2}r_{t-1}-(\Sigma_{12}\Sigma_{22}^{-1})K(Kr_{t-2},...,Kr_{t-L})^{T}.
\]

On the right-hand side, we recall the identities $\text{Var}(S_{t-1}\mid\omega=1)=\text{Cov}(\sum_{k=1}^{K}\ell_{t-1,k},S_{t-1}\mid\omega=1)=KA_{t-1}=K(Kr_{t-1}-(K-1))$,
and $\text{Cov}(S_{t-1},S_{t-m}\mid\omega=1)=K\text{Cov}(\ell_{t-1,1},S_{t-m}\mid\omega=1)=K\cdot(Kr_{t-m})$.
So, we can write $K^{2}r_{t-1}=\text{Var}(S_{t-1}\mid\omega=1)+K(K-1)=\Sigma_{11}+K(K-1).$
We can also write $K(Kr_{t-2},...,Kr_{t-L})^{T}=\text{Cov}(S_{(t-2):(t-L)},S_{t-1}\mid\omega=1)=\Sigma_{21}$.
Therefore we have 
\[
(\Sigma_{11}-\Sigma_{12}\Sigma_{22}^{-1}\Sigma_{21})b_{t,1}=K(K-1)+(\Sigma_{11}-\Sigma_{12}\Sigma_{22}^{-1}\Sigma_{21}).
\]

Dividing through by $(\Sigma_{11}-\Sigma_{12}\Sigma_{22}^{-1}\Sigma_{21}),$
we get $b_{t,1}=1+\frac{K(K-1)}{\Sigma_{11}-\Sigma_{12}\Sigma_{22}^{-1}\Sigma_{21}}$. 

For the Schur complement $\Sigma_{11}-\Sigma_{12}\Sigma_{22}^{-1}\Sigma_{21},$
it is equal to 
\[
\text{\text{Var}}(S_{t-1}\mid\omega=1)-\text{Cov}(S_{t-1},S_{(t-2):(t-L)}\mid\omega=1)\text{Cov}(S_{(t-2):(t-L)}\mid\omega=1)^{-1}\text{Cov}(S_{(t-2):(t-L)},S_{t-1}\mid\omega=1).
\]
By joint normality of $(S_{t-1},...,S_{t-L})$ when $\omega=1$ \citep[pp. 116-117]{eaton1983multivariate}, this
is equal to $\text{Var}(S_{t-1}\mid S_{t-2},...,S_{t-L},\omega=1)$.
Therefore, if we let $H_{t-1}=\frac{1}{K}\text{Var}(S_{t-1}\mid S_{t-2},...,S_{t-L},\omega=1),$
then $b_{t,1}=1+\frac{K(K-1)}{\Sigma_{11}-\Sigma_{12}\Sigma_{22}^{-1}\Sigma_{21}}=1+\frac{K-1}{H_{t-1}}.$ 
\end{proof}
Because $S_{t-1}$ contains generation $t-1$'s private signals, we
get $\text{Var}(S_{t-1}\mid S_{t-2},...,S_{t-L},\omega=1)\ge\text{Var}(\sum_{j=1}^{K}\lambda_{t-1,j}\mid\omega=1)=K$.
So, $H_{t-1}\ge1$ and $1\le b_{t,1}\le K$. 

Now let $D_{t}:=A_{t}-Kr_{t-1}$. Using the fact that $A_{t}=Kr_{t}-(K-1),$
we may rewrite $D_{t}=K(r_{t}-r_{t-1})-(K-1).$ So $D_{t}$ is, up
to a linear transformation, the number of signals aggregated in generation
$t.$ 
\begin{lem}
\label{lem:D_and_b}$D_{t}=1+(K-1)b_{t,1}$. 
\end{lem}
\begin{proof}
Given $D_{t}=K(r_{t}-r_{t-1})-(K-1)$, it is enough to show that $r_{t}-r_{t-1}=1+\frac{K-1}{K}b_{t,1}$. 

From rational log-actions, $\ell_{t,k}=\lambda_{t,k}+\sum_{m=1}^{L}\beta_{t,m}S_{t-m}$.
Take conditional expectation of both sides and using the signal-counting
interpretation, we get 
\begin{equation}
r_{t}=1+\sum_{m=1}^{L}b_{t,m}r_{t-m}.\label{eq:way1}
\end{equation}

Using the fact that $\lambda_{t,k}$ is conditionally independent
of $S_{t-1},$ we get 

\begin{align*}
\text{Cov}(\ell_{t,1},S_{t-1}\mid\omega=1) & =\sum_{m=1}^{L}\text{Cov}(\beta_{t,m}S_{t-m},S_{t-1}\mid\omega=1)\\
 & =\sum_{m=1}^{L}\beta_{t,m}\text{Cov}(S_{t-m},S_{t-1}\mid\omega=1)\\
 & =\beta_{t,1}\text{Var}(S_{t-1}\mid\omega=1)+\sum_{m=2}^{L}\beta_{t,m}\text{Cov}(S_{t-m},S_{t-1}\mid\omega=1)\\
 & =\beta_{t,1}KA_{t-1}+\sum_{m=2}^{L}\beta_{t,m}K^{2}r_{t-m}\\
 & =\frac{b_{t,1}}{K}(K^{2}r_{t-1}-K(K-1))+\sum_{m=2}^{L}\frac{b_{t,m}}{K}K^{2}r_{t-m}\\
 & =b_{t,1}(Kr_{t-1}-(K-1))+K\sum_{m=2}^{L}b_{t,m}r_{t-m}
\end{align*}

But alternatively, we can also use the identity from Lemma \ref{lem:social_cov_var_match}
to write $\text{Cov}(\ell_{t,1},S_{t-1}\mid\omega=1)=Kr_{t-1}$. 

Therefore, $b_{t,1}(Kr_{t-1}-(K-1))+K\sum_{m=2}^{L}b_{t,m}r_{t-m}=Kr_{t-1}$,
which can be rearranged to give 
\begin{equation}
r_{t-1}+\frac{K-1}{K}b_{t,1}=\sum_{m=1}^{L}b_{t,m}r_{t-m}.\label{eq:way2}
\end{equation}

Subtract Equation (\ref{eq:way2}) from Equation (\ref{eq:way1})
to get $r_{t}-r_{t-1}-\frac{K-1}{K}b_{t,1}=1\iff r_{t}-r_{t-1}=1+\frac{K-1}{K}b_{t,1}$
as desired. 
\end{proof}
Since we know $1\le b_{t,1}\le K,$ Lemma \ref{lem:D_and_b} implies
the bounds $K\le D_{t}\le1+(K-1)K$ for $t\ge L+1.$ 

Next, we seek to write $H_{t-1}$ as a continued fraction involving
$D_{t-1},D_{t-2},...,D_{t-(L-1)}$. Thus, by Lemma \ref{lem:b_and_var}
and Lemma \ref{lem:D_and_b}, we also relate $D_{t}$ to such a continued
fraction, towards finding a recursion for $D$. 

Define $Y_{t,1}:=S_{t-1}-S_{t-2}$, $Y_{t,2}:=S_{t-2}-S_{t-3},$...,
$Y_{t,L-1}:=S_{t-L+1}-S_{t-L}$, $Y_{t,L}:=S_{t-L}$. This is an invertible
linear change of variables, so conditioning on $(S_{t-2},...,S_{t-L})$
is the same as conditioning on $(Y_{t,2},...,Y_{t,L})$ and we have
$S_{t-1}=Y_{t,1}+S_{t-2}$. So, $\text{Var}(S_{t-1}\mid S_{t-2},...,S_{t-L},\omega=1)=\text{Var}(Y_{t,1}\mid Y_{t,2},...,Y_{t,L},\omega=1)$.
Using Lemma \ref{lem:social_cov_var_match} and the identities $\text{Var}(S_{s}\mid\omega=1)=KA_{s}$
and $\text{Cov}(S_{s},S_{s-1}\mid\omega=1)=K^{2}r_{s-1}=K(A_{s-1}+(K-1))$,
we can find an expression for the covariance matrix for the vector
of differences $Y_{t}$. The key simplification for the analysis is
that this covariance matrix is very sparse: for any $L$ and any $t,$
if $|i-j|>1$ then $Y_{t,i}$ and $Y_{t,j}$ are conditionally uncorrelated!
To see why, consider $\text{Cov}(Y_{t,1},Y_{t,3}\mid\omega=1)$ (for
$L\ge4$) and note this is $\text{Cov}(S_{t-1}-S_{t-2},S_{t-3}-S_{t-4}\mid\omega=1).$
Even though the $S$ terms are correlated with each other, we have
the cancellation $\text{Cov}(S_{t-1},S_{t-3}-S_{t-4}\mid\omega=1)=\text{Cov}(S_{t-2},S_{t-3}-S_{t-4}\mid\omega=1)$
thanks to Lemma \ref{lem:social_cov_var_match}, as both covariances
are governed by the conditional variances of $S_{t-3}$ and $S_{t-4}$
(since both generations $t-1$ and $t-2$ observe generations $t-3$
and $t-4$). So the difference $S_{t-1}-S_{t-2}$ is exactly uncorrelated
with $S_{t-3}-S_{t-4}.$ 
\begin{lem}
\label{Gt} Let $G_{t}:=\frac{1}{K}\text{Cov}(Y_{t,1},...,Y_{t,L}\mid\omega=1)$.
Then, $G_{t}$ is tri-diagonal with 

\[
G_{t}=\left(\begin{array}{ccccc}
D_{t-1}-(K-1) & K-1 & 0 & \cdots & 0\\
K-1 & D_{t-2}-(K-1) & K-1 & \ddots & \vdots\\
0 & K-1 & D_{t-3}-(K-1) & \ddots & 0\\
\vdots & \ddots & \ddots & \ddots & K-1\\
0 & \cdots & 0 & K-1 & A_{t-L}
\end{array}\right).
\]
Moreover, $H_{t-1}$ is the Schur complement of the lower-right $(L-1)\times(L-1)$
block of $G_{t}$. 
\end{lem}
\begin{proof}
Apply Lemma \ref{lem:social_cov_var_match} to deduce that for $1\le m<m'\le L,$
$\text{Cov}(S_{t-m},S_{t-m'}\mid\omega=1)=K^{2}r_{t-m'}$. 

For the diagonal entries, $\text{Cov}(Y_{t,m},Y_{t,m}\mid\omega=1)$
for $m=1,...,L-1$ gives 
{\footnotesize
\begin{align*}
\text{Cov}(S_{t-m}-S_{t-m-1},S_{t-m}-S_{t-m-1}\mid\omega=1)= & \text{Var}(S_{t-m}\mid\omega=1)+\text{Var}(S_{t-m-1}\mid\omega=1)-2\text{Cov}(S_{t-m},S_{t-m-1}\mid\omega=1)\\
= & KA_{t-m}+KA_{t-m-1}-2K^{2}r_{t-m-1}\\
= & [K^{2}r_{t-m}-K(K-1)]+[K^{2}r_{t-m-1}-K(K-1)]-2K^{2}r_{t-m-1}\\
= & K[K(r_{t-m}-r_{t-m-1})-(K-1)-(K-1)]\\
= & K[D_{t-m}-(K-1)]
\end{align*}
}
And since $Y_{t,L}:=S_{t-L},$ we simply have $\text{Var}(Y_{t,L}\mid\omega=1)=\text{Var}(S_{t-L}\mid\omega=1)=KA_{t-L}$. 

For the $(i,j)$-th entry where $j=i+1$ (the other case is symmetric),
\begin{align*}
\text{Cov}(Y_{t,i},Y_{t,i+1}\mid\omega=1) & =\text{Cov}(S_{t-i}-S_{t-(i+1)},S_{t-(i+1)}-S_{t-(i+2)}\mid\omega=1)\\
 & =K^{2}[r_{t-(i+1)}-r_{t-(i+2)}+r_{t-(i+2)}]-\text{Var}(S_{t-(i+1)}\mid\omega=1)\\
 & =K^{2}r_{t-(i+1)}-(K^{2}r_{t-(i+1)}-K(K-1))\\
 & =K(K-1)
\end{align*}
And finally for $(i,j)$-th entry where $|i-j|>1,$ note all indices
$i,i+1,j,j+1$ are distinct, so 
\begin{align*}
\text{Cov}(Y_{t,i},Y_{t,j}\mid\omega=1) & =\text{Cov}(S_{t-i}-S_{t-(i+1)},S_{t-j}-S_{t-(j+1)}\mid\omega=1)\\
 & =K^{2}[r_{t-j}-r_{t-(j+1)}-r_{t-j}+r_{t-(j+1)}]\\
 & =0
\end{align*}

Finally, to compute the Schur complement, write $G_{t}=\left(\begin{array}{cc}
G_{t,11} & G_{t,12}\\
G_{t,21} & G_{t,22}
\end{array}\right)$ where $G_{t,11}$ is $1\times1$ and $G_{t,22}$ is $(L-1)\times(L-1).$
The Schur complement is $G_{t,11}-G_{t,12}G_{t,22}^{-1}G_{t,21}$,
but since $G_{t}$ is $\frac{1}{K}$ times the covariance matrix of
the jointly normal random variables $(Y_{t,m})_{m=1}^{L}$, we get
\[
\text{Var}(Y_{t,1}\mid Y_{t,2},...,Y_{t,L},\omega=1)=K(G_{t,11}-G_{t,12}G_{t,22}^{-1}G_{t,21}).
\]
But by definition $H_{t-1}=\frac{1}{K}\text{Var}(S_{t-1}\mid S_{t-2},...,S_{t-L},\omega=1)$
and we said before $\text{Var}(S_{t-1}\mid S_{t-2},...,S_{t-L},\omega=1)=\text{Var}(Y_{t,1}\mid Y_{t,2},...,Y_{t,L},\omega=1),$
so $H_{t-1}$ is the Schur complement of the lower-right $(L-1)\times(L-1)$
block of $G_{t}$. 
\end{proof}
By Theorem 3.1 of \cite{frommer2021matrix}, for a tri-diagonal matrix $\left(\begin{array}{ccccc}
a_{1} & b & 0 & \cdots & 0\\
b & a_{2} & b & \ddots & \vdots\\
0 & b & a_{3} & \ddots & 0\\
\vdots & \ddots & \ddots & \ddots & b\\
0 & \cdots & 0 & b & a_{L}
\end{array}\right)$, the Schur complement of the lower-right $(L-1)\times(L-1)$ block
can be written as a continued fraction, 
\[
a_{1}-\frac{b^{2}}{a_{2}-\frac{b^{2}}{a_{3}-\frac{b^{2}}{\ddots-\frac{b^{2}}{a_{L}}}}}
\]
In particular, let $\xi_{t,L}:=A_{t-L}$ and recursively define $\xi_{t,m}:=(D_{t-m}-(K-1))-\frac{(K-1)^{2}}{\xi_{t,m+1}}$
for $m=L-1,...,1$. Then for every $t,$ $H_{t-1}=\xi_{t,1}$. Using
Lemma \ref{lem:b_and_var} and Lemma \ref{lem:D_and_b}, we get $D_{t}=1+(K-1)b_{t,1}=1+(K-1)(1+\frac{K-1}{H_{t-1}})=K+\frac{(K-1)^{2}}{\xi_{t,1}}.$
We record now a useful fact about the partial fractions for later: 
\begin{lem}
\label{lem:xi_positivity}For every $m\in\{1,...,L-1\},$ $\xi_{t,m}=\frac{1}{K}\text{Var}(Y_{t,m}\mid Y_{t,m+1},...,Y_{t,L},\omega=1)\ge1.$ 
\end{lem}
\begin{proof}
This $\frac{1}{K}\text{Var}(Y_{t,m}\mid Y_{t,m+1},...,Y_{t,L},\omega=1)\ge1$
comes from the fact that $\text{Var}(Y_{t,m}\mid Y_{t,m+1},...,Y_{t,L},\omega=1)\ge\text{Var}(\sum_{j=1}^{K}\lambda_{t-m,j}\mid Y_{t,m+1},...,Y_{t,L},\omega=1)=K$. 

The recursive definition of $\xi_{t,m}$ is such that each $\xi_{t,m}$
is a Schur complement of a submatrix of $G_{t}$. Let $G_{t}^{(m)}$
refer to the submatrix containing only rows and columns $m,m+1,...,L$
of $G_{t}$. We may write $G_{t}^{(m)}=\left(\begin{array}{cc}
G_{t,11}^{(m)} & G_{t,12}^{(m)}\\
G_{t,21}^{(m)} & G_{t,22}^{(m)}
\end{array}\right)$ where $G_{t,11}^{(m)}$ is $1\times1$ and $G_{t,22}^{(m)}$ is $(L-m)\times(L-m).$
Then the Schur complement of the lower-right block $G_{t,22}^{(m)}$
(which is also $G_{t}^{(m+1)})$ is $G_{t,11}^{(m)}-G_{t,12}^{(m)}[G_{t}^{(m+1)}]^{-1}G_{t,21}^{(m)}$.
This matches the iterative definition of $\xi_{t,m}.$ Also, since
$G_{t}^{(m)}$ is $1/K$ times the covariance matrix of the conditionally
jointly normal random variables $(Y_{t,m},...,Y_{t,L}),$ the Schur
complement $\xi_{t,m}$ is also equal to $\frac{1}{K}\text{Var}(Y_{t,m}\mid Y_{t,m+1},...,Y_{t,L},\omega=1).$ 
\end{proof}
Assume $(D_{t})$ converges to a limit $D^{*}$. Consider  a version
of the continued fraction iteration with an initial condition of infinity.
For every $D$, define $\tilde{\xi}_{L}(D):=\infty$. Then, iteratively,
define $\tilde{\xi}_{m}(D):=(D-(K-1))-\frac{(K-1)^{2}}{\tilde{\xi}_{m+1}(D)}$
for $m=L-1,...,1.$ A necessary condition for $\lim_{t \to\infty}D_{t}=D^{*}$
is that $D^{*}=K+\frac{(K-1)^{2}}{\tilde{\xi}_{1}(D^{*})}$ (the initial
condition $\tilde{\xi}_{L}=\infty$ comes from the fact that $a_{t-L}\to\infty$
as $t\to\infty$, and note $\xi_{t,m}\ge1$ for each $t$ so errors
don't blow up in the recursion). Let $F(\lambda):=K+\frac{(K-1)^{2}}{\tilde{\xi}_{1}(\lambda)}$,
so $D^{*}$ is a fixed point of $F.$ 

Now define the $L\times L$ tri-diagonal matrix
\[
N_{K,L}=\left(\begin{array}{ccccccc}
K & K-1 & 0 & 0 & 0 & \cdots & 0\\
K-1 & K-1 & K-1 & 0 & 0 & \ddots & \vdots\\
0 & K-1 & K-1 & K-1 & 0 & \ddots & 0\\
\vdots & \ddots & \ddots & \ddots & \ddots & \ddots & K-1\\
0 & \cdots & 0 & 0 & 0 & K-1 & K-1
\end{array}\right)
\]
 and let $\hat{N}_{K,L}$ represent its $(L-1)\times(L-1)$ lower-right
block. Let $\Delta_{L}(\lambda):=\det(\lambda I_{L\times L}-N_{K,L})$
and $\Delta_{L-1}(\lambda):=\det(\lambda I_{(L-1)\times(L-1)}-\hat{N}_{K,L})$
be the characteristic polynomials of $N_{K,L}$ and $\hat{N}_{K,L}$,
respectively. 
\begin{lem}
\label{lem:polynomial_and_fixed_point}If $\Delta_{L-1}(\lambda)\ne0,$
then $\Delta_{L}(\lambda)=\Delta_{L-1}(\lambda)\cdot(\lambda-F(\lambda))$. 
\end{lem}
\begin{proof}
We write 
\[
\lambda I_{L\times L}-N_{K,L}=\left(\begin{array}{ccccccc}
\lambda-K & -(K-1) & 0 & 0 & 0 & \cdots & 0\\
-(K-1) & \lambda-(K-1) & -(K-1) & 0 & 0 & \ddots & \vdots\\
0 & -(K-1) & \lambda-(K-1) & -(K-1) & 0 & \ddots & 0\\
\vdots & \ddots & \ddots & \ddots & \ddots & \ddots & -(K-1)\\
0 & \cdots & 0 & 0 & 0 & -(K-1) & \lambda-(K-1)
\end{array}\right).
\]
 The $(L-1)\times(L-1)$ block in the lower-right is $\lambda I_{(L-1)\times(L-1)}-\hat{N}_{K,L}$.
We use the block determinant formula to get: 
{\footnotesize
\[
\det(\lambda I_{L\times L}-N_{K,L})=\det(\lambda I_{(L-1)\times(L-1)}-\hat{N}_{K,L})\cdot\det((\lambda-K)-(-(K-1),0,...0)\cdot(\lambda I_{(L-1)\times(L-1)}-\hat{N}_{K,L})^{-1}\cdot(-(K-1),0,...0)^{T}))
\]
}
But we have 
\[
(-(K-1),0,...0)\cdot(\lambda I_{(L-1)\times(L-1)}-\hat{N}_{K,L})^{-1}\cdot(-(K-1),0,...0)^{T})=(K-1)^{2}\cdot[(\lambda I_{(L-1)\times(L-1)}-\hat{N}_{K,L})^{-1}]_{1,1}
\]

The matrix $\lambda I_{(L-1)\times(L-1)}-\hat{N}_{K,L}$ is tri-diagonal,
so again we have a continued-fraction representation of the (1,1)
entry of its inverse: 

{\scriptsize
\[
\left[\left(\begin{array}{ccccc}
\lambda-(K-1) & -(K-1) & 0 & \cdots & 0\\
-(K-1) & \lambda-(K-1) & -(K-1) & \ddots & \vdots\\
0 & -(K-1) & \lambda-(K-1) & \ddots & 0\\
\vdots & \ddots & \ddots & \ddots & -(K-1)\\
0 & \cdots & 0 & -(K-1) & \lambda-(K-1)
\end{array}\right)^{-1}\right]_{1,1}=\left(\lambda-(K-1)-\frac{(K-1)^{2}}{\lambda-(K-1)-\frac{(K-1)^{2}}{\ddots-\frac{(K-1)^{2}}{\lambda-(K-1)}}}\right)^{-1}
\]
 }where $\lambda-(K-1)$ appears $L-1$ times in the continued fraction.
But in the iteratively defined $\tilde{\xi}_{m}(\lambda)$ starting
with $\tilde{\xi}_{L}=\infty,$ we get $\tilde{\xi}_{L-1}(\lambda):=(\lambda-(K-1))$
and more generally $\tilde{\xi}_{m}(\lambda)$ is the same continued
fraction with $\lambda-(K-1)$ appearing $L-m$ times for $m=L-1,L-2,...,1.$
So the continued fraction above is $\tilde{\xi}_{1}(\lambda).$ 

Thus, we have $\Delta_{L}(\lambda)=\Delta_{L-1}(\lambda)\cdot((\lambda-K)-\frac{(K-1)^{2}}{\tilde{\xi}_{1}(\lambda)}).$
Since $F(\lambda):=K+\frac{(K-1)^{2}}{\tilde{\xi}_{1}(\lambda)}$,
we get $\lambda-F(\lambda)=\lambda-K-\frac{(K-1)^{2}}{\tilde{\xi}_{1}(\lambda)}$,
hence the claim $\Delta_{L}(\lambda)=\Delta_{L-1}(\lambda)\cdot(\lambda-F(\lambda)).$ 
\end{proof}
In particular, Lemma \ref{lem:polynomial_and_fixed_point} implies
that if $\Delta_{L-1}(D^{*})\ne0,$ then $\Delta_{L}(D^{*})=0$ (since
$D^{*}-F(D^{*})=0$), so $D^{*}$ must be an eigenvalue of $N_{K,L}$.
To complete the proof, we now rule out the case of $\Delta_{L-1}(D^{*})=0$
and show that $D^{*}$ cannot be any eigenvalue of $N_{K,L}$ other
than the largest one. 

First, we show $D^{*}I_{(L-1)\times(L-1)}-\hat{N}_{K,L}$ must be
positive definite. For $1\le m\le L-1$, write $\theta_{m}$ for the
determinant of the leading $m\times m$ submatrix of $D^{*}I_{(L-1)\times(L-1)}-\hat{N}_{K,L}$,
and define $\theta_{0}:=1.$ Let $\eta_{1}=D^{*}-(K-1)$ and for $2\le m\le L-1,$
let $\eta_{m}$ be the continued fraction $$D^{*}-(K-1)-\frac{(K-1)^{2}}{D^{*}-(K-1)-\frac{(K-1)^{2}}{\ddots-\frac{(K-1)^{2}}{D^{*}-(K-1)}}}$$
where $D^{*}-(K-1)$ appears $m$ times. By a property of tri-diagonal
matrices, $\eta_{m}=\frac{\theta_{m}}{\theta_{m-1}}$ for $1\le m\le L-1$.
By Sylvester's criterion, if $D^{*}I_{(L-1)\times(L-1)}-\hat{N}_{K,L}$
is not positive definite, then there is some smallest $m'\in\{1,...,L-1\}$
such that $\theta_{m'}\le0.$ Then, $\eta_{m'}\le0.$ But, in the
proof of Lemma \ref{lem:polynomial_and_fixed_point}, we showed that
$\eta_{m'}=\tilde{\xi}_{L-m'}(D^{*})$. But when $D_{t}\to D^{*}$,
$\tilde{\xi}_{L-m'}=\lim_{t\to\infty}\xi_{t,L-m'}$ and we know from
Lemma \ref{lem:xi_positivity} that $\xi_{t,L-m'}\ge1$ for every
$t.$ So the limit $\tilde{\xi}_{L-m'}$ cannot be weakly negative,
a contradiction. 

So, $D^{*}I_{(L-1)\times(L-1)}-\hat{N}_{K,L}$ must be positive definite,
which rules out $\Delta_{L-1}(D^{*})=0$ and implies $D^{*}$ is strictly
larger than the largest eigenvalue of $\hat{N}_{K,L}.$ But $\hat{N}_{K,L}$
is the lower-right block of $N_{K,L}$, so the eigenvalue interlacing
theorem implies that there is only one eigenvalue of $N_{K,L}$ that
satisfies this: namely, the largest eigenvalue of $N_{K,L}.$

To prove properties of $g(K,L),$ first note that the maximum row
sum of $N_{K,L}$ is $3K-3,$ so $\rho(K,L)\le3K-3.$ This implies
$\frac{\rho(K,L)+K-1}{K}\le\frac{3K-3+K-1}{K}=4-4/K<4.$

When $L=1,$ we get $g(K,L)=\frac{\rho(K,L)+K-1}{K}=\frac{2K-1}{K},$
which is strictly increasing in $K.$ Now consider the case of $L\ge2.$
Since $N_{K,L}$ is symmetric, its largest eigenvalue can be expressed
as $\max_{\parallel v\parallel\le1}v^{T}N_{K,L}v$.

For the case of going from $K=1$ to $K=2,$ we know $\rho(1,L)=1$
and $g(1,L)=1$. By taking $v=(1,0,...,)^{T}$, we also get $\rho(2,L)\ge(1,0,...,)\cdot N_{K,L}\cdot(1,0,...)^{T}=2$.
So $g(2,L)\ge3/2>g(1,L).$ 

For the case of going from $K$ to $K+1$ starting from $K\ge2,$
we first show that $\rho(K,L)$ is strictly increasing in $K$ and
convex in $K$. 
\begin{itemize}
\item Strictly increasing in $K$: Since $N_{K,L}$ is non-negative and
irreducible, by Perron--Frobenius it has an eigenvector $v^{*}$
associated with its largest eigenvalue such that $v_{i}^{*}>0$ for
every $i.$ Normalize $\parallel v^{*}\parallel=1.$ Then, $\rho(K,L)=(v^{*})^{T}N_{K,L}v^{*}$.
But now, we also have $\rho(K+1,L)=\max_{\parallel v\parallel\le1}v^{T}N_{K+1,L}v$
. Observe that $N_{K+1,L}=N_{K,L}+T_{L},$ where $T_{L}$ is the $L\times L$
tri-diagonal matrix with $1$ in the diagonal and off-diagonal entries
and 0 elsewhere. In the maximization problem for $\rho(K+1,L),$ let
$v=v^{*}$ and we find that $\rho(K+1,L)\ge(v^{*})^{T}N_{K,L}v^{*}+(v^{*})^{T}T_{L}v^{*}=\rho(K,L)+(v^{*})^{T}T_{L}v^{*}.$
Using the structure of $T_{L},$we have $(v^{*})^{T}T_{L}v^{*}=\sum_{i=1}^{L}(v_{i}^{*})^{2}+2\sum_{i=1}^{L-1}v_{i}^{*}v_{i+1}^{*}$.
But $v_{i}^{*}>0$ for every $i,$ so this sum is strictly positive.
This shows $\rho(K,L)$ is strictly increasing in $K$ for $K\ge2.$ 
\item Convex in $K$: Let $N_{K,L}=(K-1)T_{L}+E$ where $E$ is  the $L\times L$
matrix with 1 in the upper-left corner and 0 elsewhere. Then, we may
write $\rho(K,L)=\max_{\parallel v\parallel\le1}\left((K-1)v^{T}T_{L}v+(v_{1})^{2}\right)$.
For each fixed unit vector $v,$ the maximand is a linear function
of $(K-1),$ so $\rho(K,L)$ is convex in $K.$ 
\end{itemize}
By the convexity of $K\mapsto\rho(K,L)$, we deduce $K\mapsto\frac{\rho(K,L)-\rho(1,L)}{K-1}=\frac{\rho(K,L)-1}{K-1}$
is weakly increasing. We have $g(K,L)=\frac{\rho(K,L)+K-1}{K}=1+\frac{K-1}{K}\frac{\rho(K,L)-1}{K-1}$
for $K\ge2$. But $\frac{K-1}{K}$ is strictly increasing in $K,$
$\frac{\rho(K,L)-1}{K-1}$ is weakly increasing in $K$ and strictly
positive for $K\ge2$, as $\rho(K,L)>\rho(1,L)=1$ for $K\ge2$ since
$\rho$ is strictly increasing in $K.$ So $g(K,L)$ is strictly increasing
in $K\ge2$ for any $L\ge2$. 

To show that $g(K,L)$ is strictly increasing in $L$ for $K\ge2$,
let $v\in\mathbb{R}^{L}$ be a unit eigenvector of $N_{K,L}$ with
eigenvalue $\rho(K,L),$ where $v_{i}>0$ for every $i$ by the Perron-Frobenius
theorem. Then, $v^{T}N_{K,L}v=\rho(K,L).$ Now consider the unit vectors
in $\mathbb{R}^{L+1}$ parameterized by $\epsilon>0$, $\tilde{v}(\epsilon):=(v,\epsilon)/\sqrt{1+\epsilon^{2}}$.
Using the structure of $N_{K,L+1},$ we have 
\begin{align*}
\tilde{v}(\epsilon)^{T}N_{K,L+1}\tilde{v}(\epsilon) & =\frac{(v,\epsilon)^{T}N_{K,L+1}(v,\epsilon)}{1+\epsilon^{2}}\\
 & =\frac{v^{T}N_{K,L}v+2\epsilon(K-1)v_{L}+(K-1)\epsilon{}^{2}}{1+\epsilon^{2}}\\
 & =\frac{\rho(K,L)+2\epsilon(K-1)v_{L}+(K-1)\epsilon{}^{2}}{1+\epsilon^{2}}.
\end{align*}
Differentiating with respect to $\epsilon$ and evaluating at $\epsilon=0,$
we get: 

\[
\left[\frac{(2(K-1)v_{L}+2(K-1)\epsilon)(1+\epsilon^{2})-(\rho(K,L)+2\epsilon(K-1)v_{L}+(K-1)\epsilon{}^{2})\cdot(2\epsilon)}{(1+\epsilon^{2})^{2}}\right]_{\epsilon=0}=2(K-1)v_{L}
\]

This is strictly positive because $K\ge2$ and $v_{L}>0.$ This shows
that for some small enough $\epsilon>0,$ $\tilde{v}(\epsilon)^{T}N_{K,L+1}\tilde{v}(\epsilon)>\tilde{v}(0)N_{K,L+1}\tilde{v}(0)=\rho(K,L).$
So, $\rho(K,L+1)=\max_{\parallel v\parallel\le1}v^{T}N_{K,L+1}v>\rho(K,L).$ 

% To compute the limit $\lim_{K\to\infty}g(K,L)$, again write $N_{K,L}=(K-1)T+E.$
% The eigenvalues of $T$ are
% $\{1+2\cos\frac{j\pi}{L+1}:j\in\{1,...,L\}\}$ \cite[p. 154]{smith1985numerical}. Using $\rho_{\max}$
% and $\rho_{\min}$ to refer to the largest and smallest eigenvalues
% of a symmetric matrix, we have by Weyl's inequality 
% \[
% \rho_{\max}((K-1)T)+\rho_{\min}(E)\le\rho_{\max}((K-1)T+E)\le\rho_{\max}((K-1)T)+\rho_{\max}(E).
% \]
% The largest eigenvalue of $T$ is $1+2\cos\frac{\pi}{L+1}$, since
% cosine is strictly decreasing on $(0,\pi).$ So, $\rho_{\max}((K-1)T)=(K-1)(1+2\cos\frac{\pi}{L+1}).$
% The only two eigenvalues of $E$ are 0 and 1. Therefore, 
% \[
% (K-1)(1+2\cos\frac{\pi}{L+1})\le\rho(K,L)\le(K-1)(1+2\cos\frac{\pi}{L+1})+1.
% \]

% So for every $K$, 
% \[
% \frac{K-1}{K}(1+2\cos\frac{\pi}{L+1})+\frac{K-1}{K}\le g(K,L)\le\frac{K-1}{K}(1+2\cos\frac{\pi}{L+1})+1
% \]

% As $K\to\infty,$ both the upper bound and lower bound of $g(K,L)$
% have a limit of $2+2\cos\frac{\pi}{L+1}$, therefore $\lim_{K\to\infty}g(K,L)=2+2\cos(\frac{\pi}{L+1}).$ 

Finally, we show that $D_{t}$ converges.

\begin{prop} For all $K\ge2$ and $L\ge1$, $D_{t}\to\rho(K,L)$.
Hence $r_{t}-r_{t-1}=(D_{t}+K-1)/K\to g(K,L)$. \end{prop}
\begin{proof}
Write $a=K-1\ge1$. Using the fact that $r_{t}\to\infty$ and $A_{t}=Kr_{t}-(K-1)$,
we get $A_{t}\to\infty$. If $L=1$, the continued fraction for $D_{t}$
reads $D_{t}=K+a^{2}/A_{t-1}\to K=\rho(K,1)$ and we are done. 

So from now on assume $L\ge2$, $n=L-1$ and $d_{t}=D_{t}-a$. The
continued fraction for $D_{t}$ gives that, for all $t\ge2L+1$ we
can compute $d_{t+1}$ using the iteration 
\begin{equation}
\xi_{n+1}=A_{t-n},\qquad\xi_{m}=d_{t+1-m}-\frac{a^{2}}{\xi_{m+1}} \quad(1\le m\le n),\qquad d_{t+1}=1+\frac{a^{2}}{\xi_{1}},\label{eq:cf}
\end{equation}
where $\xi_{m}\ge1$ for every $m$. In particular, $1\le d_{t}\le1+a^{2}$. 

For $0\le j\le n+1$, let $C^{(j)}_{t}$ be the determinant of the
$j\times j$ tridiagonal matrix with diagonal $(d_{t},d_{t-1},\dots,d_{t-j+1})$
and off-diagonal entries $a$ (so $C^{(0)}_{t}=1$), and let $q_{j}(t)=C^{(j)}_{t}/a^{j}$
and $q(t)=(q_{1}(t),\dots,q_{n}(t))$. We have $d_{t}=a\,q_{1}(t)$,
so it is enough to show the vector $q(t)$ converges as $t\to\infty.$
The idea is to find an asymptotic relationship $T$ between the successive
iterates $q(t)$ and $q(t+1)$ (Step 1), then identify a Lyapunov
function for $T$ (Step 2). This implies there is a unique limit point
of $q(t)$, which shows the convergence of $D_{t}$ (Steps 3 and 4). 

\medskip{}
\noindent\textbf{Step 1: determinant coordinates.} 

Fix $j\le n$, and let $p_{j}=d_{t-j+1}$ and $p_{m}=d_{t+1-m}-a^{2}/p_{m+1}$
for $m<j$. This is recursion~\eqref{eq:cf} with $\xi_{j+1}$ replaced
by $+\infty$, so, since $v\mapsto d-a^{2}/v$ is increasing on $(0,\infty)$,
$1\le\xi_{m}\le p_{m}\le d_{t+1-m}\le1+a^{2}$. The $p_{m}$ are the
pivots of Gaussian elimination from the lower-right corner, so $C^{(j)}_{t}=p_{1}\cdots p_{j}$
and 
\begin{equation}
a^{-j}\le q_{j}(t)\le\bigl((1+a^{2})/a\bigr)^{j}\qquad(1\le j\le n).\label{eq:box}
\end{equation}
Thus $q(t)$ is valued in a compact rectangle $Q\subset(0,\infty)^{n}$
for $t\ge2L+1$.

Now take $j=n$ and let $\varepsilon_{t}=d_{t+1}-1-a^{2}/p_{1}\ge0$.
The values $\xi_{n}$ and $p_{n}$ differ by $a^{2}/A_{t-n}$, and
each of the $n$ remaining steps (to levels $n-1,\dots,1$ and then
to $d_{t+1}$) multiplies the difference by at most $a^{2}$, as $|a^{2}/v-a^{2}/v'|\le a^{2}|v-v'|$
for $v,v'\ge1$. Hence $\varepsilon_{t}\le a^{2L}/A_{t-n}\to0$. The
matrix defining $C^{(n+1)}_{t+1}$ has pivots $p_{n},\dots,p_{1}$
followed by $d_{t+1}-a^{2}/p_{1}=1+\varepsilon_{t}$, so $C^{(n+1)}_{t+1}=(1+\varepsilon_{t})\,C^{(n)}_{t}$.
The Desnanot--Jacobi identity for the matrix defining $C^{(j+1)}_{t+1}$,
whose two off-diagonal maximal minors are triangular with determinant
$a^{j}$, gives $C^{(j)}_{t+1}C^{(j)}_{t}-C^{(j+1)}_{t+1}C^{(j-1)}_{t}=a^{2j}$
for $1\le j\le n$. Dividing by $a^{2j}$, 
\begin{equation}
q_{j}(t+1)\,q_{j}(t)=1+q_{j+1}(t+1)\,q_{j-1}(t),\qquad q_{0}\equiv1,\qquad q_{n+1}(t+1)=(1+\varepsilon_{t})\,q_{n}(t)/a.\label{eq:qrec}
\end{equation}
Let $T\colon(0,\infty)^{n}\to(0,\infty)^{n}$, $q\mapsto q'$, be
the continuous map computed for $j=n,n-1,\dots,1$ by 
\begin{equation}
q'_{n}=\frac{1}{q_{n}}+\frac{q_{n-1}}{a},\qquad q'_{j}=\frac{1+q'_{j+1}\,q_{j-1}}{q_{j}}\quad(j<n),\qquad\text{where }q_{0}=1.\label{eq:T}
\end{equation}
By~\eqref{eq:qrec}, $q(t+1)$ is given by~\eqref{eq:T} with $q_{n-1}/a$
replaced by $(1+\varepsilon_{t})q_{n-1}/a$. Since $\varepsilon_{t}\to0$
and the denominators are bounded away from $0$ on $Q$, $q(t+1)-T(q(t))\to0$.

\medskip{}
\noindent\textbf{Step 2: a Lyapunov function.} Let $h(u,v)=u/v+v/u+1/(uv)$;
let $\psi(x)=0$ for $x\le1$ and $\psi(x)=\tfrac{1}{2}(x^{2}-x^{-2})-2\log x$
for $x\ge1$; let $\Gamma(x)=\tfrac{a}{2}x^{2}-(a-\tfrac{1}{a})\log x-a\psi(x)$;
and 
\[
V(q)=\sum^{n-1}_{j=0}h(q_{j},q_{j+1})+\Gamma(1/q_{n}),\qquad q_{0}=1.
\]
We claim that $V(Tq)\le V(q)$, with equality only if $(Tq)_{n}=q_{n}$.
Since $h(u,z)+h(z,v)=(\frac{1}{u}+\frac{1}{v})(z+\frac{1+uv}{z})$
is unchanged when $z$ is replaced by $(1+uv)/z$, the updates of
$q_{n-1},\dots,q_{1}$ in~\eqref{eq:T} leave $V$ unchanged (each
replaces $z=q_{j}$ by $(1+uv)/z$, where $u=q_{j-1}$ and $v=q'_{j+1}$
are its current neighbors). Put $u=q_{n-1}$, $x=1/q_{n}$, $y=1/q'_{n}$
and $r=x/y$; then $u=a(1/y-x)$, so $xy<1$. A computation gives
$h(u,\frac{1}{y})-h(u,\frac{1}{x})=(x-y)\bigl(\frac{1}{ax}+ax-\frac{a}{y}\bigr)$,
and then 
\[
V(Tq)-V(q)=h(u,\tfrac{1}{y})-h(u,\tfrac{1}{x})+\Gamma(y)-\Gamma(x)=-a\,E(x,y)-\tfrac{1}{a}\bigl(\tfrac{1}{r}-1+\log r\bigr),
\]
where $E(x,y)=r-1-\log r-\frac{1}{2}(x-y)^{2}-\psi(x)+\psi(y)$. The
last term is $\le0$, with equality iff $r=1$, so it suffices to
show $E\ge0$ when $xy\le1$. Note that $E(x,1/x)=0$ for all $x>0$.
If $x,y\le1$, then $E(x,y)=f(x)-f(y)-f'(y)(x-y)\ge0$ for $f(w)=-\log w-\frac{1}{2}w^{2}$,
which is convex on $(0,1]$. If $x>1$, then $y\le1/x<1$ and $\partial E/\partial y=(x-y)(1-y^{-2})<0$,
so $E(x,y)\ge E(x,1/x)=0$. If $y>1$, then $x\le1/y<1$ and $\partial E/\partial x=(y-x)(1-\frac{1}{xy})\le0$,
so $E(x,y)\ge E(1/y,y)=0$.

Consequently, if $V(T^{n}q)=V(q)$, then each step of $q^{(k)}=T^{k}q$,
$0\le k\le n$, preserves $V$, so $q^{(k)}_{n}=c$ for $0\le k\le n$.
By~\eqref{eq:T}, $q^{(k)}_{n-1}=a(c-1/c)$ for $0\le k\le n-1$,
and descending induction on $j$, using $q^{(k+1)}_{j}q^{(k)}_{j}=1+q^{(k+1)}_{j+1}q^{(k)}_{j-1}$,
shows that $q^{(k)}_{j}$ does not depend on $k\in\{0,\dots,j\}$.
In particular, $Tq=q$.

\medskip{}
\noindent\textbf{Step 3: the fixed point.} Now we show $T$ has at
most one fixed point and characterize it. Let $q$ satisfy $Tq=q$,
and set $q_{-1}=0$, $q_{0}=1$ and $q_{n+1}=q_{n}/a$. Then~\eqref{eq:T}
says $q^{2}_{j}=1+q_{j-1}q_{j+1}$ for $0\le j\le n$. Subtracting
the identities for $j-1$ and $j$ gives $(q_{j+1}+q_{j-1})/q_{j}=(q_{j}+q_{j-2})/q_{j-1}$
for $1\le j\le n$, so all these ratios equal their value $q_{1}$
at $j=0$; that is, $q_{j+1}=q_{1}q_{j}-q_{j-1}$ for $0\le j\le n$.
Hence the positive vector $v=(q_{n},\dots,q_{1},1)^{\top}$ satisfies
$N_{K,L}v=a(1+q_{1})v$ (the first row uses $q_{n+1}=q_{n}/a$), and
the Perron--Frobenius theorem gives $a(1+q_{1})=\rho(K,L)$. This
determines $q_{1}$ and then $q$: $T$ has at most one fixed point
$q^{*}$, and $a(1+q^{*}_{1})=\rho(K,L)$.

\medskip{}
\noindent\textbf{Step 4: convergence.} Let $\Omega$ be the set of
limit points of $q(t)$. By~\eqref{eq:box}, $\Omega$ is nonempty
and compact, and $T(\Omega)=\Omega$: if $q(t_{k})\to x$, then $q(t_{k}+1)\to Tx$,
and every limit point $y$ of $q(t_{k}-1)$ satisfies $Ty=x$. Let
$q_{-}$ minimize $V$ on $\Omega$. Since $T^{n}q_{-}\in\Omega$,
$V(q_{-})\le V(T^{n}q_{-})\le V(q_{-})$, so $q_{-}$ is a fixed point
by Step~2; thus $q^{*}$ exists and $q_{-}=q^{*}$. Let $q_{+}$
maximize $V$ on $\Omega$, and pick $p\in\Omega$ with $T^{n}p=q_{+}$.
Then $V(q_{+})=V(T^{n}p)\le V(p)\le V(q_{+})$, so $p=q^{*}$ and
$q_{+}=T^{n}p=q^{*}$. Hence $V\equiv V(q^{*})$ on $\Omega$, so
every $q\in\Omega$ satisfies $V(T^{n}q)=V(q)$ and is therefore equal
to $q^{*}$. Thus $\Omega=\{q^{*}\}$ and $q(t)\to q^{*}$. Finally,
$d_{t}=a\,q_{1}(t)$, so $D_{t}=a(1+q_{1}(t))\to a(1+q^{*}_{1})=\rho(K,L)$. 
\end{proof}

\subsection{Proof of Theorem \ref{thm:d_regular}}

We begin with a probability fact and then state and prove several graph theory lemmas.

\begin{lem}\label{l:Chernoffbounds}
Consider a binomially distributed random variable with $C_1 K^a$ trials and success probability $C_2 K^b$, where $a>0$ and $b<0$. For any $a+b<\gamma < a$, the probability of more than $K^{\gamma}$ successes decays at an exponential rate in $K$.
\end{lem}

\begin{proof}
First, suppose that $a+b < 0$. Then the mean $\mu$ of this random variable is vanishing. By the  Chernoff bound, we can bound the probability that it is greater than $K^{\gamma}$ above by $e^{-D(K^{\gamma-a}/C_1 || C_2K^b) C_1K^a}$ where $D(x||y)$ is the Kullback-Leibler divergence between Bernoulli random variables with success probabilities $x$ and $y$. This KL-divergence is \begin{align*}D(K^{\gamma-a}/C_1  || C_2K^b) & = K^{\gamma-a} \log(K^{\gamma - a - b}/C_2)/C_1 + (1-K^{\gamma-a}/C_1 )\log((1-K^{\gamma-a}/C_1)/(1-C_2K^b))
\\ & \geq K^{\gamma - a}
\end{align*}
for $K$ sufficiently large, where we use $\gamma - a -b > 0$ to obtain $\log(K^{\gamma - a - b}/C_2)/C_1> 1$ for $K$ large and observe the second term $(1-K^{\gamma-a}/C_1 )\log((1-K^{\gamma-a}/C_1)/(1-C_2K^b))$ vanishes.  So $$e^{-D(K^{\gamma-a}/C_1 || C_2K^b) C_1K^a} \leq e^{-C_1 K^{\gamma}}$$
as desired.

For $a + b \geq 0$, the mean of this random variable grows at a polynomial rate. So the result follows from a standard application of the Chernoff bound.
\end{proof}

Fix $j$ and $j'$. We want to bound the number of walks from $j$ and $j'$ to some node $k$ of lengths up to $t^*$. For each $t < t^*,$ let $\mathcal{B}_{jj',t}$ be the set of nodes $k$ such that the sum of the number of walks of length at most $t$ from $j$ to $k$ and the number of walks of length at most $t$ from $j'$ to $k$ is greater than $1$.

\begin{lem}\label{l:bad_set}
Given $\gamma \in (0,\alpha)$, there exists $C>0$ such that a.a.s. the cardinality of $\mathcal{B}_{jj',t}$ is at most $CK^{(t-1)\alpha+\gamma}$ for all $j \neq j'$ and all $t < t^*-1$ and at most $CK^{\max((t-1)\alpha + \gamma, 2 t\alpha -1+\gamma)}$ for all $j \neq j'$ when $t=t^*-1$.
\end{lem}

\begin{proof}

We can explore the network via breadth-first search (see, for example, \cite{olvera2022strong} for details). Whenever we explore from a node $k$ for the first time, we can take the set of neighbors of $k$ to be drawn uniformly at random independently of the neighborhoods of all other nodes.

Let  $\gamma \in (0,\alpha)$. We claim that for any $j \neq j'$ and any $t < t^*-1$, we can choose $C^{(t)}$ such that the probability that the cardinality of $\mathcal{B}_{jj',t}$ is greater than $C^{(t)}K^{(t-1)\alpha+\gamma}$ decays at an exponential rate in $K$. We establish this by induction on $t$.

Begin with the case $t=1$. We can chose the $2d$ neighbors of $j$ and $j'$ sequentially, and the cardinality of $\mathcal{B}_{jj',t}$ is bounded above by the number of times that we select a node which  has already been chosen or is one of the starting nodes $j$ or $j'$. The probability of this is bounded above at each step by $(2d+2)/K = 2K^{\alpha-1},$ so $|\mathcal{B}_{ij,t}|$ is bounded above by a binomial random variable with $2d$ trials and success probability $2K^{\alpha-1}$. By Lemma~\ref{l:Chernoffbounds}, the probability that this random variable is greater than $K^{\gamma}$ decays exponentially in $K$.

Now suppose that $\mathcal{B}_{jj',t-1}$ has cardinality at most $C^{(t)}K^{(t-2)\alpha +\gamma}$ and $t < t^*-1$. Each of these nodes has $d$ neighbors, so the number of nodes reached via walks of length $t$ with penultimate node in $\mathcal{B}_{jj',t-1}$ is  at most $C^{(t-1)}K^{(t-1)\alpha + \gamma}.$ Given a walk of length $t-1$ with penultimate node $k \notin \mathcal{B}_{jj',t-1}$, we can draw $d$ neighbors independently at random via our breadth-first search. We must add the node to $\mathcal{B}_{jj',t}$ if it  has already been chosen, and the total number of such nodes is bounded above by $2\sum_{\ell = 0}^{t} d^{\ell} < 3 d^{t}$. So the number of neighbors of such nodes $k$ that have already been chosen  is bounded above by the a binomial random variable with $2d^{t}=2K^{t\alpha}$ draws and success probability $3K^{t\alpha-1}.$ Since $t<t^*-1$, we have $(t+1)\alpha-1<0$. So by Lemma~\ref{l:Chernoffbounds}, the probability this random variable is greater than $K^{(t-1)\alpha + \gamma}$ decays exponentially in $K$. We conclude that $\mathcal{B}_{jj',t}$ has cardinality at most $(C^{(t-1)}+1)K^{(t-1)\alpha + \gamma}$. Taking $C^{(t)} = C^{(t-1)}+1$ completes the induction for $t<t^*-1$.

Finally, suppose $t=t^*-1$. As before, the number of nodes reached via walks of length $t-1$ with penultimate node in $\mathcal{B}_{jj',t-1}$ is  at most $C^{(t-1)}K^{(t-1)\alpha + \gamma}.$ Given a walk of length $t-1$ with penultimate node $k \notin \mathcal{B}_{jj',t}$, we can draw $d$ neighbors independently at random via our breadth-first search. As in the previous paragraph, we can bound the number of new nodes added to $\mathcal{B}_{jj', t}$ above by a binomial random variable with success probability $3K^{t\alpha-1}$ and $2K^{t\alpha}$ draws. By Lemma~\ref{l:Chernoffbounds}, this has success probability at most $K^{2 t\alpha -1+\gamma}$ away from an event with exponentially decaying probability. Combining these bounds and taking $C^{(t)} = C^{(t-1)}+1$  completes the $t=t^*-1$ case.

So the result holds outside of an event with exponentially decaying probability for each $j$, $j'$, and $t < t^*$. Since there are polynomially many such events, it holds asymptotically almost surely.
\end{proof}

\begin{lem}\label{l:max_walks}
Let $\nu > 0$. A.a.s., there are at most $K^{\nu}$ walks of length $t$ from $i$ to $k$ for all $t < t^*,$ $i$, and $k$.
\end{lem}

\begin{proof}
The case $t=1$ is trivial. Suppose there are at most $K^{\nu_t}$ walks of length $t$ from $i$ to $k$ for some $t$ and all $i$ and $k$. We will show that if $t +1 < t^*$, then for any $\nu_{t+1}>\nu_t$ there are at most $K^{\nu_{t+1}}$ walks of length $t + 1$ from $i$ to $k$ for all $i$ and $k$. To see this, note that there at most $d^t$ nodes $j$ reachable via a walk of length $t$ from $i$. By Lemma~\ref{l:Chernoffbounds}, the probability that greater than $K^{\nu_{t+1}-\nu_t}$ of these has a link to $k$ decays exponentially in $K$. So we can assume this event does not hold for any $i$ and $j$ a.a.s. For each of these nodes $j$, there are at most $K^{\nu_t}$ walks of length $t$ from $i$ to $j$. So there are at most $K^{\nu_{t+1}}$ walks of length $t+1$ from $i$ to $k$. The result follows from choosing any sequence
$0 < \nu_1 < \nu_2 < \hdots < \nu_{t^*-1} \leq \nu.$
\end{proof}

Our final lemma shows that the number of walks of length $t^*$ between any two nodes is close to its expected value.

\begin{lem}\label{l:tstar_walks}
For some $\eta \in (0, (\alpha t^*-1)/2)$, a.a.s. the number of walks of length $t^*$ from $j$ to $k$ is in $$[(1-K^{-\eta}) K^{\alpha t^*-1}, (1+K^{-\eta}) K^{\alpha t^*-1}]$$ for all $j$ and $k$.
\end{lem}
\begin{proof}
For each $t < t^*$, let $\mathcal{B}_{j,t}$ be the set of nodes $k$ such that there is more than one walk of length at most $t$ from $j$ to $k$. For any $j \neq j'$, $\mathcal{B}_{j,t} \subset \mathcal{B}_{jj',t}$. So by Lemma~\ref{l:bad_set}, we can choose $C>0$ and $\gamma \in (0,\alpha)$ such that $|\mathcal{B}_{j,t}| < CK^{\max((t-1)\alpha + \gamma, 2 t\alpha -1+\gamma)}$ a.a.s. We condition on this event holding for $t=t^*-1$.

Choose $\nu < \alpha - \gamma$. By Lemma~\ref{l:max_walks}, we can also condition on the event that there are at most $K^{\nu}$ walks of length $t$ from $i$ to $k$ for all $t < t^*,$ $i$, and $k$.

For each $i$, there are $d^{t^*-1}$ walks of length $t^*-1$ from $i$. At most $CK^{\max((t^*-2)\alpha + \gamma, 2 (t^*-1)\alpha -1+\gamma)+\nu}$ of these walks end at nodes in $\mathcal{B}_{j, t^*-1}$ because $|\mathcal{B}_{j, t^*-1}|\leq CK^{\max((t^*-2)\alpha + \gamma, 2 (t^*-1)\alpha -1+\gamma)}$. The remaining walks have distinct endpoints outside this set, and we can call the set of such endpoints $\mathcal{C}_j$. For suitable $\gamma$ and $\nu$, the exponent $\max((t^*-2)\alpha + \gamma, 2 (t^*-1)\alpha -1+\gamma)+\nu$ is less than $(t^*-1)\alpha.$

Fix some $k$. Recall that $|\mathcal{C}_j| \leq d^{t^*-1}=K^{\alpha (t^*-1)}$. By a Chernoff bound, the number of nodes in $\mathcal{C}_j$ with a link to $k$ is outside $[(1-K^{-\eta}/2) K^{\alpha t^*-1}, (1+K^{-\eta}/2) K^{\alpha t^*-1}]$ with an exponentially decaying probability for some $\eta \in (0, (\alpha t^*-1)/2)$.

To complete the proof, we must bound the number of walks from $j$ to $k$ of length $t^*$ with penultimate node in  $\mathcal{B}_{j, t^*-1}$. There are at most  $CK^{\max((t^*-2)\alpha + \gamma, 2 (t^*-1)\alpha -1+\gamma)}$ nodes in $\mathcal{B}_{j, t^*-1}.$ We have not yet assumed anything about the outgoing links from these nodes except ruling out events with exponentially decaying probabilities. So applying Lemma~\ref{l:Chernoffbounds} again, the number of nodes in $\mathcal{B}_{j, t^*-1}$ with links to $k$ is bounded above by $K^{\max(\alpha(t^*-1), 2\alpha(t^*-1)-1)+2\gamma-1}$ outside of an event with exponentially decaying probability. Each of these nodes $k'$ with such a link contributes at most $K^{\nu}$ walks by Lemma~\ref{l:max_walks}, because there are at most $K^{\nu}$ paths of length $t^*-1$ from $j$ to $k'$. So the number of walks from $j$ to $k$ of length $t^*$ with penultimate node in  $\mathcal{B}_{j, t^*-1}$ is at most $K^{\max(\alpha(t^*-1),2\alpha(t^*-1)-1)+2\gamma+\nu-1,}$. The total number of such walks is therefore at most $K^{-\eta} K^{\alpha t^*-1}/2$ for suitable choices of $\eta,$ $\gamma$ and $\nu$.

Combining the bound on the number of walks with penultimate node in $\mathcal{C}_j$ and the bound on the number of walks with penultimate node in $\mathcal{B}_{j,t^*-1}$ proves the lemma.
\end{proof}

We can now prove the theorem. A main step is to bound the weights used by agents in generations $2,\hdots,t^*+1$.

For each generation $t$, we let $r(t) = d^{t-1}$ and let $v(t) =0$ if $2\alpha(t-1)<1$ and $v(t) = d^{2(t-1)}/K$ if $2\alpha(t-1)>1$. Define a weight $w(t) = \frac{r(t)}{r(t) + (d-1)v(t)}$. We claim we can choose $C_t>0$ and $\delta_t >0$ for each $1 \leq t \leq t^*$ such that a.a.s., all agents $i$ in generation $t+1$ place weights $\beta_{ij} \in [w(t)(1-C_tK^{-\delta_t}), w(t)(1+C_tK^{-\delta_t})]$ on each neighbor $j$.

We set $\delta_1=\delta_2=\hdots=\delta_{t^*-1}$ for any value satisfying $\delta_1<\alpha$ and $\delta_1 > \alpha t^*-1$. Since $\alpha t^*-1 < \alpha$, this is possible for any permissible $\alpha$. We choose any $\delta_{t^*}$ such that $$\delta_{t^*-1}>  \alpha t^*  -1 +\delta_{t^*},$$
which is possible because  $\delta_{t^*-1}>  \alpha t^*  -1$. To establish the bound on weights, we induct on $t$. In period $2$, all agents place weight $1$ on each observation. So the claim holds for any $\delta_1$, and in particular for our choice of $\delta_1$ above.

Suppose the claim holds up to period $t-1$ for $t \leq t^*$. We can condition on the event that each $i$ in generations $t' \leq t$ places weight $\beta_{ij} \in [w(t'-1)(1-C_{t'-1}K^{-\delta_{t'-1}}), w(t'-1)(1+C_{t'-1}K^{-\delta_{t'-1}})]$ on each neighbor $j$. Throughout the proof of the inductive step, we will condition on events that hold except with exponentially decaying probabilities by Lemmas~\ref{l:Chernoffbounds}-\ref{l:tstar_walks} and variants of their arguments. Because there are polynomially many of these events, this will imply the inductive step holds a.a.s. for each $1 \leq t \leq t^*$.

We want to bound the variance and covariance of agents in generation $t$. We begin with variances. Let $j$ be in generation $t$. By Proposition \ref{prop:signal_counting}, the variance of $\ell_j$ is $4r_j/\sigma^2$, so it is sufficient to characterize $r_j$. This in turn is equal to the sum of the weights $j$ places on private signals. By our inductive hypothesis, the total weight placed by each agent in generation $t' \leq t$ on predecessors is in $[dw(t'-1)(1-C_{t'-1}K^{-\delta_{t'-1}}), dw(t'-1)(1+C_{t'-1}K^{-\delta_{t'-1}})]$. So we can choose $C'$ such that $ d^{t-1} (1-C'K^{-\delta_{t'-1}}) \leq r_i \leq d^{t-1}(1+C'K^{-\delta_{t'-1}})$.

Suppose $i$ is in generation $t+1$ and observes agent $j$ in generation $t$. We now want to bound the covariances between $j$ and other agents $j'$ in $N_i$. By Lemma~\ref{l:Chernoffbounds}, there are at most $K^{\nu}$ such agents $j'$ who share a neighbor with $j$. 

We begin by analyzing the covariances with agents $j'$ who do not have a common neighbor with $j$. We will argue as in Lemma~\ref{l:bad_set} to bound the covariance between $j$ and $j'$. Because $j$ and $j'$ have no common neighbors, $\mathcal{B}_{jj',1}$ is empty. If $2\alpha (t'-1)<1 $, then arguing as in the proof of Lemma~\ref{l:bad_set}, we obtain that $\mathcal{B}_{jj', t'}$ has cardinality at most $CK^{\alpha(t'-2)+\gamma}$. By Lemma~\ref{l:max_walks} and the bound on weights, each element of this set contributes at most $8CK^{2\nu}/\sigma^2$ to the covariance between $j$ and $j'$. So for each $j'$ without a common neighbor, the covariance $j$ and $j'$ is at most $16 CK^{(t-3)\alpha + \gamma + 2\nu}/\sigma^2$.

Suppose that $2\alpha(t'-1)>1$. When $2\alpha (t'-1)>1$, as in the proof of Lemma~\ref{l:bad_set}, we can express $\mathcal{B}_{jj', t'}$ as the union of neighbors of nodes in $\mathcal{B}_{jj', t'-1}$ and nodes reached for the first time at depth $t'$ in our breadth-first search. By a Chernoff bound, the probability that the cardinality of the set of nodes reached for the first time at distance $t'-1$ by walks from both $j$ and $j'$ is outside $[(1-K^{-\nu'})K^{2\alpha(t'-1)-1}, K^{2\alpha(t'-1)-1}(1+K^{-\nu'})]$ decays exponentially in $K$ for $0 < \nu' < \alpha(t'-1)-\frac12$. So the cardinality of the set of neighbors of nodes in $\mathcal{B}_{jj', t'-1}$ is bounded above by the sum of a term $C K^{\alpha(t'-2)+\gamma}$ and a term $K^{2 \alpha(t''-2) +(t'-t'')\alpha-1}(1+K^{-\nu'})$ for each $t'' < t'$ satisfying $\alpha (t''-1)>1$. The case different from the previous paragraph is when the highest order term in the covariance between $j$ and $j'$ comes from nodes reached for the first time at depth $t-1$ in the breadth-first search.

We now want to analyze the contribution to covariance from the set of nodes in $\mathcal{B}_{jj', t-1}$ reached for the first time at depth $t-1$ by walks from both $j$ and $j'$. Recall the cardinality of this set is $[(1-K^{-\nu'})K^{2\alpha(t-1)-1}, K^{2\alpha(t-1)-1}(1+K^{-\nu'})]$. By Lemma~\ref{l:Chernoffbounds}, at most a fraction $K^{-\nu'+2\eta}$ of these nodes can be reached by more than one walk from $j$ or more than one walk from $j'$. By Lemma~\ref{l:max_walks}, there are at most $K^{\eta}$ walks from each of $j$ and $j'$ to any such node. The contribution to the covariance from these terms is therefore at most $8 K^{-\nu}K^{2\alpha(t-1)-1}/\sigma^2$. For the remaining nodes, there is exactly one path from $j$ and exactly one path from $k$ of length $t-1$. The contribution to covariance from each is the product of the weights on these paths, which (after increasing $C''$ if needed) is in $[1-C'' K^{-\delta_{t-1}}, 1+C''K^{-\delta_{t-1}}]$. Increasing $C''$ further, we can conclude that the covariance between $j$ and $j'$ is within a factor in $[1-C''K^{-\delta_{t-1}}, 1+ C''K^{-\delta_{t-1}}]$ of $4v(t)/\sigma^2$ if $\delta_{t-1} \leq \nu'$ and within  a factor in $[1-C''K^{-\nu'}, 1+ C''K^{-\nu'}]$ of $4v(t)/\sigma^2$ otherwise.

It remains to analyze the covariances of $j$ with $j'$ who have common neighbors with $j$. When $t < t^*$, for any such agents, we can assume by Lemma~\ref{l:bad_set} that $|\bigcup_{t'< t} \mathcal{B}_{jj',t'}| < 2C K^{(t-2)\alpha + \gamma}$ for any $\gamma>0$ and some $C>0$. By Lemma~\ref{l:max_walks} and the bound on weights, each element of this set contributes at most $8K^{2\nu}/\sigma^2$ to the covariance between $j$ and $j'$. So for each $j'$ with a common neighbor, the covariance between $j$ and $j'$ is at most $16C K^{(t-2)\alpha + \gamma + 2\nu}/\sigma^2$.

When $t=t^*$, the contribution to $|\mathcal{B}_{jj',t-1}|$ from neighbors of nodes in $\mathcal{B}_{jj',t-2}$ is at most $2C K^{(t-2)\alpha + \gamma}$. This again contributes at most $16C K^{(t-2)\alpha + \gamma + 2\nu}/\sigma^2$ to covariance. There is also another term capturing the contribution from covariance from nodes in $\mathcal{B}_{jj', t-1}$ reached for the first time at depth $t-1$, which we bound as in the case when $j$ and $j'$ do not have common neighbors above.

This completes our bounds on variances and covariances.

We are now ready to bound the change in weights. We first bound the deviation between weights and then bound the average weight.

By Proposition \ref{prop:linear}, the weights are \[
\vec{\beta}_{i}=2\left(\mathbb{E}[(\ell_{j(1)},...,\ell_{j(d)})\mid\omega=1]\times\textsc{\emph{Cov}}[\ell_{j(1)},...,\ell_{j(d_{i})}\mid\omega=1]^{-1}\right),
\]
where $j(1),\hdots,j(d)$ are the neighbors of $i$. The entries of the expectation are $$2r_{j(1)}/\sigma^2,\hdots,2r_{j(d)}/\sigma^2 \in [2(1-C'K^{-\delta_{t'-1}})d^{t-1}/\sigma^2,2(1+C'K^{-\delta_{t'-1}})d^{t-1}/\sigma^2].$$ We can write 
$$\textsc{\emph{Cov}}[\ell_{j(1)},...,\ell_{j(d_{i})}\mid\omega=1] = \frac{4}{\sigma^2}[(r(t)-v(t))I +v(t)\mathbf{1}\mathbf{1}^T + U],$$
we obtain
$$\textsc{\emph{Cov}}[\ell_{j(1)},...,\ell_{j(d_{i})}\mid\omega=1]^{-1} = \frac{\sigma^2}{4}\cdot [((r(t)-v(t))I +v(t)\mathbf{1}\mathbf{1}^T)^{-1} (I -\frac{4}{\sigma^2} U \textsc{\emph{Cov}}[\ell_{j(1)},...,\ell_{j(d_{i})}\mid\omega=1]^{-1})].$$

We can express $((r(t)-v(t))I +v(t)\mathbf{1}\mathbf{1}^T)^{-1}$ as the sum of $ (r(t)-v(t))^{-1}I$ and a multiple of $\mathbf{1}\mathbf{1}^T$. Terms involving the multiple of $\mathbf{1}\mathbf{1}^T$ contribute to all weights $\vec{\beta}_i$ equally, so we can bound the deviation across weights by analyzing $$\frac{\sigma^2}{4}\cdot [((r(t)-v(t))^{-1} (I - U \textsc{\emph{Cov}}[\ell_{j(1)},...,\ell_{j(d_{i})}\mid\omega=1]^{-1})].$$
We have
$$\frac{\sigma^2}{2} \cdot \mathbb{E}[(\ell_{j(1)},...,\ell_{j(d)})\mid\omega=1] [((r(t)-v(t))^{-1} (I - U \textsc{\emph{Cov}}[\ell_{j(1)},...,\ell_{j(d_{i})}\mid\omega=1]^{-1})],$$
which is equal to 
$$\frac{\vec{r}}{r(t)-v(t)} \cdot(I - U \textsc{\emph{Cov}}[\ell_{j(1)},...,\ell_{j(d_{i})}\mid\omega=1]^{-1}).$$
Recall that $v(t) = K^{2\alpha(t-1)-1}$ and $r(t) = K^{\alpha (t-1)}$. 

Then
\begin{align*}
\beta_{ij} & \leq \frac{\max_{j \in N_i} r_j}{r(t)-v(t)} (1 + \|U \textsc{\emph{Cov}}[\ell_{j(1)},...,\ell_{j(d_{i})}\mid\omega=1]^{-1}\|_{\infty})
\\ & \leq \frac{\max_{j \in N_i} r_j}{r(t)-v(t)} (1 + \|U\|_{\infty} \|\textsc{\emph{Cov}}[\ell_{j(1)},...,\ell_{j(d_{i})}\mid\omega=1]^{-1}\|_{\infty} ).
\end{align*}
For the first term, $r_j \leq r(t)(1+C'K^{-\delta_{t'-1}})$ and $v(t)$ is a lower-order polynomial than $r(t)$. We next use our bounds on covariances above to bound the norm $\|U\|_{\infty}$. Given $j(k)$, the diagonal entry  $|U_{kk}|$ is $ O(K^{-\delta_{t'-1}}r(t))$. There are at most $K^{\nu}$ other $j(k')$ sharing a common neighbor with $j(k)$, and for each of these we showed the covariance is $O(K^{(t-2)\alpha + \gamma + 2\nu})$ since $\delta_{t'-1} < \alpha$ or can be grouped with the $j(k')$ not sharing a common neighbor. Since $\gamma$ and $\nu$ can be arbitrary, we can choose them small enough so that the total contribution of the $O(K^{(t-2)\alpha + \gamma + 2\nu})$ terms is also $O(K^{-\delta_{t'-1}}r(t))$. There are fewer than $d$ remaining covariances, and we showed above that the error term $U_{kk'}$ for each is $$O(\max(K^{(t-3)\alpha + \gamma + 2\nu}, K^{-\nu'}v(t), K^{-\delta_{t'-1}} v(t))).$$ When the first term in the maximum applies, the total contribution is $O(K^{(t-2)\alpha + \gamma + 2\nu})$ which is again $O(K^{-\delta_{t'-1}}r(t))$. When the second term applies, the total contribution is $O(dK^{-\nu'}v(t))$. Since we can take $\nu'$ arbitrarily close to $\alpha (t-1)-\frac12$, we can take this to be $O(K^{\alpha t-\frac12+\epsilon})$ for any $\epsilon>0$. This in turn is $O(K^{-\delta_{t'-1}}r(t))$ as long as $\delta_{t'} + \alpha < \frac12$, which holds since $\delta_{t'-1}<\alpha$ and $\alpha < \frac14$. When the third term applies, the total contribution is $O(dK^{-\delta_{t'-1}}v(t))$. When $t<t^*$, this is $O(K^{-\delta_{t'}}r(t))$ because $dv(t) < r(t)$. When $t=t^*$, this is $O(K^{-\delta_{t'}}r(t))$ since $$\delta_{t^*-1}>  \alpha t^*  -1 +\delta_{t^*}.$$
We have
$$ \|\textsc{\emph{Cov}}[\ell_{j(1)},...,\ell_{j(d_{i})}\mid\omega=1]^{-1}\|_{\infty} \leq \frac{\|((r(t)-v(t))I +v(t)\mathbf{1}\mathbf{1}^T)^{-1} \|_{\infty}}{1-\|((r(t)-v(t))I +v(t)\mathbf{1}\mathbf{1}^T)^{-1}\|_{\infty} \|U\|_{\infty}}.$$
A simple computation gives $\|((r(t)-v(t))I +v(t)\mathbf{1}\mathbf{1}^T)^{-1}\|_{\infty} \leq \frac{2}{r(t)-v(t)}$ while $\|U\|_{\infty}$ is $o(r(t))$, so the right-hand side is $O(K^{-\alpha(t-1)})$.

Combining these terms and recalling that $\delta_{t'-1} \geq \delta_{t'}$, our upper bound on $\beta_{ij}$ is $1+O(K^{-\delta_{t'}})$. We obtain a lower bound similarly. So the deviation in weights $\beta_{ij}$ across $j$ are $O(K^{-\delta_{t'}})$.

We now turn to bounding the average weight. Proposition \ref{prop:signal_counting} requires that
$$\sum_{j} \beta_{ij} r_j = \sum_{j,j'} \beta_{ij}\vec{\beta}_{ij'}\textsc{\emph{Cov}}[\ell_j, \ell_{j'}].$$
We have $\beta_{ij} = \beta(1 + O(K^{-\delta_{t'}}))$ for some $\beta$. Recall $r_j = r(t)(1 + O(K^{-\delta_{t'}}))$ and $\sum_{j'}\textsc{\emph{Cov}}[\ell_j, \ell_{j'}]=r(t)+(d-1)v(t) + O(r(t)K^{-\delta_{t'}})$ for each $j$. We obtain
$$ d\beta r(t)(1+ O(K^{-\delta_{t'}})) = d\beta^2 (r(t) + (d-1)v(t)+O(K^{-\delta_{t'}} r(t)))$$
and therefore $\beta = \frac{r(t)}{r(t) + (d-1)v(t)} \cdot (1+O(K^{-\delta_{t'}}))$ as desired. This completes the induction, so we have established our bounds on weights $w(t)$ for each $1\leq t \leq t^*$.

Now we can apply our bounds on weights along with Lemma~\ref{l:tstar_walks} to bound the weight each agent in generation $t^*+1$ places on each agent in generation $1$, which we will show is close  to $1$. Recall from the proof of Proposition \ref{prop:linear} that we use the notation $w_{j,k}$ for the weight placed by agent $j$ on $k$'s log-signal.  For all $j$ in generation $t^*$ and $k$ in generation $1$ and some $\eta \in (0, (\alpha t^*-1)/2)$, by Lemma~\ref{l:tstar_walks}, a.a.s. the number of walks from $j$ to $k$ is in  $$[(1-K^{-\eta}) K^{\alpha t^*-1}, (1+K^{-\eta}) K^{\alpha t^*-1}].$$ Then letting $w^* = \prod_{t=1}^{t^*} w(t)$, the product of weights along any of these walks is within a factor $1+O(K^{-\delta_{t'}})$ of $w^*$. It follows from the construction of each $w(t)$ that $w^* = K^{1-\alpha t^*}(1+O(K^{-\epsilon_1}))$ for some $\epsilon_1>0$. So we have $K^{\alpha t^*-1}(1+O(K^{-\eta}))$ walks from $j$ to $k$ and each has weight $w^*(1 + O(K^{-\epsilon_1}))$. Shrinking $\epsilon_1$ if necessary, we conclude that $w_{j,k} = 1 + O(K^{-\epsilon_1})$. It follows that $r_j \geq K -CK^{1-\epsilon_1}$ for some $C>0$.

Now consider arbitrary $i$ in generation $t^*+2$. We will bound the total weight $\sum_{k = K+1}^{2K} w_{i,k}$ that agent $i$ places on agents in generation $2$. We can relate this to the total weight $\sum_{k=1}^K w_{i,k}$ that agent $i$ places on agents in generation $1$:
$$ \sum_{k=1}^K w_{i,k} = d\sum_{k = K+1}^{2K} w_{i,k}.$$
The variance of $\ell_i$ is $4r_i/\sigma^2$ and $r_i$ cannot be larger than the total number of agents reachable by walks from $i$, which is less than $3K$. On the other hand the variance is bounded below by
$$\sum_{k=1}^K 4w_{i,k}^2/\sigma^2.$$
Given $\sum_{k=1}^K w_{i,k} \geq \underline{w},$ the variance is minimized by setting all $w_{i,k}$ equal to $\underline{w}/K$, which gives variance
$$ \frac{4}{\sigma^2} \sum_{k=1}^K\left(\frac{\underline{w}}{K}\right)^2.$$
The summation cannot be larger than $r_i \leq 3K$, so we must have $\underline{w} < 2K$. Therefore
$$\sum_{k = K+1}^{2K} w_{i,k} \leq \frac{2K}{d}.$$
But then the signal-counting property can only hold if  $r_i \leq K + O(K^{1-\epsilon_2})$ for some $\epsilon_2>0$. Taking $\epsilon < \min(\epsilon_1,\epsilon_2)$ proves the theorem.

\subsection{Proof of Proposition \ref{prop:starting_3}}
\begin{proof}
We first establish a lemma that expresses $\vec{\beta}_{i,\cdot}$
in closed-form for an agent $i$ in generation $t+1$. Let $\ell_{\text{sum}}$
be the sum of the log-actions played in generation $t-1$.
By the log-linearity of rational strategies (Proposition \ref{prop:linear}),
there must exist some $\mu_{\text{sum}},\sigma_{\text{sum}}^{2}>0$
so that the conditional distributions of $\ell_{\text{sum}}$ in the
two states are $\mathcal{N}(\pm\mu_{\text{sum}},\sigma_{\text{sum}}^{2}).$
\begin{lem}
\label{lem:strategy_in_gens} Each element in $\vec{\beta}_{i,\cdot}$
is $\left(\frac{\mu_{\text{sum }}^{2}}{\sigma_{\text{sum}}^{2}}+\frac{1}{\sigma^{2}}\right)/\left(K\frac{\mu_{\text{sum }}^{2}}{\sigma_{\text{sum}}^{2}}+\frac{1}{\sigma^{2}}\right)$.
\end{lem}
\begin{proof}
An application of Proposition \ref{prop:linear} shows each agent
$j$ in generation $t$ aggregates $\ell_{\text{sum}}$ and own private
signal $\lambda_{j}$ according to $\ell_{j}=2\cdot\frac{\mu_{\text{sum }}}{\sigma_{\text{sum}}^{2}}\ell_{\text{sum}}+\lambda_{j}.$

Next, consider the problem of someone in generation $t+1$ who observes
the log-actions $\ell_{j}$ of the $K$ agents $j=(t-1)K+k$ for $1\le k\le K$
from generation $t.$ By symmetry, $i$ places the same weight on
these $K$ log-actions. To find this weight, we calculate
\begin{align*}
\mathbb{E}\left[\sum_{k=1}^{K}\ell_{(t-1)K+k}\mid\omega=1\right] & =2K\frac{\mu_{\text{sum }}^{2}}{\sigma_{\text{sum}}^{2}}+2K\frac{1}{\sigma^{2}}\\
\textsc{Var}\left[\sum_{k=1}^{K}\ell_{(t-1)K+k}\mid\omega=1\right] & =K\cdot\left(4\cdot\frac{\mu_{\text{sum }}^{2}}{\sigma_{\text{sum}}^{2}}+4\cdot\frac{1}{\sigma^{2}}\right)+K\cdot(K-1)\cdot4\cdot\frac{\mu_{\text{sum }}^{2}}{\sigma_{\text{sum}}^{2}}
\end{align*}
 So by Proposition \ref{prop:linear}, 
\[
\beta_{i,j}=\frac{2\cdot\left(2K\frac{\mu_{\text{sum }}^{2}}{\sigma_{\text{sum}}^{2}}+2K\frac{1}{\sigma^{2}}\right)}{K\cdot\left(4\cdot\frac{\mu_{\text{sum }}^{2}}{\sigma_{\text{sum}}^{2}}+4\cdot\frac{1}{\sigma^{2}}\right)+K\cdot(K-1)\cdot4\cdot\frac{\mu_{\text{sum }}^{2}}{\sigma_{\text{sum}}^{2}}}=\frac{\frac{\mu_{\text{sum }}^{2}}{\sigma_{\text{sum}}^{2}}+\frac{1}{\sigma^{2}}}{K\frac{\mu_{\text{sum }}^{2}}{\sigma_{\text{sum}}^{2}}+\frac{1}{\sigma^{2}}}
\]
 for every $j=(t-1)K+k$ for $1\le k\le K$, as desired.
\end{proof}
Consider an agent $i$ in generation $t.$ From Proposition \ref{prop:signal_counting},
there is some $x_{old}>0$ so that $\ell_{i}\sim\mathcal{N}(\pm x_{old},2x_{old})$
conditional on the two states. In fact, from Proposition \ref{prop:linear},
$x_{old}=2\cdot\frac{\mu_{\text{sum }}^{2}}{\sigma_{\text{sum}}^{2}}+\frac{2}{\sigma^{2}}.$
For an agent in generation $t+1,$ using the same argument and applying
the formula for $\vec{\beta}_{i,\cdot}$ from Lemma \ref{lem:strategy_in_gens},
we have $x_{new}=\frac{2K(\frac{\mu_{\text{sum }}^{2}}{\sigma_{\text{sum}}^{2}}+\frac{1}{\sigma^{2}})^{2}}{K\frac{\mu_{\text{sum }}^{2}}{\sigma_{\text{sum}}^{2}}+\frac{1}{\sigma^{2}}}+\frac{2}{\sigma^{2}}.$

A hypothetical agent who observes $\ell_{\text{sum}}$ (the sum of
log-actions in generation $t-1$) with conditional distributions $\mathcal{N}(\pm\mu_{\text{sum}},\sigma_{\text{sum}}^{2})$
and three extra independent private signals (in addition to the one
private signal usually observed) would play a log-action with conditional
distributions $\mathcal{N}(\pm y,2y)$ where $y=\left[2\frac{\mu_{\text{sum }}^{2}}{\sigma_{\text{sum}}^{2}}+\frac{6}{\sigma^{2}}\right]+\frac{2}{\sigma^{2}}.$
We have
\begin{align*}
(y-x_{new})\cdot(K\frac{\mu_{\text{sum }}^{2}}{\sigma_{\text{sum}}^{2}}+\frac{1}{\sigma^{2}})= & \left[2\frac{\mu_{\text{sum }}^{2}}{\sigma_{\text{sum}}^{2}}+\frac{6}{\sigma^{2}}\right]\cdot\left[K\frac{\mu_{\text{sum }}^{2}}{\sigma_{\text{sum}}^{2}}+\frac{1}{\sigma^{2}}\right]-2K(\frac{\mu_{\text{sum }}^{2}}{\sigma_{\text{sum}}^{2}}+\frac{1}{\sigma^{2}})^{2}\\
= & (2+4K)\cdot\frac{\mu_{\text{sum }}^{2}}{\sigma_{\text{sum}}^{2}}\cdot\frac{1}{\sigma^{2}}+\frac{6}{\sigma^{4}}-4K\cdot\frac{\mu_{\text{sum }}^{2}}{\sigma_{\text{sum}}^{2}}\cdot\frac{1}{\sigma^{2}}-2K\frac{1}{\sigma^{4}}\\
\ge & 2K\frac{1}{\sigma^{2}}\left(\frac{\mu_{\text{sum }}^{2}}{\sigma_{\text{sum}}^{2}}-\frac{1}{\sigma^{2}}\right).
\end{align*}

We must have $\mathbb{P}[\ell_{\text{sum}}>0\mid\omega=1]\ge\mathbb{P}[\lambda_{1}>0\mid\omega=1]$,
a probability that just depends on the ratio of the mean and standard
deviation. So $\frac{\mu_{\text{sum }}}{\sigma_{\text{sum}}}\ge\frac{1}{\sigma}$,
i.e. $\frac{\mu_{\text{sum }}^{2}}{\sigma_{\text{sum}}^{2}}\ge\frac{1}{\sigma^{2}}$.
Hence the difference above is positive. This shows $x_{new}-x_{old}\le3\cdot\frac{2}{\sigma^{2}}$.
\end{proof}

\subsection{Proof of Proposition \ref{prop:zeroth_gen}}
\begin{proof}
It is clear that $r_{0}=1$ and that $r_{i}=2$ for $1\le i\le K$.
By applying Example \ref{ex:Many-Neighbors-with}
with $K_{1}=1,$ $K_{2}=K,$ we see that for every agent $j$ in generation
2, we get $r_{j}\le r_{1}+1+\frac{(K-1)(1+1)}{K+1}<5$. For an agent
$j'$ in generation $t\ge3$, the same arguments in the proof of Proposition
\ref{prop:starting_3} apply, showing that $r_{j'}-r_{i'}\le3$ where
$i'$ is any agent in generation $t-1.$
\end{proof}

\subsection{Proof of Proposition \ref{prop:lower_bound}}
\begin{proof}
    Clearly, this result holds for $2\le t\le L$, as agents in these
early generations perfectly aggregate $K$ signals in each generation.
The proof of Theorem \ref{thm:L_gen} defines the quantity $D_{t}=K(r_{t}-r_{t-1})-(K-1)$
for $t\ge L+1.$ Lemma \ref{lem:D_and_b} implies that $D_{t}\ge K$ for every $t\ge L+1.$
This implies $K(r_{t}-r_{t-1})-(K-1)\ge K\Rightarrow r_{t}-r_{t-1}\ge2-1/K$.
\end{proof}

\subsection{Proof of Proposition~\ref{prop:optimal_silos}}
\begin{proof}
We know from Corollary \ref{cor:silo} that $g_{k}=2-\frac{1}{|S_{n}|}$ for $k\in S_{n}$
and $g_{K+1}=\sum_{n=1}^{N}2-\frac{1}{|S_{n}|}$. 

\selectlanguage{english}%
Let $s_{n}:=|S_{n}|>0$ denote silo sizes. Then $\sum_{n=1}^{N}s_{n}=K$.
Summing workers' $g_{k}$'s gives: 
\[
\sum_{k=1}^{K}g_{k}=\sum_{n=1}^{N}\sum_{k\in S_{n}}\left(2-\frac{1}{s_{n}}\right)=\sum_{n=1}^{N}s_{n}\left(2-\frac{1}{s_{n}}\right)=\sum_{n=1}^{N}(2s_{n}-1)=2K-N.
\]
For position $K+1$, we get 
\[
g_{K+1}=\sum_{n=1}^{N}\left(2-\frac{1}{s_{n}}\right)=2N-\sum_{n=1}^{N}\frac{1}{s_{n}}.
\]
Hence the organization's objective can be written as 
\begin{equation}
\sum_{k=1}^{K}g_{k}+\alpha g_{K+1}=(2K-N)+\alpha\left(2N-\sum_{n=1}^{N}\frac{1}{s_{n}}\right)=\sum_{n=1}^{N}f_{\alpha}(s_{n}),\label{eq:obj}
\end{equation}
where for each integer $s\ge1$ we define the per-silo value 
\[
f_{\alpha}(s):=(2s-1)+\alpha\left(2-\frac{1}{s}\right).
\]
Thus the problem is: choose $N\in\{1,\dots,K\}$ and integers $s_{1},\dots,s_{N}\ge1$
with $\sum s_{n}=K$ to maximize $\sum_{n=1}^{N}f_{\alpha}(s_{n})$.

\noindent\textbf{(1) In any optimal organization, any two silo sizes
differ by at most $1$.} Fix $\alpha \ge 1$. Consider an organization
structure with two silos with sizes $s$ and $s'$ where $s'\ge s+2.$
Consider transferring one agent from the larger silo to the smaller
silo. The change in objective is 
\[
\bigl[f_{\alpha}(s'-1)+f_{\alpha}(s+1)\bigr]-\bigl[f_{\alpha}(s')+f_{\alpha}(s)\bigr]=\alpha\left(\frac{1}{s(s+1)}-\frac{1}{s'(s'-1)}\right).
\]
Since $s'\ge s+2$ implies $s'(s'-1)\ge(s+2)(s+1)>s(s+1)$, we have
$\frac{1}{s(s+1)}-\frac{1}{s'(s'-1)}>0$, and therefore the change
is strictly positive. Hence any feasible profile with two silos differing
by at least $2$ can be strictly improved, so in any optimum all silo
sizes differ by at most $1$.

\medskip{}
\noindent\textbf{(2) Optimal number of silos is weakly increasing in $\alpha$.}
For each $N\in\{1,\dots,K\}$ define the best achievable value with
exactly $N$ silos: 
\[
F_{N}(\alpha):=\max_{\substack{s_{1}+\cdots+s_{N}=K\\
s_{n}\ge1
}
}\sum_{n=1}^{N}f_{\alpha}(s_{n}).
\]
From (\ref{eq:obj}), 
\[
\sum_{n=1}^{N}f_{\alpha}(s_{n})=(2K-N)+\alpha\left(2N-\sum_{n=1}^{N}\frac{1}{s_{n}}\right).
\]
For fixed $N$, maximizing this expression is equivalent to minimizing
$\sum_{n=1}^{N}1/s_{n}$ subject to $\sum_{n=1}^{N}s_{n}=K$. Let
\[
H_{N}:=\min_{\substack{s_{1}+\cdots+s_{N}=K\\
s_{n}\ge1
}
}\sum_{n=1}^{N}\frac{1}{s_{n}}.
\]
Then 
\[
F_{N}(\alpha)=(2K-N)+\alpha(2N-H_{N}),
\]
a linear function of $\alpha$ with slope $B_{N}:=2N-H_{N}$.

\smallskip{}
\noindent We claim that the slopes $(B_{N})_{N=1}^{K}$ are strictly increasing:
$B_{N+1}>B_{N}$ for all $N\le K-1$.

\smallskip{}
\noindent To see this, fix $N\le K-1$ and let $(s_{1},\dots,s_{N})$ attain
$H_{N}$. Since $\sum_{n=1}^{N}s_{n}=K$ and $N\le K-1$, at least
one silo has size $\hat{s}\ge2$. Split that silo into sizes $1$
and $\hat{s}-1$, producing an $(N+1)$-silo feasible profile. The
reciprocal sum increases by 
\[
\left(1+\frac{1}{\hat{s}-1}\right)-\frac{1}{\hat{s}}=1+\frac{1}{\hat{s}(\hat{s}-1)}\le1+\frac{1}{2}=\frac{3}{2}.
\]
Therefore $H_{N+1}\le H_{N}+\frac{3}{2}$, and so 
\[
B_{N+1}-B_{N}=\bigl(2(N+1)-H_{N+1}\bigr)-\bigl(2N-H_{N}\bigr)=2-(H_{N+1}-H_{N})\ge2-\frac{3}{2}=\frac{1}{2}>0.
\]

\smallskip{}
\noindent Now take $\alpha_{2}>\alpha_{1} \ge 1$ and suppose $N_{1}$ is optimal
at $\alpha_{1}$ (i.e., $F_{N_{1}}(\alpha_{1})=\max_{N}F_{N}(\alpha_{1})$).
If, contrary to the claim, every maximizer at $\alpha_{2}$ used strictly
fewer than $N_{1}$ silos, let $N_{2}<N_{1}$ be such that $F_{N_{2}}(\alpha_{2})=\max_{N}F_{N}(\alpha_{2})$.
Consider 
\[
D(\alpha):=F_{N_{1}}(\alpha)-F_{N_{2}}(\alpha).
\]
This is a linear function of $\alpha$ with slope $B_{N_{1}}-B_{N_{2}}>0$
by the lemma. Since $D(\alpha_{2})\le0$ and $D$ has strictly positive
slope, it follows that $D(\alpha_{1})<D(\alpha_{2})\le0$, hence $F_{N_{1}}(\alpha_{1})<F_{N_{2}}(\alpha_{1})$,
contradicting the optimality of $N_{1}$ at $\alpha_{1}$. Therefore
there exists an optimal organization at $\alpha_{2}$ with at least
$N_{1}$ silos.

\medskip{}
\noindent\textbf{(3) If $\alpha>2$, the unique optimum is all singletons.}
Let $s\ge2$ and compare one silo of size $s$ to splitting it into
sizes $1$ and $s-1$. The gain is 
\begin{align*}
\Delta(s):= & f_{\alpha}(1)+f_{\alpha}(s-1)-f_{\alpha}(s)\\
= & [1+\alpha]+[2(s-1)-1+\alpha(2-\frac{1}{s-1})]-[(2s-1)+\alpha(2-\frac{1}{s})]\\
= & \alpha-1-\frac{\alpha}{s(s-1)}
\end{align*}
Since $s(s-1)\ge2$, we have $\Delta(s)\ge\frac{\alpha}{2}-1>0$ whenever
$\alpha>2$. Thus any organization containing a silo of size at least
$2$ can be strictly improved by splitting off a singleton. Iterating,
the only organization to which no such improvement applies has $K$
silos of size $1$, i.e., $N=K$. Because each improvement is strict,
every optimal organization involves all singletons. 

\medskip{}
\noindent\textbf{(4) If $1\le\alpha<2$, optimal silos have size $2$ except
possibly one silo with size 1 or 3.} We progressively rule out the
other possibilities. 

\smallskip{}
\noindent\emph{Step 1: No silo of size $\ge4$.} Fix $s\ge4$ and compare one
silo of size $s$ to two silos of sizes $2$ and $s-2$: 
\[
\Delta_{4}(s):=f_{\alpha}(2)+f_{\alpha}(s-2)-f_{\alpha}(s)=-1+\alpha\left(\frac{3}{2}-\frac{2}{s(s-2)}\right).
\]
For $s\ge4$, $\frac{2}{s(s-2)}\le\frac{1}{4}$, hence 
\[
\Delta_{4}(s)\ge-1+\alpha\left(\frac{5}{4}\right)\ge-1+\frac{5}{4}=\frac{1}{4}>0\qquad\text{for all }\alpha\ge1.
\]
So no optimal organization contains any silos of size $4$ or larger
when $\alpha\ge1$. Thus all optimal silo sizes lie in $\{1,2,3\}$.

\smallskip{}
\noindent\emph{Step 2: At most one silo of size $1$ when $\alpha<2$.} Compare
two singleton silos to one silo of size $2$: 
\[
f_{\alpha}(2)-2f_{\alpha}(1)=1-\frac{\alpha}{2}.
\]
This is strictly positive for $\alpha<2$. Hence no optimum for $1\le\alpha<2$
contains two size-$1$ silos.

\smallskip{}
\noindent\emph{Step 3: No optimum contains both a size-$1$ and a size-$3$
silo.} Compare silos of sizes $(1,3)$ to $(2,2)$: 
\[
2f_{\alpha}(2)-\bigl(f_{\alpha}(1)+f_{\alpha}(3)\bigr)=\frac{\alpha}{3}>0.
\]
Thus any profile containing both sizes $1$ and $3$ can be strictly
improved.

\smallskip{}
\noindent\emph{Step 4: At most one size-$3$ silo when $\alpha\ge1$.} Compare
$(3,3)$ to $(2,2,2)$: 
\[
3f_{\alpha}(2)-2f_{\alpha}(3)=-1+\frac{7}{6}\alpha\ge-1+\frac{7}{6}=\frac{1}{6}>0\qquad\text{for all }\alpha\ge1.
\]
Hence no optimum contains two size-$3$ silos when $\alpha\ge1$.

\smallskip{}
\noindent Combining Steps 1--4, any optimal organization for $1\le\alpha<2$
consists only of size-$2$ silos, plus possibly one odd-sized silo,
which must be of size $1$ or $3$ (and not both).

\noindent If $K$ is even, the sum of silo sizes is even, so it is
impossible to have exactly one odd-sized silo. Hence all silos must
have size $2$, and the only feasible way to sum to $K$ is $N=K/2$
with every silo of size $2$. If $K$ is odd, feasibility requires
at least one odd-sized silo; since the only possible odd sizes are
$1$ and $3$ and we cannot have both, every optimal organization
has exactly one silo of size $1$ \emph{or} exactly one silo of size
$3$, and all remaining silos have size $2$.
\end{proof}

\subsection{Proof of Proposition \ref{prop:summing_welfare}}

We work with the general decision problem, as discussed at the end of Section \ref{sec:Welfare}. It is without loss to normalize the expected utility of an agent who perfectly knows the state to be zero. As in \cite{ali2018role}, 
we may assume without loss of generality that $a\mapsto u(a,0)$ is strictly decreasing and $a\mapsto u(a,1)$ is strictly increasing. Since $A$ is compact, we can denote the lowest and highest actions in $A$ as $\underline{a},\bar{a}\in A.$

First, we establish a lemma that tells us when aggregative efficiency
is in $(0,1),$ Bayesian beliefs converge to the true state 
at an exponential rate.
\begin{lem}
\label{lem:exp_convergence}Suppose the aggregative efficiency of
a network is $AE\in(0,1).$ Let $a_i$ be agent $i$'s posterior belief. For every $\epsilon>0$ so that $0<AE-\epsilon<AE+\epsilon<1,$

\[1-\mathbb{P}[1-e^{-i(AE-\epsilon)\cdot(2/\sigma^{2})}\le a_{i}\le1-e^{-i(AE+\epsilon)\cdot(2/\sigma^{2})}\mid\omega=1]= O(e^{-i})\]
and 
\[1-\mathbb{P}[e^{-i(AE+\epsilon)\cdot(2/\sigma^{2})}\le a_{i}\le e^{-i(AE-\epsilon)\cdot(2/\sigma^{2})}\mid\omega=0]= O(e^{-i}).\]
\end{lem}
\begin{proof}
Conditional on $\omega=1,$ $\ell_i=r_{i}\cdot\frac{2}{\sigma^{2}}+z_{i}$
where $z_{i}\sim\mathcal{N}(0,r_{i}\cdot\frac{4}{\sigma^{2}}).$ Since
$a_{i}=\frac{\exp(\ell_{i})}{1+\exp(\ell_{i})},$ we get
\[
a_{i}=1-\frac{1}{1+\exp(i\cdot(AE+[(r_{i}/i)-AE)])\cdot(2/\sigma^{2}))\cdot\exp(z_{i})}.
\]
 From this, we have 
\begin{align*}
a_{i} & \ge1-\frac{1}{\exp(i\cdot(AE+[(r_{i}/i)-AE)])\cdot(2/\sigma^{2}))\cdot\exp(z_{i})}\\
 & \ge1-\frac{1}{\exp(i\cdot(AE-(\epsilon/2))\cdot(2/\sigma^{2}))\cdot\exp(z_{i})}\text{ for }i\text{ large}\\
 & =1-\frac{1}{\exp(i\cdot(AE-\epsilon)\cdot(2/\sigma^{2}))\cdot\exp(z_{i}+i\cdot(\epsilon/2)\cdot(2/\sigma^{2}))}.
\end{align*}
 So for large $i$, $\mathbb{P}[a_{i}\ge1-e^{-i(AE-\epsilon)\cdot(2/\sigma^{2})}]\ge\mathbb{P}[z_{i}+i\cdot(\epsilon/2)\cdot(2/\sigma^{2})\ge0]$.
But the mean of $z_{i}+i\cdot(\epsilon/2)\cdot(2/\sigma^{2})$ grows
linearly in $i$ and the standard deviation grows at most at the rate of $\sqrt{i},$ and it is well known that the complement of the Gaussian distribution function $\Phi(x)$ converges to 0 at the rate of $e^{x^2}$. This shows $\mathbb{P}[z_{i}+i\cdot(\epsilon/2)\cdot(2/\sigma^{2})<0]=O(e^{-i}).$

For the other direction, we have 
\begin{align*}
a_{i} & \le1-\frac{1}{1+\exp(i\cdot(AE+\epsilon/3)\cdot(2/\sigma^{2}))\cdot\exp(z_{i})}\text{ for }i\text{ large}\\
 & =1-\frac{1}{1+\exp(i\cdot(AE+2\epsilon/3)\cdot(2/\sigma^{2}))\cdot\exp(-i\cdot(\epsilon/3)\cdot(2/\sigma^{2})+z_{i})}
\end{align*}
For large $i,$ $\mathbb{P}[a_{i}\le1-\frac{1}{1+\exp(i\cdot(AE+2\epsilon/3)\cdot(2/\sigma^{2}))}]\ge\mathbb{P}[-i\cdot(\epsilon/3)\cdot(2/\sigma^{2})+z_{i}\le0].$
But the mean of $-i\cdot(\epsilon/3)\cdot(2/\sigma^{2})+z_{i}$ decreases
linearly in $i$ and the standard deviation grows at most at the rate of $\sqrt{i},$
so by the same reasoning as before $\mathbb{P}[-i\cdot(\epsilon/3)\cdot(2/\sigma^{2})+z_{i}>0]=O(e^{-i}).$
Finally, for large $i$, $\frac{1}{1+\exp(i\cdot(AE+2\epsilon/3)\cdot(2/\sigma^{2}))}>\frac{1}{\exp(i\cdot(AE+\epsilon)\cdot(2/\sigma^{2}))}$,
so in fact $\mathbb{P}[a_{i}>1-e^{-i(AE+\epsilon)\cdot(2/\sigma^{2})})]=O(e^{-i})$
as well.

The claim for $\omega=0$ is symmetric.
\end{proof}
Lemma \ref{lem:exp_convergence} implies that for any signal variance,
the undiscounted infinite sums of the expected utilities $\sum_{i}v_i^M,\sum_{i}v_i^{M'}>-\infty$
are convergent.

\begin{lem} \label{lem:convergent_utility}
For any network $M$ where aggregative efficiency is strictly positive,
$\sum_{i=1}^{\infty}v_{i}^{M}$ is convergent.
\end{lem}
\begin{proof}
We know that for every $i,$ $\frac{1}{2}u(\underline{a},1)+\frac{1}{2}u(\bar{a},0)\le v_{i}^{M}\le0$,
so each $v_{i}^{M}$ is bounded. We obtain a lower bound on $v_{i}^{M}$
by supposing that $i$ chooses the action $\bar{a}$ if their Bayesian
posterior belief is 0.5 or higher and chooses the action $\underline{a}$
if their Bayesian posterior belief is strictly lower than 0.5. From
Lemma \ref{lem:exp_convergence}, the probability that $i$'s posterior belief is strictly
lower than 0.5 when $\omega=1$ converges to 0 at an exponential rate
in $i$, and similarly for the probability that $i$'s posterior belief
is 0.5 or higher when $\omega=0$. So, $v_{i}^{M}\to0$ at an exponential
rate in $i$ also. 
\end{proof}

Next, we establish a monotonicity result that says when the signal-counting interpretation holds, the conditional mean of log-belief in the $\omega=1$ state is monotonic with respect to the agent's expected utility. 

\begin{lem}
\label{lem:monotonic_ri} Suppose $i$ and $j$'s log-beliefs have
the following conditional distributions (given the two states): $\ell_{i}\sim\mathcal{N}(\pm r_{i}\cdot\frac{2}{\sigma_{\alpha}^{2}},r_{i}\cdot\frac{4}{\sigma_{\alpha}^{2}})$
and $\ell_{j}\sim\mathcal{N}(\pm r_{j}\cdot\frac{2}{\sigma_{\beta}^{2}},r_{j}\cdot\frac{4}{\sigma_{\beta}^{2}})$.
If $r_{i}\cdot\frac{2}{\sigma_{\alpha}^{2}}>r_{j}\cdot\frac{2}{\sigma_{\beta}^{2}}$,
then $i$'s expected utility in the decision problem is strictly higher
than $j$'s expected utility in the same decision problem.
\end{lem}
\begin{proof}
This is clearly true if $r_{j}=0,$ so suppose $r_{j}>0.$ Let $c=\frac{r_{j}}{r_{i}}\cdot\frac{\sigma_{\alpha}^{2}}{\sigma_{\beta}^{2}}$.
with $0<c<1$, and let $\varphi:=c\cdot\ell_{i}$. Since there is a
one-to-one map between $\varphi$ and $i$'s belief, $i$ would have
the same expected utility if they only observed $\varphi$ before
choosing an action from $A.$ We have the conditional distributions
$\varphi\sim\mathcal{N}(\pm r_{j}\cdot\frac{2}{\sigma_{\beta}^{2}},c\cdot r_{j}\cdot\frac{4}{\sigma_{\beta}^{2}})$
in the two states. Let $\zeta$ be an independent random variable
with the distribution $\zeta\sim\mathcal{N}(0,(1-c)\cdot r_{j}\cdot\frac{4}{\sigma_{\beta}^{2}})$.
Since $(\varphi+\zeta)\sim\mathcal{N}(\pm r_{j}\cdot\frac{2}{\sigma_{\beta}^{2}},r_{j}\cdot\frac{4}{\sigma_{\beta}^{2}})$,
which is the same as the conditional distributions of $j$'s log-belief
in the two states, by the same argument we know that $j$ would have
the same expected utility if they only observed $\varphi+\zeta$ before
choosing an action from $A$. 

Every realization of $\varphi$ induces a belief $p(\varphi)$ in
$(0,1)$ for $i.$ Since $\ell_{i}$ has full-support on $\mathbb{R}$,
the distribution of beliefs induced by $\varphi$ also has full-support
on $[0,1].$ Conditional on every such belief, $j$ has weakly lower
expected utility than $i$. Now we show $j$ has strictly lower expected
utility conditional on a positive measure of beliefs. Given every
$p\in[0,1],$ $\text{argmax}_{a\in A}\{p\cdot u(a,1)+(1-p)\cdot u(a,0)\}$
is non-empty and compact since $u$ is continuous in its first argument.
Let the lowest and highest maximizers be denoted $\underline{a}^{*}(p)$
and $\bar{a}^{*}(p).$ We must have $(\underline{a}^{*}(p),\bar{a}^{*}(p))\ne(\underline{a},\bar{a})$
for a positive measure of $p$, since no action is weakly dominated.
Let $p$ be such a belief and without loss suppose $\underline{a}^{*}(p)>\underline{a}$.
Then there exists $\epsilon>0$ such that for every belief in $(0,\epsilon),$
no action $a\ge\underline{a}^{*}(p)$ is optimal. When $j$ has such
a belief, they must choose an action that is strictly suboptimal from
the perspective of someone who knows the probability of $\omega=1$
is $p$. But since the noise term $\zeta$ has full-support on $\mathbb{R},$
there is positive probability that $j$'s belief is in $(0,\epsilon)$
when $i$'s belief is $p$. This then shows that $j$ has strictly
lower expected utility than $i$ when $i$ has the belief $p,$ and
this holds for a strictly positive measure of $i$'s beliefs. Hence,
$j$ has strictly lower expected utility than $i$ overall. 
\end{proof}
Finally, we show $\sum_{i}\delta^{i-1}v_i^M>\sum_{i}\delta^{i-1}v_i^{M'}$
under the hypotheses of Proposition \ref{prop:summing_welfare}.
\begin{proof}
Find $\epsilon>0$ small enough so that $AE_M-\epsilon>(AE_{M'}+\epsilon)\cdot(1+\epsilon).$
There is some $N$ so that each agent $i\ge N$ aggregates at least
$(AE_M-\epsilon)\cdot i$ signals in network $M$ and no more than
$(AE_{M'}+\epsilon)\cdot i$ signals in network $M'$. Let $C$ be the
difference in expected utility between getting $AE_M-\epsilon$
signals of variance 1 and $(AE_{M'}+\epsilon)\cdot(1+\epsilon)$ signals
of variance 1. Find $\underline{\hat{\sigma}}^{2}>0$ large enough
so that whenever $\sigma^{2}\ge\underline{\hat{\sigma}}^{2},$ all
agents before $N$ in both networks will get utility so close to the expected utility under no information
that the undiscounted difference between the sums of utilities for
the first $N-1$ agents across the two networks is strictly smaller
than $C/2$.

Let $\underline{\sigma}^{2}=\max(\underline{\hat{\sigma}}^{2},1/\epsilon,N).$
We will show that there exists a function $\underline{\delta}:[\underline{\sigma}^{2},\infty)\to(0,1)$
such that $\sum_{i}\delta^{i-1}v_i^M>\sum_{i}\delta^{i-1}v_i^{M'}$
whenever $\delta\ge\underline{\delta}(\sigma^{2})$. Let $\underline{\delta}(\sigma^{2})$
be  $(1/2)^{1/\lceil \sigma^2 \rceil}$ so that $\underline{\delta}(\sigma^{2})^{\left\lceil \sigma^{2}\right\rceil }C=C/2.$
Whenever $\sigma^{2}\ge\underline{\sigma}^{2}$ and $\delta\ge\underline{\delta}(\sigma^{2})$,
first note $v_i^M>v_i^{M'}$ for every $i\ge N,$ and that $\sum_{i=1}^{N-1}\delta^{i-1}v_i^M-\sum_{i}^{N-1}\delta^{i-1}v_i^{M'}\ge-C/2$
since $\sigma^{2}\ge\underline{\sigma}^{2}\ge\underline{\hat{\sigma}}^{2}.$
Because both $\sum_{i}\delta^{i-1}v_i^M$ and $\sum_{i}\delta^{i-1}v_i^{M'}$
are convergent, it suffices to identify one agent $i^{*}\ge N$ so
that $\delta^{i^{*}-1}(v_{i^{*}}^{M}-v_{i^{*}}^{M'})\ge C/2.$ Consider
the agent $i^{*}=\left\lceil \sigma^{2}\right\rceil $, where $i^{*}\ge N$
since  $\sigma^{2}\ge\underline{\sigma}^{2}\ge N.$ In network $M$, this agent
has more than $\left\lceil \sigma^{2}\right\rceil \cdot(AE_M-\epsilon)$
signals with variance $\sigma^{2},$ so $v_{i^{*}}^{M}$ is higher
than the expected utility of $AE_M-\epsilon$ signals of variance
1 (by Lemma \ref{lem:monotonic_ri}). At the same time, $\left\lceil \sigma^{2}\right\rceil \cdot\text{\ensuremath{\frac{1}{\sigma^{2}}\le(1+\epsilon)} since }$
$\sigma^{2}\ge\underline{\sigma}^{2}\ge1/\epsilon,$ so $v_{i^{*}}^{M'}$
is lower than the expected utility of $(AE_{M'}+\epsilon)\cdot(1+\epsilon)$
signals of variance 1 (again using Lemma \ref{lem:monotonic_ri}). This shows $(v_{i^{*}}^{M}-v_{i^{*}}^{M'})\ge C$,
and we know $\delta^{i^{*}-1}(v_{i^{*}}^{M}-v_{i^{*}}^{M'})\ge\underline{\delta}(\sigma^{2})^{\left\lceil \sigma^{2}\right\rceil }\cdot C=C/2.$
\end{proof}

\subsection{Details on Example \ref{ex:welfare_loss}}

We want to show that $$\lim_{K \rightarrow \infty} \frac{\sum_{i}v_{i}^{M}}{\sum_{i}v_{i}^{M'}} = \infty.$$

For each integer $j>0$, let $v^{(j)}$ be the expected utility of an agent observing $j$ independent signals with precision $\sigma_0^2$.

On network $M$, for each agent in generation $t$ we have $r_i \leq K + 3(t-2)$ for all $i$ by Proposition \ref{prop:starting_3}. Since $K + 3(t-2) \leq 2K$ whenever $t \leq K/3 + 2$, all agents in the first $\lfloor K/3+2 \rfloor$ generations have expected utility at most $v^{(2)}$.  So
$$\sum_{i}v_{i}^{M} \le K (  K/3+1) v^{(2)}  < 0.$$

On the complete network $M'$, each agent in generation $t$ has $r_i \geq (t-1)K$, so $$0 > \sum_{i}v_{i}^{M'} > K \sum_{j=0}^{\infty} v^{(j)}.$$
By Lemma \ref{lem:convergent_utility}, the sum $\sum_{j=0}^{\infty} v^{(j)}$ is convergent. So we have
$$0 > \sum_{i}v_{i}^{M'} > -C'K$$
for some constant $C' > 0$ (which is independent of $K$).

Combining these bounds, we have
$$\frac{\sum_{i}v_{i}^{M}}{\sum_{i}v_{i}^{M'}} > \frac{K( K/3+1) v^{(2)} }{-C'K}.$$
The right-hand side diverges to infinity as $K \rightarrow \infty$.

\subsection{Proof of Corollary \ref{cor:mentor}}
\begin{proof}
We claim that for any agent $i$ in generation $t$, the action $\ell_{i}$
is equal to the sum of $\lambda_{i}$ and $\lambda_{j}$ for all agents
$j$ in generations $1,\ldots,t-1$. The proof is by induction on
$t$. The claim holds for the first generation because all agents
in the first generation choose $\ell_{i}=\lambda_{i}$.

Consider an agent in generation $t$. By the inductive hypothesis,
she observes neighbors' actions $\ell_{j}=\lambda_{j}+\sum_{j'\leq(t-2)K}\lambda_{j'}$
for all $j$ in generation $t-1$ and observes $s_{j}$ for one such
$j$. Therefore, she can compute $\sum_{j'\leq(t-2)K}\lambda_{j'}$
and $\lambda_{j}$ for all $j$ in generation $t-1$. Since these
signals are independent and she has access to no information about
other signals from her generation, she chooses $\ell_{i}=\lambda_{i}+\sum_{j\leq(t-1)K}\lambda_{j}.$
By induction, we have $r_{i}=K(t-1)+1>i-K$ for all agents in generation
$t$.
\end{proof}

\subsection{Proof of Corollary \ref{cor:silo}}

Each information silo is equivalent to a maximal generations network,
so the expression for $r_{i}$ for agents in information silos follows
immediately from Theorem \ref{thm:efficiency_dc}.

The actions of agents in separate information silos are conditionally
independent. For an agent in position $t(K+1)$, we have $\frac{r_{t(K+1)}}{t}\geq\sum_{n=1}^{N}2 - \frac{1}{|S_{n}|}$
for $t$ large, because that agent observes conditionally independent
actions of  agents $k_n$ in each silo $n$ with $\lim_{t \to \infty}\frac{r_{(t-1)(K+1)+k_n}}{t}=2 - \frac{1}{|S_{n}|}$
for $1\leq n\leq N$. On the other hand, an upper bound on $\frac{r_{t(K+1)}}{t}$ is
$\frac{1}{t}+\sum_{n=1}^{N}\frac{r_{t(K+1)+k_{n}}}{t}$. This is because agent in position $k_{n}$ in generation
$t$ observes all of the generation $t-1$ agents in silo $S_{n}$,
just as the position $K+1$ agent of generation $t$ does. But $\lim_{t\to\infty}\frac{r_{t(K+1)+k_{n}}}{t}\le\lim_{t\to\infty}\frac{r_{(t-1)(K+1)+k_{n}}+|S_{n}|}{t}=2-\frac{1}{|S_{n}|}$.

\section{Computational Aspects of Calculating $r_{i}$'s and Application to
Collegiate Facebook Networks}\label{a:computation}

In this section, we first provide an algorithm for computing $r_{i}$'s
for all agents in a finite network. The number of operations required
by this algorithm grows at the rate of $nd^{3}+n^{2}d$, where $n$
is the total number of agents and $d$ is an upper bound on the number
of observations of each agent. This is especially efficient in sparse
networks where agents have far fewer neighbors than the size of the
network. 

Next, we apply our framework of measuring the efficiency of information
aggregation to simulate rational sequential social learning in a collection
of 100 Facebook friendship networks for different universities in the United States. We are not aware of applications of models of rational social learning to larger empirical networks. Exploiting
the computational tractability of our setting with binary state and
Gaussian signals, we calculate the exact $r_{i}$'s for all agents in the collegiate Facebook networks. This exercise requires randomly selecting an order for agents' actions, though we find this order has little impact on average learning outcomes.

\subsection{Efficiently Computing $r_{i}$'s in Sparse Social Networks }

The key idea is to maintain in memory the conditional covariance matrix
for agents' log-actions and to update this matrix as we compute $r_{i}$
for each new agent. 
\begin{lyxalgorithm}
\label{alg:iterative_cov} 

Maintain in memory $r_{i}$ for $1\le i\le n$ and an $n\times n$
conditional covariance matrix $C$.

Initialize $r_{1}=1$ and $C_{1,1}=1$.

For $2\le i\le n$:
\begin{enumerate}
\item Let $d_{i}=|N(i)|$. 
\item Let $\Sigma_{i}=(C_{j,k})_{j,k\in N(i)}$ be the $d_i \times d_i$ conditional
covariance matrix of $i$'s neighbors' log-actions. 
\item Let $r_{N(i)}=(r_{j})_{j\in N(i)}$ be a vector containing the $r_{j}$'s
for $i$'s neighbors. 
\item Compute $\Sigma_{i}^{-1}$ and compute $\beta_{i}=r_{N(i)}\cdot\Sigma_{i}^{-1}$,
which is $i$'s rational weights on neighbors' log-actions. 
\item Compute $r_{i}=1+\sum_{j\in N(i)}\beta_{i,j}r_{j}$. 
\item Update $C$: 
\begin{itemize}
\item Let $C_{i,i}=r_{i}.$ 
\item For each $k<i,$ let $C_{i,k}=\sum_{j\in N(i)}\beta_{i,j}C_{j,k}$
and also symmetrically set $C_{k,i}=C_{i,k}$. 
\end{itemize}
\end{enumerate}
\end{lyxalgorithm}
\begin{prop} \label{prop:algorithm}
Algorithm \ref{alg:iterative_cov} computes $(r_{i})_{i=1}^{n}$ and
its time complexity is $O(nd^{3}+n^{2}d)$ if $|N(i)|\le d$ for every agent $i$.  
\end{prop}
In particular, in networks where there is a universal bound $d$ on
the number of neighbors observed by any agent, Algorithm \ref{alg:iterative_cov}
runs in $O(n^{2})$ time in the number of agents $n$ in the network. 

\begin{proof}
It is without loss to assume $\sigma^{2}=4$, since $r_{i}$'s do
not depend on signal precision. Then for every $i,$ $r_{i}=2\mathbb{E}[\ell_{i}\mid\omega=1]=\text{Var}[\ell_{i}\mid\omega=1]$.
For the initialization, since $i=1$ has no social information, clearly
$r_{1}=1$ and so $C_{1,1}=1$ also under the assumption $\sigma^{2}=4.$
Assume by induction that $r_{i}$ and $C_{j,k}$ are correctly computed
for all $i,j,k\le I$. The formula from Proposition  \ref{prop:linear}
implies that $\beta_{I+1}=r_{N(I+1)}\cdot\Sigma_{I+1}^{-1}$ are the
rational weights for $I+1.$ We also know 
\begin{align*}
r_{I+1} & =2\mathbb{E}[\ell_{I+1}\mid\omega=1]\\
 & =2\mathbb{E}[\lambda_{I+1}+\sum_{j\in N(I+1)}\beta_{I+1,j}\ell_{j}\mid\omega=1]\\
 & =2\cdot[\frac{1}{2}+\sum_{j\in N(I+1)}\beta_{I+1,j}\cdot\frac{1}{2}r_{j}]\\
 & =1+\sum_{j\in N(I+1)}\beta_{I+1,j}r_{j}.
\end{align*}
Finally, $C_{I+1,I+1}=\text{Var}[\ell_{I+1}\mid\omega=1]=r_{I+1}$
when $\sigma^{2}=4$. For $k<I+1,$ we have $\text{Cov}[\ell_{I+1},\ell_{k}\mid\omega=1]=\text{Cov[}\sum_{j\in N(I+1)}\beta_{I+1,j}\ell_{j},\ell_{k}\mid\omega=1]=\sum_{j\in N(I+1)}\beta_{I+1,j}C_{j,k}$,
where we have used the fact that $\lambda_{I+1}$ is conditionally
independent of $\ell_{k}$. 

To analyze this algorithm's time complexity, we note that the matrix
inverse operation $\Sigma_{i}^{-1}$ requires $O(|N(i)|^{3})$ operations
using Gaussian elimination, and $|N(i)|^{3}\le d^{3}.$ Computing
$\beta_{i}$ and $r_{i}$ require a lower-order number of operations
than $O(d^{3}).$ Finally, updating each entry $C_{i,k}$ for $k<i$
is an order $O(|N(i)|)$ operation (since the $C_{j,k}$ terms are
already stored in memory), and no more than $n$ such conditional
covariance entries need to be computed. This shows each step of the
for loop requires $O(d^{3}+nd)$ operations. Since there are $n-1$
total steps in the for loop, we get $O(nd^{3}+n^{2}d)$ as the complexity
of the algorithm. 
\end{proof}

\subsection{Application to Collegiate Facebook Networks}
\label{subsec:facebook}
This dataset was first used by \cite*{traud2012social} and
was provided directly by the former chief technology officer of Facebook.
It contains all the Facebook users on the first 100 university
campuses in the US where Facebook was made available, as of September 2005.
The dataset shows information for each user, including their graduation
year and their Facebook friends within the same university. We restrict
attention to undergraduate students with graduation years between
2006 and 2009. Under this restriction, the 100 Facebook networks ranged
in size from 471 (Caltech) to 26748 (Penn State University), with
a mean of 7714 and a standard deviation of 5760. The average degree
of users in these networks ranged from 32 (Michigan Technological
University) to 113 (Harvard University), with a mean of 70 and a standard
deviation of 17. 

For each of the 100 networks, we draw 1000 random orderings of the
students and calculate the number of signals $r_{i}$ aggregated
for each student under sequential rational learning for each ordering. We highlight three main findings, which we describe in more detail below:
\begin{enumerate}
    \item The random ordering of students has little effect on the amount of information loss.
    \item There is substantial information loss in many of the networks.
    \item There is substantial variation in information loss across networks, and much of this variation can be explained by the amount of clustering in the network.
\end{enumerate}

For the measures of learning efficiency we report below, we found
very little variation as we ordered the students in the universities
in different random ways. For every measure, the standard deviation
of the measure across the 1000 orderings was on average no larger than
2\% of its mean across the orderings. Therefore, below we will focus
on the average of each measure across the 1000 random orderings and
study how these measures co-vary with network clustering. 

For each network, we calculate $$\frac{1}{N}\sum_{i=1}^{N}\frac{r_{i}}{i} \text{
and } \frac{1}{N'}\sum_{i:N^{\text{ind}}(i) \neq \emptyset}\frac{r_{i}-1}{|N^{\text{ind}}(i)|},$$
where the indirect neighborhood $N^{\text{ind}}(i)$ is the set of students reachable by a path from $i$ in the observation network and $N'$ is the number of students
whose indirect neighborhood is nonempty. We call the first measure the
\emph{fraction of potential signals aggregated} and the second measure the 
\emph{fraction of reachable signals aggregated}. The fraction
of potential signals aggregated is an analog of aggregative efficiency
in finite networks. The fraction of reachable signals aggregated adjusts
 $r_i$ to account for the size of $i$'s indirect neighborhood. So low fractions of reachable signals aggregated reflect inefficiencies stemming from information confounding, and not merely from agents having access to few signals. 

\begin{table}

\begin{centering}
\begin{tabular}{|c|c|c|c|c|c|c|}
\hline 
\multicolumn{7}{|c|}{Fraction of potential signals aggregated}\tabularnewline
\hline 
Min & 1st Qu. & Median & Mean & 3rd Qu. & Max & Sd.\tabularnewline
\hline 
\hline 
0.2641 & 0.3484 & 0.4839 & 0.4702 & 0.5707 & 0.7464 & 0.1235\tabularnewline
\hline 
\end{tabular}
\par\end{centering}
\begin{centering}
\begin{tabular}{|c|c|c|c|c|c|c|}
\hline 
\multicolumn{7}{|c|}{Fraction of reachable signals aggregated}\tabularnewline
\hline 
Min & 1st Qu. & Median & Mean & 3rd Qu. & Max & Sd.\tabularnewline
\hline 
\hline 
0.4766 & 0.5971 & 0.7306 & 0.7050 & 0.7975 & 0.9267 & 0.1182\tabularnewline
\hline 
\end{tabular}
\par\end{centering}
\caption{Summary statistics of fraction of potential signals aggregated and
fraction of reachable signals aggregated for the 100 collegiate Facebook
networks. }\label{tab:Summary-statistics-of}

\end{table}

Table \ref{tab:Summary-statistics-of} displays summary statistics
of these two measures for the 100 collegiate Facebook networks. The
amount of information loss is sizable and heterogeneous in these networks.
On average, more than half of the potential signals are lost through
social learning. These losses are due to a combination of some students not being reachable
in one's indirect neighborhood and confounding obstructing aggregation of the signals from one's indirect neighborhood.
Using the fraction of reachable signals measure to normalize by one's
indirect neighborhood size, we see that information confounding causes
on average a 30\% loss of signals in these networks. 

Towards explaining some of the variability in these measures of social
learning efficiency across different universities, we compute the global clustering
coefficient for each (unordered) collegiate Facebook network (see,
for example, \cite{jackson2010social}).\footnote{Results are similar if we use a different measure of network clustering, the average of local clustering coefficients (also defined in \cite{jackson2010social}): the two correlations reported below are then 0.65 and 0.79.} This is defined as the fraction of closed triplets in the network:
that is, if $i$ and $j$ are Facebook friends and $j$ and $k$ are
Facebook friends, how likely is it that $i$ and $k$ are also Facebook
friends? Intuitively, higher clustering helps reduce information confounding,
since confounding never happens in a transitive network: if $i$ is
friends with all of their friends' friends, then they can account
for all of their friends' social information. 

\begin{figure}
\begin{centering}
\includegraphics[scale=0.5]{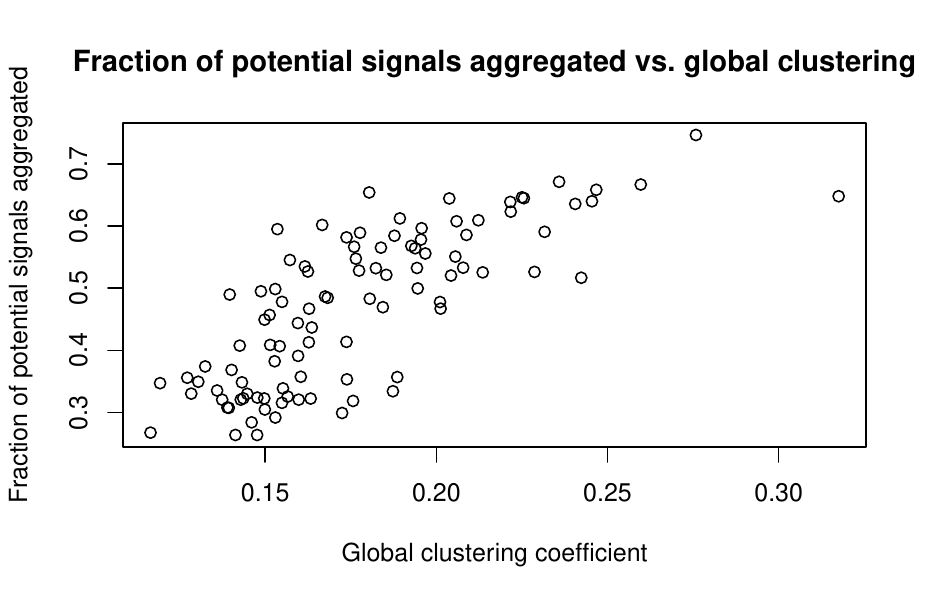}\includegraphics[scale=0.5]{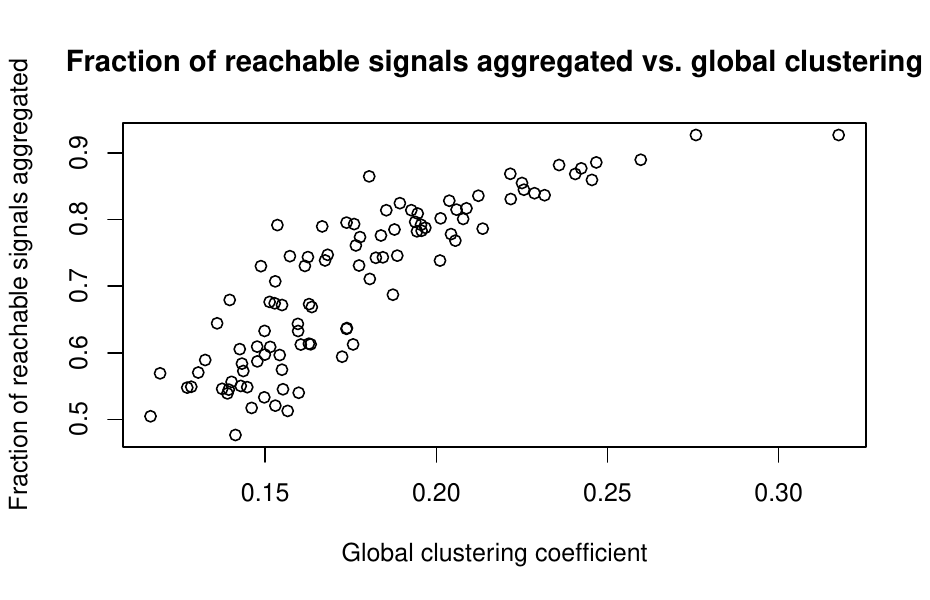}
\par\end{centering}
\caption{Correlation between  global clustering coefficient and two
measures of information aggregation for collegiate Facebook networks. }\label{fig:Correlation-between-between}

\end{figure}

Figure \ref{fig:Correlation-between-between} correlates the global clustering coefficient with the two measures of information
aggregation. We find a strong positive correlation between the global
clustering coefficient and both the fraction of potential signals
aggregated (correlation of 0.77) and the fraction of reachable signals
aggregated (correlation of 0.85). For example, the figure shows the
two outlier institutions in terms of high network clustering (Caltech
and Haverford College) also have among the highest levels of information
aggregation under both measures. The correlation between the clustering
coefficient and the fraction of reachable signals aggregated is stronger
than the correlation between the clustering coefficient and the fraction
of potential signals aggregated. This may be because more closely
knit communities help reduce confounding but also reduce the size
of students' indirect neighborhoods, so the effect on $r_{i}/i$ is
noisier.

\end{document}